\documentclass[12pt,onecolumn,draftclsnofoot,journal]{IEEEtran}
\usepackage{graphicx}
\usepackage{float}
\usepackage{epstopdf}
\usepackage[cmex10]{amsmath}
\usepackage{array}
\usepackage{cite}
\usepackage{amssymb}
\usepackage{amsfonts}
\usepackage{amsmath}
\usepackage{stackrel}
\usepackage{booktabs,multirow}
\usepackage{arydshln}
\usepackage{slashbox}
\usepackage{amsthm}
\usepackage{listings}
\usepackage{algorithm}
\usepackage{algorithmicx}
\usepackage{algpseudocode}
\usepackage{tablefootnote}
\usepackage{threeparttable}

\newtheorem{lem}{Lemma}
\newtheorem{rem}{Remark}
\newtheorem{theo}{Theorem}

\newtheorem{pro}{Proposition}
\newtheorem{ex}{Example}
\newfloat{routine}{htbp}{loa}
\floatname{routine}{Routine} 

\makeatletter
\newcommand{\algmargin}{\the\ALG@thistlm}
\makeatother
\newlength{\forwidth}
\algdef{SE}[parFOR]{parFor}{EndparFor}[1]
  {\parbox[t]{\dimexpr\linewidth-\algmargin}{
     \hangindent\forwidth\strut\algorithmicfor\ #1\ \algorithmicdo\strut}}{\algorithmicend\ \algorithmicfor}
\algnewcommand{\parState}[1]{\State
  \parbox[t]{\dimexpr\linewidth-\algmargin}{\strut #1\strut}}

\newlength{\ifwidth}
\algdef{SE}[parIF]{parIf}{EndparIf}[1]
  {\parbox[t]{\dimexpr\linewidth-\algmargin}{
     \hangindent\ifwidth\strut\algorithmicif\ #1\ \algorithmicdo\strut}}{\algorithmicend\ \algorithmicif}
\begin{document}

\title{New Characterization and Efficient Exhaustive Search Algorithm for Elementary Trapping Sets of Variable-Regular LDPC Codes}
\author{Yoones Hashemi,\IEEEmembership{ Student Member, IEEE}, and Amir H. Banihashemi,\IEEEmembership{ Senior Member, IEEE}}
\maketitle
%\IEEEdisplaynontitleabstractindextext
\IEEEpeerreviewmaketitle

\begin{abstract}
In this paper, we propose a new characterization for elementary trapping sets (ETSs) of variable-regular low-density parity-check (LDPC) codes. 
Recently, Karimi and Banihashemi proposed a characterization of ETSs, which was based on viewing an ETS as a layered superset (LSS) of a short cycle in the code's Tanner graph.
A notable advantage of LSS characterization is that it corresponds to a simple LSS-based search algorithm (expansion technique) that starts from short cycles of the graph and finds
the ETSs with LSS structure efficiently. Compared to the LSS-based characterization of Karimi and Banihashemi, which is based on a single LSS expansion technique,  
the new characterization involves two additional expansion techniques. The introduction of the new techniques mitigates two 
problems that LSS-based characterization/search suffers from:  (1) exhaustiveness: not every ETS structure is an LSS of a cycle, (2) search efficiency: 
LSS-based search algorithm often requires the enumeration of cycles with length much larger than the girth of  the graph, where the multiplicity
of such cycles increases rapidly with their length. We prove that using the three expansion techniques, any ETS structure can be obtained starting from
a simple cycle, no matter how large the size of the structure $a$ or the number of its unsatisfied check nodes $b$ are, i.e., the characterization is exhaustive.
We also demonstrate that for the proposed characterization/search to exhaustively cover all the ETS structures within the $(a,b)$ classes with
$a \leq a_{max}, b \leq b_{max}$, for any value of $a_{max}$ and $b_{max}$, the length of the short cycles required to be enumerated is 
less than that of the LSS-based characterization/search. We, in fact, show that such a length for the proposed search algorithm is minimal.
We also prove that the three expansion techniques, proposed here, are the only expansions needed for characterization of ETS structures starting from 
simple cycles in the graph, if one requires each and every intermediate sub-structure to be an ETS as well.   
Extensive simulation results are provided to show that, compared to LSS-based search, 
significant improvement in search speed and memory requirements can be achieved.     
\end{abstract}

\section{introduction}

Finite-length low-density parity-check (LDPC) codes under iterative decoding algorithms suffer from the \textit{error floor} phenomenon. In the error floor region, 
the slope of error rate curves decreases when the channel quality improves beyond a certain point. It is well-known that error-floor performance of LDPC codes is 
related to the presence of certain problematic graphical structures in the Tanner graph of the code, commonly referred to as \textit{trapping sets} \cite{richardson}.
Graphically, a trapping set (TS) is often seen as the induced subgraph of some variable nodes in the code's Tanner graph.
In such a representation,  the TS is commonly characterized by the number of its variable nodes $a$, and the number
of its odd-degree (unsatisfied) check nodes $b$, and is said to belong to the class of $(a,b)$ trapping sets.
Among TSs, the most harmful ones are known to be the \textit{elementary trapping sets (ETSs)},\cite{zhang2009toward},\cite{laendner2009},\cite{milenkovic2007asymptotic},
whose induced subgraphs contain only degree-1 and degree-2 check nodes.  

The knowledge of trapping sets and their structure has been used in estimating the error floor of LDPC codes~\cite{cole},~\cite{sina1}, in modifying iterative decoders to have better error-floor performance \cite{cavus}, \cite{Han}, \cite{kyung2012finding}, \cite{zhang2013controlling}, and in constructing LDPC codes with low error floors \cite{Ivkovic}, \cite{Asvadi}, \cite{nguyen}, \cite{KAB-2012}, \cite{ABA-2012}. 
There has also been a flurry of activity in characterization of trapping sets  \cite{vasic2009trapping}, \cite{Dolecek}, \cite{Diao}, \cite{dolecek2010analysis}, \cite{schlegel2010dynamics}, \cite{laendner2010characterization}, \cite{nguyen}, \cite{mehdi2014}, \cite{yoones2015}, 
and in developing search algorithms to find them \cite{Rosnes},  \cite{wang2009finding}, \cite{abu2010trapping}, \cite{Wang}, \cite{mehdi2012}, \cite{kyung2012finding}. 
Most of the works on TSs in the literature are non-exhaustive \cite{abu2010trapping},\cite{mehdi2012}, are limited to structured codes \cite{zhang2013controlling},\cite{schlegel2010dynamics}, are concerned with relatively small TSs \cite{laendner2010characterization}, only deal with a certain variable node degree \cite{vasic2009trapping},\cite{nguyen}, or are applicable to relatively short block lengths \cite{wang2009finding},\cite{kyung2012finding}. 

Recently, Karimi and Banihashemi \cite{mehdi2014} characterized the ETSs of variable-regular LDPC codes. They studied the graphical structure of the ETSs, and 
demonstrated that a vast majority of the non-isomorphic structures of ETSs are layered super sets (LSS) of short cycles, i.e., each structure $\mathcal{S}$ corresponds to a sequence of ETS structures that starts from a short cycle and grows one variable node at a time to reach $\mathcal{S}$.
This characterization corresponds to a simple search algorithm that starts from the instances of short cycles in the Tanner graph and can find the ETSs with the LSS property. 
A shortcoming of this characterization/search, however, is that some of the ETS structures do not lend themselves to the above characterization, i.e., they are not LSSs of any of their cycle subsets. Such structures are labeled as ``NA" in \cite{mehdi2014}.

 In an earlier paper \cite{yoones2015}, we complemented the results of \cite{mehdi2014}. By careful examination of NA structures, we demonstrated that they are all LSSs of some basic graphical structures which are slightly more complex than cycles.  These basic structures were called \textit{prime} with respect to the LSS property, or ``LSS-prime'' in brief, implying that they are not layered super sets of smaller ETSs. 
Results of  \cite{mehdi2014} (and its corrections in \cite{hashemi2015corrections}) along with those of \cite{yoones2015} provided a simple LSS-based search algorithm that can 
find all the instances of  ETSs of variable-regular LDPC codes, for any range of $a$ and $b$ values, in a guaranteed fashion, starting from the instances of LSS-prime structures. 

The LSS-based search algorithm, however, suffers from a major drawback. Although, the search algorithm itself is simple, to find the dominant ETSs in a guaranteed fashion, one may need to enumerate 
cycles of length much larger than the girth in the Tanner graph of the code, as the input to the search algorithm. The multiplicity of such cycles increases rapidly with the cycle length and thus 
the enumeration, storage and processing of these cycles pose a practical hurdle in implementing the search algorithm. The same issue applies to non-cycle LSS-prime structures. This becomes particularly
problematic for larger variable and check degrees and cases where ETSs with larger $a$ and $b$ values are of interest. In fact, our experimental results show that by increasing the variable and average check degrees, the range of $a$ and $b$ values, the code rate and the block length, the LSS-based search algorithm becomes quickly infeasible to implement.  
  
To the best of our knowledge, the LSS-based search algorithm is the most efficient exhaustive search algorithm for ETSs of a general variable-regular LDPC code.
Yet, the complexity of the algorithm is still too high for many practical codes. Motivated by this, in this paper, we develop a new characterization
of ETS structures of variable-regular LDPC codes. The new characterization is based on a hierarchical graph-based expansion approach which involves 
LSS expansion plus two additional expansion techniques. It characterizes each ETS structure with an embedded sequence of ETS structures
that starts from a simple cycle and expands, at each step, to a larger ETS using one of the three expansions, until it reaches the ETS structure of interest. 
The characterization is {\em minimal}, in the sense that, none of the expansion steps can be divided into smaller expansions such that the resulting new sub-structures are 
still ETSs. The new characterization corresponds to an efficient exhaustive search algorithm for ETSs 
that requires only short-length simple cycles of the graph as the input. The maximum length of the input cycles is often smaller than that of the LSS-based search algorithm, and is, in fact, provably minimal.  
We develop a general approach to characterize all the non-isomorphic structures of ETSs within a certain range of $a \leq a_{max}, b \leq b_{max}$, based on the three expansions, 
for LDPC codes with different variable degrees and girths. We then provide ETS characterization tables for LDPC codes with different variable degrees, girths, 
and with different rate and block length ranges. These tables are used to devise efficient search algorithms, based on the three expansion techniques, 
for exhaustively finding ETSs. The search algorithms are then applied to a large number of LDPC codes, both random and structured, and with a variety of degree distributions, 
girths, rates and block lengths, to demonstrate the strength and versatility of the proposed scheme. Compared to LSS-based search, improvements of up to three orders of magnitude 
are achieved in the run-time and memory requirements.     

The remainder of this paper is organized as follows. Basic definitions and notations are provided in Section \ref{sec:pre}. In this section, we also 
describe a software that is used to find non-isomorphic structures of ETSs in a given class. In Section \ref{sec:dot-based}, the LSS-based characterization/search is revisited 
and its shortcomings are discussed. In Section \ref{sec:dpl-based}, the new expansion techniques are introduced, the new characterization is developed,
and characterization tables are provided. We also present the new search algorithm, and discuss its complexity in this section. 
 Numerical results are presented in Section \ref{sec:numerical}, and Section \ref{conclude} is devoted to conclusions.

\section{Preliminaries }
\label{sec:pre}
\subsection{Definitions and Notations}
Consider an undirected graph $G=(F, E)$, where the two sets $F=\{f_1,\dots,f_k\}$ and $E=\{e_1,\dots,e_m\}$, are the sets of \textit{nodes} and \textit{edges} of $G$,
respectively. We say that an edge $e$ is \textit{incident} to a node $f$ if $e$ is connected to $f$.  If there exists an edge $e_k$ which is incident to 
two distinct nodes $f_i$ and $f_j$, we call these nodes \textit{adjacent}. In this case, we represent $e_k$ by $f_i f_j$ or $f_j f_i$. 
The \textit{neighborhood} of a node $f$, denoted by $\mathcal{N}(f)$, is the set of nodes adjacent to $f$. 
The degree of a node $f$ is denoted by \textit{deg(f)}, and is defined as the number of edges incident to $f$. 
The \textit{maximum degree} and the \textit{minimum degree} of a graph $G$,  denoted by $\Delta (G)$ and $\delta (G)$, respectively, are defined to be the maximum and minimum 
degree of its nodes, respectively. 

Given an undirected graph $G=(F,E)$,
 a \textit{walk} between two nodes $f_1$ and $f_{k+1}$ is a sequence of nodes and edges
 $f_1$, $e_1$, $f_2$, $e_2$, $\dots$, $f_k$, $e_k$, $f_{k+1}$, where $e_i=f_i f_{i+1}$, $\forall i \in [1,k]$. 
In this definition, the nodes $f_1,f_2,\dots,f_{k+1}$ are not necessarily distinct. The same applies to the edges $e_1,e_2,\dots,e_k$.
 A {\em path} is a walk  with no repeated nodes or edges, except the first and the last nodes that can be the same. If the first and the last nodes are distinct, we call the path an {\em open path}.
Otherwise, we call it a {\em cycle}. The \textit{length} of a walk, a  path, or a cycle is the number of its edges. 
 A {\em lollipop walk} is a walk $f_1$, $e_1$, $f_2$, $e_2$, $\dots$, $f_k$, $e_k$, $f_{k+1}$, such that all the edges and all the nodes are distinct, except that $f_{k+1}=f_m$, for some $m \in (1,k)$.
A \textit{chord} of a cycle is an edge which is not part of the cycle but is incident to two distinct nodes in the cycle.
A \textit{simple cycle} or a {\em chordless cycle} is a cycle which does not have any chord. 
Throughout this paper, we use the notation $s_k$ for a simple cycle of length $k$. Notation $c_k$ is used for cycles of length $k$ that do have at least one chord.

A graph is called {\em connected} when there is a {\em path} between every pair of nodes. A \textit{tree} is a connected graph that contains no cycles. 
A \textit{rooted tree} is a tree in which one specific node is assigned as the \textit{root}. The \textit{depth of a node} in a rooted tree is the length of the
path from the node to the root. The \textit{depth of a tree} is the maximum depth of any node in the tree. 
\textit{Depth-one tree} is a tree with depth one. 
A node $f$ is called \textit{leaf} if $deg(f)=1$. Although this terminology is commonly used for trees, in this paper, we use it for a general graph that may contain cycles. 
A \textit{leafless graph} is a connected graph $G$ with $\delta (G)=2$.

The graphs $G_1 =(F_1,E_1)$ and $G_2 =(F_2,E_2)$ are
\textit{isomorphic} if there is a bijection $p : F_1 \rightarrow F_2$ such that
nodes $f_1, f_2 \in F_1$ are joined by an edge if and only if $p(f_1)$
and $p(f_2)$ are joined by an edge. Otherwise, the graphs are \textit{non-isomorphic}.

Any $m \times n$ parity check matrix $H$ of an $(n,k)$ LDPC code $\mathcal{C}$ can be represented by its bipartite Tanner graph $G=(V \cup C, E)$, 
where $V=\{ v_1,v_2,\dots,v_n \}$ is the set of variable nodes and $C=\{ c_1,c_2,\dots,c_m \}$ is the set of check nodes. 
The degrees of nodes $v_i$ and $c_i$  are denoted by $d_{v_{i}}$ and $d_{\mathrm{c}_{i}}$, respectively. A Tanner graph is 
called {\em variable-regular} with variable degree $d_{\mathrm{v}}$ if $d_{v_i} = d_{\mathrm{v}}$, $\forall~{v}_{i} \in V$. A $(d_{\mathrm{v}},d_{\mathrm{c}})$-regular Tanner graph
is a variable-regular graph in which $d_{c_i} = d_{\mathrm{c}}$, $\forall~{c}_{i} \in C$.
For a subset $\mathcal{S}$ of $V$, the subset $\Gamma{(\mathcal{S})}$ of $C$ denotes the set of neighbors of $\mathcal{S}$ in $G$.
The \textit{induced subgraph} of $\mathcal{S}$ in $G$, denoted by $G(\mathcal{S})$, is the graph with the set of nodes $\mathcal{S} \cup \Gamma{(\mathcal{S})}$ and 
the set of edges $\{f_i f_j \in E : f_i \in \mathcal{S}, f_j \in \Gamma{(\mathcal{S})}\}$. The set of check nodes with odd and even degrees in $G(\mathcal{S})$ are denoted by $\Gamma_{o}{(\mathcal{S})}$ and $\Gamma_{e}{(\mathcal{S})}$, respectively. In this paper, the terms \textit{unsatisfied check nodes} and \textit{satisfied check nodes} are used to refer to the check nodes in
$\Gamma_{o}{(\mathcal{S})}$ and $\Gamma_{e}{(\mathcal{S})}$, respectively. 
The \textit{size} of an induced subgraph $G(\mathcal{S})$ is defined to be the number of its variable nodes.  
We assume that an induced subgraph is connected. Disconnected subgraphs can be considered as the union of connected ones. 
The length of the shortest cycle in a Tanner graph is called \textit{girth}, and is denoted by $g$. 
We study the variable-regular graphs with $d_{\mathrm{v}} \geq 3$, that are free of 4-cycles ($g > 4$). All the induced subgraphs with the same size $a$, 
and the same number of unsatisfied check nodes $b$, are considered to belong to the same \textit{$(a,b)$ class}.

Given a Tanner graph G,
a set $\mathcal{S}\subset V$ is called an \textit{(a,b) trapping set (TS)} if $|\mathcal{S}| = a$ and $|\Gamma_{o}{(\mathcal{S})}| = b$.
Alternatively, $\mathcal{S}$ is said to belong to the {\em class of (a,b) TSs}. Parameter $a$ is referred to as the {\em size} of the TS.
In the rest of the paper, depending on the context, the term ``trapping set'' may be used to refer to the set of variable nodes $\mathcal{S}$, 
or to the induced subgraph $G(\mathcal{S})$ of $\mathcal{S}$ in the Tanner graph $G$. Similarly, we may use ${\cal S}$ to mean $G(\mathcal{S})$.
 An \textit{elementary trapping set (ETS)} is a trapping set for which all the check nodes in $G(\mathcal{S})$ have degree 1 or 2.
 A set $\mathcal{S}\subset V$ is called an \textit{(a,b) absorbing set (AS)} if $\mathcal{S}$ is an $(a,b)$ trapping set and all
 the variable nodes in $\mathcal{S}$  are connected to more nodes in $\Gamma_{e}{(\mathcal{S})}$ than in $\Gamma_{o}{(\mathcal{S})}$.
 An {\em elementary absorbing set (EAS)} $\mathcal{S}$ is an absorbing set for which all the check nodes in $G(\mathcal{S})$ have degree 1 or 2.
 A \textit{fully absorbing set (FAS)} $\mathcal{S}\subset V$ is an absorbing set for which all the nodes in $V \backslash \mathcal{S}$ have strictly more neighbors in $C \backslash \Gamma_{o}{(\mathcal{S})}$ 
than in $\Gamma_{o}{(\mathcal{S})}$.
A set $\mathcal{S}\subset V$ is called an \textit{(a,b) fully elementary absorbing set (FEAS)} if $\mathcal{S}$ is an $(a,b)$
EAS and if all the nodes in $V\backslash \mathcal{S}$ have strictly more neighbors 
 in $C\backslash \Gamma_{o}{(\mathcal{S})}$ than in $\Gamma_{o}{(\mathcal{S})}$.

Elementary trapping sets are the subject of this paper. To simplify the representation of ETSs, similar to \cite{mehdi2014},~\cite{yoones2015}, we use an 
alternate graph representation of ETSs, called \textit{normal graph}.
The normal graph of an ETS $\mathcal{S}$ is obtained from $G(\mathcal{S})$ by removing all the check nodes of degree one and their incident edges, and by replacing all the 
degree-2 check nodes and their two incident edges by a single edge. It is easy to see that there is a one-to-one correspondence between the bipartite graph
$G(\mathcal{S})$ and the normal graph of $\mathcal{S}$. 
\begin{lem}
Consider the normal graph of an $(a,b)$ ETS structure of a variable-regular Tanner graph with variable degree $d_{\mathrm{v}}$.
The number of nodes and edges of this normal graph are $a$ and $(a d_{\mathrm{v}} - b)/2$, respectively. We thus have $b=a d_{\mathrm{v}} - 2 e$,
where $e$ is the number of edges of the normal graph. 
\label{lem1}
\end{lem}

 We call a set $\mathcal{S}\subset V$ an \textit{(a,b) leafless ETS (LETS)} if $\mathcal{S}$ is an $(a,b)$
 ETS and if the normal graph of $\mathcal{S}$ is leafless. 
\begin{ex}
 Fig. \ref{fig:leafless}(a) represents a LETS in the $(4,2)$ class and its leafless normal graph, while Fig. \ref{fig:leafless}(b) represents an ETS in the $(4,4)$ class and its normal graph with a leaf. 
Symbols \scalebox{0.7}{$\square$} and \scalebox{0.7}{$\blacksquare$} are used to represent satisfied and unsatisfied check nodes in the induced bipartite subgraphs, respectively, and the symbol 
\scalebox{1.2}{$\circ$} is used to represent variable nodes in both the induced subgraphs and normal graphs. 
\begin{figure}[] 
\centering
\includegraphics [width=0.38\textwidth]{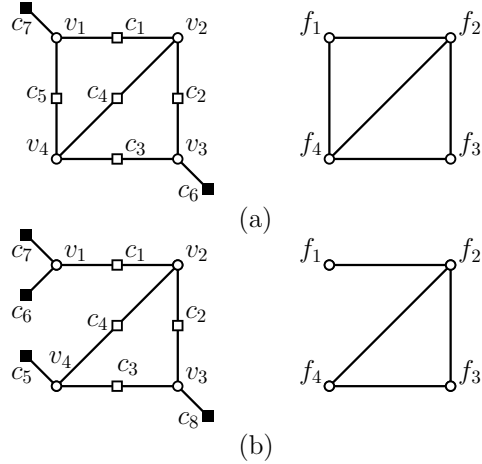}
\caption{(a) A LETS in the $(4,2)$ class and its leafless normal graph, (b) An ETS in the $(4,4)$ class and its normal graph which has a leaf $(f_1)$.}
\label{fig:leafless}
\end{figure}
\end{ex}

In general, EASs are a subset of LETSs. For the special case of $d_{\mathrm{v}}=3$, LETSs and EASs are equal sets. 
To the best of our knowledge, almost all the structures reported in the literature as error-prone structures of variable-regular LDPC codes are LETSs. 
LETSs have thus been the subject of many studies including~\cite{butler2014},\cite{mehdi2012},\cite{mehdi2014},\cite{yoones2015} and \cite{nguyen}.
In the rest of this paper, also, we focus on LETSs.

In this paper, we characterize the non-isomorphic LETS structures, and also devise search algorithms to find all the instances of LETS structures of a given
Tanner graph in a guaranteed fashion. In the part related to characterization, we use the normal graph representation of LETS structures.
In the part on search algorithms, on the other hand, we perform the search on a Tanner graph, and are thus concerned with the bipartite graph 
representation of a LETS structure. Nevertheless, for consistency and to prevent confusion, we still use terminologies corresponding to normal graphs.
For example, we use ``all the instances of $s_k$'' to mean ``all the instances of the structure whose normal graph is $s_k$."

\subsection{Nauty Program}
\label{nau}

\textit{Nauty} \cite{mckay2014practical}, is a package of programs for computing automorphism groups of graphs. There is a program called \textit{geng} in this package 
that generates all the non-isomorphic graphical structures with a given number of nodes and edges very efficiently. 
This program has many input options to generate connected graphs, triangle-free graphs, or graphs with certain minimum and maximum node degrees.
In this work, we use {\em geng} to generate all the non-isomorphic graphical structures of LETSs in different classes. To simplify this task, we 
use the normal graph representation of LETSs. 
Since we are interested in connected LETS structures in variable-regular Tanner graphs with variable degree $d_{\mathrm{v}}$ that are 4-cycle free,
we look for normal graphs with no parallel edges and with minimum and maximum node degrees equal to $\delta (G)=2$ and $\Delta(G) \leq d_{\mathrm{v}}$, respectively. 
By setting the range of node degrees to $[2, d_{\mathrm{v}}]$, the number of nodes to $a$, the number of edges to $(ad_{\mathrm{v}}-b)/2$, and by choosing the simple connected graph 
in {\em geng}'s input options, this program generates all the non-isomorphic graphical structures of connected LETSs in the $(a,b)$ class for variable-regular Tanner graphs with variable degree $d_{\mathrm{v}}$. 
If we are interested in Tanner graphs with $g \geq 8$, we also add the option of ``triangle-free.'' 

\section{ LSS-based (dot-based) Characterization/Search Approach}
\label{sec:dot-based}

In this section, by providing a new perspective, we revisit the LSS-based characterization/search approach, proposed in~\cite{mehdi2014}. The new
perspective is based on the normal graph representation of LETS structures, and a graph expansion technique applied to normal graphs, which we refer to as ``depth-one tree,'' or ``$dot$,'' in brief. 
Characterization and search algorithms based on $dot$ expansion are then provided followed by a discussion on shortcomings of the dot-based approach.  All the discussions are in the context of variable-regular
Tanner graphs with variable degree $d_{\mathrm{v}}$, even if this is not explicitly mentioned sometimes.  

\subsection{LSS-based Approach}
\label{sec:lss }
To characterize LETSs of variable-regular LDPC codes, 
in \cite{mehdi2014},  Karimi and Banihashemi first generated all the graphical structures of LETSs in each $(a,b)$ class. They then examined each structure $\mathcal{S}$ 
to find out whether $\mathcal{S}$ is a layered super set (LSS) of any of its cycles, i.e., whether $\mathcal{S}$ corresponds to a sequence of LETS structures that starts 
from a cycle and grows one variable node at a time to reach $\mathcal{S}$. 
It was shown in \cite{mehdi2014} that a vast majority of non-isomorphic LETS structures are LSS of cycles. The relatively small number of structures
that do not lend themselves to LSS characterization (starting from cycles) are labeled as ``NA'' in  \cite{mehdi2014}.
Subsequently, it was shown in \cite{yoones2015} that NA structures, themselves, are also LSSs of some basic graphical structures which are only slightly more complex 
than cycles. These basic structures are called \textit{prime} with respect to the LSS property, or \textit{LSS-prime} in brief, implying that they are not layered super sets of smaller structures. 
By this definition, cycles that are not LSS of any of their subsets are included in the set of LSS-prime structures, and are called \textit{cycle LSS-prime} structures. 
The rest of the prime structures are referred to as \textit{non-cycle LSS-prime} structures. 

\subsection{Dot-based Characterization}
\label{sec: dot-based}

A careful examination of the LSS property in the space of normal graphs reveals that this property corresponds to a simple graph-based expansion technique, referred to, here, as \textit{depth-one tree ($dot$)} expansion.
The notation $dot_m$ is used for a {\em dot} expansion with $m$ edges, as shown in Fig. \ref{fig:doten}.
In this figure, symbol
\scalebox{1.2}{$\circ$}  is used to represent the $m$ common nodes between the parent structure $\mathcal{S}$ and the $dot_m$ structure. Symbol \scalebox{1.2}{$\bullet$}, on the other hand, 
is used to represent the root node of $dot_m$ structure. For the $dot_m$ expansion to result in a valid normal graph for a LETS structure of 
a variable-regular Tanner graph with variable degree $d_\mathrm{v}$, the $m$ edges of the tree will have to be connected to nodes
of $\mathcal{S}$ with degree less than $d_\mathrm{v}$. In addition, the degree $m$ of the root node must be at least 2 and 
at most $d_\mathrm{v}$. 
\begin{figure}[] 
\centering
\includegraphics [width=0.26\textwidth]{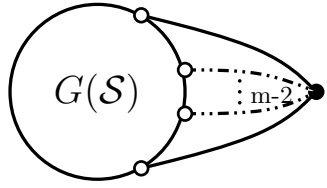}
\caption{Expanding a LETS structure $\mathcal{S}$ by applying depth-one tree expansion with $m$ edges ($dot_m$).}
\label{fig:doten}
\end{figure}
 
\begin{pro}
\label{pro:dot}
Suppose that $\mathcal{S}$ is an $(a,b)$ LETS structure of variable-regular Tanner graphs with variable degree $d_\mathrm{v}$,
where $b \geq 2$.
Expansion of $\mathcal{S}$ using $dot_m$, $ 2 \leq m \leq min\{d_\mathrm{v},b\}$, will result in 
LETS structure(s) in the  $(a+1,b+d_\mathrm{v}-2m)$ class. 
\end{pro}
\begin{proof}
In addition to being limited by the upper bound of $d_\mathrm{v}$, the maximum value of $m$ is also upper bounded by $b$. 
Therefore, $m \leq min\{d_\mathrm{v},b\}$.
In the $dot_m$ expansion, only one node is added to $\mathcal{S}$, and thus the size of the new structure is $a'=a+1$. 
Based on Lemma~\ref{lem1}, the number of edges in $\mathcal{S}$ is $|E_\mathcal{S}|=(a d_\mathrm{v}-b)/2$. 
By the nature of $dot_m$ expansion, the number of edges in the expanded structure $\mathcal{S}'$ is $|E_{\mathcal{S}'}|=|E_\mathcal{S}|+m$. 
Therefore, based on Lemma~\ref{lem1}, the number of unsatisfied check nodes $b'$ in the new structure is
$b'=(a+1)d_\mathrm{v}-2|E_{\mathcal{S}'}|=a d_\mathrm{v}-2|E_\mathcal{S}|+d_\mathrm{v}-2m=b+d_\mathrm{v}-2m$. 
\end{proof}

\begin{ex}
Consider the LETS structure $\mathcal{S}$ of Fig. \ref{fig:48}(a) in a variable-regular Tanner graph with $d_\mathrm{v}=4$ and $g=6$. This structure is in the $(4,8)$ class.
The application of $dot_2$ expansion to $\mathcal{S}$ generates two non-isomorphic structures shown in Figs. \ref{fig:48}(b) and (c) in the $(5,8)$ class. 
For each of the two expansions  $dot_3$ and $dot_4$ applied to $\mathcal{S}$, only one new structure is generated. This structure belongs to the $(5,6)$ and the $(5,4)$ class, respectively,
and is shown in Figs. \ref{fig:48}(d) and (e), respectively.
\end{ex}
\begin{figure}[] 
\centering
\includegraphics [width=0.35\textwidth]{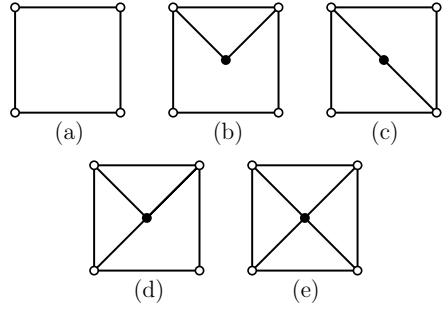}
\caption{LETS structures with $\Delta(S) \leq d_\mathrm{v}=4$: (a) The LETS structure $\mathcal{S}$ is in the $(4,8)$ class, (b)-(c) LETS structures in the $(5,8)$ class generated from $\mathcal{S}$ by $dot_2$ expansion, (d) 
LETS structure in the $(5,6)$ class generated from $\mathcal{S}$ by $dot_3$ expansion, and (e) LETS structure in the $(5,4)$ class generated from $\mathcal{S}$ by $dot_4$ expansion.}
\label{fig:48}
\end{figure}

\begin{pro}
Among the LETS structures with the same size $a$, the simple cycle of size $a$ in the $(a, a(d_\mathrm{v}-2))$ class has the largest $b$.
\label{progh}
\end{pro}
\begin{proof}
Since $b=a d_\mathrm{v}-2|E|$, the LETS structure with the largest $b$ is the one with the minimum number of edges $|E|$. Also, in leafless structures, each node has at least two incident edges. 
Therefore, the simple cycle of size $a$, in which each node has exactly two incident edges, has the minimum number of edges, $a$, among all the LETS structures with the same size. The simple cycle of size $a$ is, in fact, 
the only LETS structure in the $(a, a d_\mathrm{v}-2a=a (d_\mathrm{v}-2))$ class.
\end{proof}

\begin{rem}
In a Tanner graph with girth $g$, the smallest LETS structure is the simple cycle of size $g/2$, denoted by $s_{g/2}$. The $s_{g/2}$ structure
is the only LETS structure of size $g/2$.
\label{rem1}
\end{rem}

Given all the non-isomorphic LETS structures of size $k$, the {\em dot} expansion technique is able to generate most of the non-isomorphic LETS structures of size $k+1$. 
The LETS structures which cannot be obtained by the {\em dot} expansion from any of their subsets are called \textit{dot-prime structures}. For brevity, we refer to these structures
as being {\em prime}, in the rest of the paper. Prime structures can be categorized into: (i) \textit{cycle prime structures}, and  (ii) \textit{non-cycle prime structures}. In the first category, we have all the simple cycles, and cycles with chord that are out of the reach of the {\em dot} expansion. In the second category, we have structures that are slightly more complex than cycles. Non-cycle prime structures are discussed in details in \cite{yoones2015}. 

{\em Dot} expansion is a simple recursive technique to generate larger LETS structures from smaller ones. Based on Remark~\ref{rem1}, one needs to start from
the simple cycle of length $g/2$ and apply the {\em dot} expansion recursively, to obtain non-isomorphic LETS structures in different $(a,b)$ classes with $a > g/2$.
This technique, however, has the limitation of leaving out all the prime structures. Therefore, for a full characterization of LETS structures based on {\em dot} expansion, these
prime structures will have to be identified and added as new initial building blocks for the {\em dot} expansion.
The pseudo-code for the dot-based characterization algorithm is provided in Algorithm \ref{alg1}.

\begin{algorithm}
 \caption{{\bf (Dot-based Characterization Algorithm)} Finds all the non-isomorphic LETS structures in $(a,b)$ classes with $a \leq a_{max}$ and $b \leq b_{max}=a_{max} (d_\mathrm{v}-2)$, for a variable-regular Tanner graph with variable degree $d_\mathrm{v}$ and girth $g$. The prime and non-prime structures of size $k$ are stored in sets $\mathcal{P}_k$ and $\mathcal{L}_k, k = g/2, \ldots, a_{max}$, respectively.}
\label{alg1}
 \begin{algorithmic} [1] 
\State \textbf{Inputs:} $a_{max}, g, d_\mathrm{v}$.
\State \textbf{Initializations:} $\mathcal{L} \gets \emptyset$, $k=g/2$.
\While {$k \leq a_{max}$}
\parState {Find the prime LETS structures $\mathcal{P}_{k}$ of size $k$ which are not found in the algorithm (by comparing the structures in $\mathcal{L}$ with those generated by {\em geng}). }\label{nautmis}
\vspace{-6pt}
\State $\mathcal{L}_k$= \textbf{DotExpansion}$(\mathcal{P}_k, \mathcal{L}, a_{max})$. \label{shor}
\State $\mathcal{L}= \mathcal{L}\cup \mathcal{L}_k.$
\State $k=k+1$. \label{k=k+1}
\EndWhile
\State \textbf{Outputs:} $\{\mathcal{P}_{g/2},\dots,\mathcal{P}_{a_{max}}\}$, $\{\mathcal{L}_{g/2},\dots,\mathcal{L}_{a_{max}}\}$.
\end{algorithmic}
 \end{algorithm}
In this algorithm, all the non-isomorphic LETS structures in different $(a,b)$ classes in the range of $a \leq a_{max}$ and $b \leq b_{max}=a_{max} (d_\mathrm{v}-2)$ are generated,
and identified as being either prime or non-prime. The set $\mathcal{L}$ in Algorithm~\ref{alg1} is used to store all the LETS structures that are found so far.
In the while loop, first,  in Line \ref{nautmis}, the non-isomorphic LETS structures in different classes with size $k$ in $\mathcal{L}$ are compared with those obtained by the {\em geng} program. 
If there is a discrepancy, the LETS structures that were not generated by the algorithm (missing in $\mathcal{L}$) are chosen as the set of prime structures of size $k$, $\mathcal{P}_{k}$. 
Then, using the pseudo-code of Routine \ref{dotexp}, the set of all non-isomorphic LETS structures up to size $a_{max}$, $\mathcal{L}_k$, is generated using the {\em dot} expansion, starting from the structures in $\mathcal{P}_{k}$. 

In Routine \ref{dotexp}, set $\mathcal{L}_k^a$ is the set of non-isomorphic LETS structures of size $a$ generated by the recursive application of the {\em dot} expansion starting from the structures in $\mathcal{P}_k$.  
All the non-isomorphic LETS structures of size $a+1$ that are generated by expanding the non-isomorphic LETS structures in $\mathcal{L}_k^a$ using $dot_m$ expansions with different values of $m$
are stored in $\mathcal{L}_{tem}$. A subset of $\mathcal{L}_{tem}$ that is not generated by smaller prime structures will then be identified as the set $\mathcal{L}_k^{a+1}$.

\begin{routine}
 \caption{{\bf (DotExpansion)} Recursive expansion of a set of structures $\mathcal{P}_k$ of size $k$, up to size $a_{max}$, using {\em dot} expansion, excluding the already found
structures $\cal{L}$, and storing the rest in $\mathcal{L}_k$. $\mathcal{L}_k$= DotExpansion $(\mathcal{P}_k, \mathcal{L}, a_{max}$) }
\label{dotexp}
 \begin{algorithmic} [1] 
 \State \textbf{Initialization:}  $\mathcal{L}_k^k \gets \mathcal{P}_k$.
\For {$a=k,\dots, a_{max}-1$}
\State $\mathcal{L}_{tem} \gets \emptyset$.
\For {each LETS structure $\mathcal{S}$ in $\mathcal{L}_k^a$}
\parState {Find the subset $F'$ of nodes in $\mathcal{S}$ with degree strictly less than $d_\mathrm{v}$. }
\For {$m= 2,~\dots, d_\mathrm{v}$}
\parState {Construct $\mathcal{N}_m(\mathcal{S})$ to include all the possible subsets of $F'$ with size $m$.}
\For {each set of $m$ nodes in $\mathcal{N}_m(\mathcal{S})$ }
\parState {Generate a new structure, $\mathcal{S}'$ using $dot_m$ expansion by connecting a new node
 to the set of $m$ nodes.}
\State $\mathcal{L}_{tem}= \mathcal{L}_{tem} \cup \mathcal{S}'$. \label{nonism1}
\EndFor
\EndFor
\EndFor
\State $\mathcal{L}_k^{a+1}=\mathcal{L}_{tem} \setminus \mathcal{L}$. \label{nonism2}
\EndFor
\State $\mathcal{L}_k=\bigcup\limits _{a=k}^{a_{max}} \mathcal{L}_k^a$.
\State \textbf{Output:} $\mathcal{L}_k$.
\end{algorithmic}
\end{routine}

\begin{rem}
Expanding LETS structures with depth-one tree is the same as the LSS-based algorithm proposed in \cite{mehdi2014}. 
For the characterization of LETS structures, however, rather than examining each structure in a class to find out its LSS properties, as performed in \cite{mehdi2014},
in this work, the non-isomorphic LETS structures are generated by expanding the prime structures using the {\em dot} expansion technique. 
\end{rem}

\noindent
 {\bf Case Study}-{\em Characterization of non-isomorphic  LETS structures of $(a,b)$ classes for variable-regular graphs with $d_\mathrm{v} = 3$, $g=6$ and $a \leq 9, b \leq 9$}: 

This characterization is summarized in Table \ref{tab:3,6gh1}, where the prime structures  are boldfaced.
\begin{table}[]
\centering
\caption{dot-based characterization of non-isomorphic LETS structures of $(a,b)$ classes for variable-regular graphs with $d_\mathrm{v} = 3$, $g=6$ and $a \leq 9, b \leq 9$}
\label{tab:3,6gh1}
\begin{tabular}{||c|c| c|c|c| c| c|c|| }
\cline{1-8}
&$a = 3$&$a = 4$&$a = 5$&$a = 6$&$a = 7$&$a = 8$&$a = 9$\\
\cline{1-8}
b = 0&-&$  \begin{array}{@{}c@{}}s_3(1)\end{array} $&-&$  \begin{array}{@{}c@{}}s_4 (2)\end{array} $&-&$  \begin{array}{@{}c@{}}s_5 (3),s_6 (1),c_6 (1)\end{array} $&-\\
\cline{1-8}
b = 1&-&-&$  \begin{array}{@{}c@{}}s_3(1)\end{array} $&-&$  \begin{array}{@{}c@{}}s_4 (3),s_5 (1)\end{array} $&-&$  \begin{array}{@{}c@{}}s_5 (9),s_6 (5),c_6 (2)\\n_6 (1),c_7 (1),n_7 (1)\end{array} $\\ 
\cline{1-8}
b = 2&-&$  \begin{array}{@{}c@{}}s_3(1)\end{array} $&-&$  \begin{array}{@{}c@{}}s_4 (3),s_5 (1)\end{array} $&-&$  \begin{array}{@{}c@{}}s_5 (9),s_6 (5),c_6 (2)\\n_6 (1),c_7 (1),n_7 (1)\end{array} $&-\\
\cline{1-8}
b = 3&$  \begin{array}{@{}c@{}}\boldsymbol{s_3(1)}\end{array} $&-&$  \begin{array}{@{}c@{}}s_4(2)\end{array} $&-&$  \begin{array}{@{}c@{}}s_5 (6),s_6 (1)\\c_6 (2),n_6 (1)\end{array} $&-&$  \begin{array}{@{}c@{}}s_6 (31),s_7 (9),c_7 (9)\\n_7 (7),c_8 (3),n_8 (4)\end{array} $\\
\cline{1-8}
b = 4&-&$  \begin{array}{@{}c@{}}\boldsymbol{s_4(1)}\end{array} $&-&$  \begin{array}{@{}c@{}}s_5 (2),\boldsymbol{c_6 (1)}\\\boldsymbol{n_6 (1)}\end{array} $&-&$  \begin{array}{@{}c@{}}s_6 (12),s_7 (2),c_7 (4)\\n_7 (4),\boldsymbol{c_8 (2)},\boldsymbol{n_8 (1)}\end{array} $&-\\
\cline{1-8}
b = 5&-&-&$  \begin{array}{@{}c@{}}\boldsymbol{s_5(1)}\end{array} $&-&$  \begin{array}{@{}c@{}}s_6 (3),\boldsymbol{c_7 (1)}\\\boldsymbol{n_7 (2)}\end{array} $&-&$  \begin{array}{@{}c@{}}s_7 (19),s_8 (2),c_8 (11)\\n_8 (15),\boldsymbol{c_9 (2)},\boldsymbol{n_9 (3)}\end{array} $\\
\cline{1-8}
b = 6&-&-&-&$  \begin{array}{@{}c@{}}\boldsymbol{s_6(1)}\end{array} $&-&$  \begin{array}{@{}c@{}}s_7 (3),\boldsymbol{c_8 (2)},\boldsymbol{n_8 (5)}\end{array} $&-\\
\cline{1-8}
b = 7&-&-&-&-&$  \begin{array}{@{}c@{}}\boldsymbol{s_7(1)}\end{array} $&-&$  \begin{array}{@{}c@{}}s_8 (4),\boldsymbol{c_9 (2)},\boldsymbol{n_9 (7)}\end{array} $\\
\cline{1-8}
b = 8&-&-&-&-&-&$  \begin{array}{@{}c@{}}\boldsymbol{s_8 (1)}\end{array} $&-\\
\cline{1-8}
b = 9&-&-&-&-&-&-&$  \begin{array}{@{}c@{}}\boldsymbol{s_9 (1)}\end{array} $\\
\cline{1-8}
\end{tabular}
\end{table}
The notations $s_k$, $c_k$, and $n_k$ are used to denote the simple cycle, a prime cycle with chord, and  a non-cycle prime structure of size $k$, respectively.  
As an example of how the entries in Table~\ref{tab:3,6gh1} should be interpreted, we consider the $(7,5)$ class.
The entries for this class are: $s_6(3), c_7(1),$ and $n_7(2)$. This means that there are a total of 6 non-isomorphic LETS structures in this class, three of these structures are generated, using {\em dot} expansions, starting from $s_6$, 
one of them is a prime cycle with chord of size 7, and two of them are non-cycle prime structures of size 7. 
Having the symbol ``-" for an $(a,b)$ class means that it is impossible to have any LETS structure in that class. 
In Table \ref{tab:3,6gh1}, the smallest possible LETS structure is the cycle of size 3, $s_3$. The two non-isomorphic LETS structures of size 4 generated in Algorithm \ref{alg1} 
by the application of {\em dot} expansions to $s_3$ are shown in Fig. \ref {fig:graphex4} ($f_4$ is the new node added to $s_3$). These structures, which belong to  $(4,2)$ and $(4,0)$ classes, respectively, 
are the only non-isomorphic LETS structures of size 4 that are generated by the {\em dot} expansion. There, however, remains another LETS structure of size 4 in the $(4,4)$ class that is not generated by the {\em dot} expansion. 
This $(4,4)$ LETS structure, which is in fact $s_4$, is thus chosen as a prime structure in Step \ref{nautmis} of Algorithm \ref{alg1}. 
\begin{figure}[t] 
\centering
\includegraphics [width=0.33\textwidth]{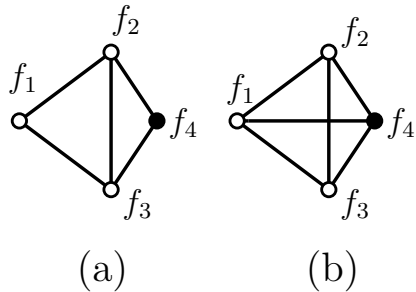}
\caption{LETS structures in (a)  $(4,2)$  and (b) $(4,0)$ classes in variable-regular graphs with $d_\mathrm{v}=3$ and $g=6$.}
\label{fig:graphex4}
\end{figure}

Fig. \ref{fig:cyc6}(a) shows two prime structures in the $(6,6)$ and $(6,4)$ classes, respectively, and Fig. \ref{fig:cyc6}(b) depicts two prime structures in the $(7,5)$ class.
\begin{figure}[t] 
\centering
\includegraphics [width=0.35\textwidth]{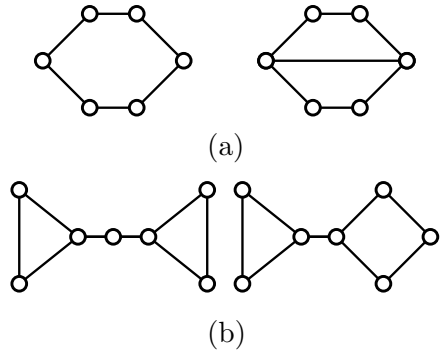}
\caption{Prime LETS structures in variable-regular graphs with $d_\mathrm{v}=3$ and $g=6$ (a) a simple cycle in the $(6,6)$ class and a cycle with chord in the $(6,4)$ class (b) non-cycle prime structures in the $(7,5)$ class.}
\label{fig:cyc6}
\end{figure}

It is worth noting that unlike \cite{mehdi2014}, the simple cycles and  prime cycles with chord are distinguished in this paper.

\subsection{Dot-based Search Algorithm}

Dot-based characterization of LETS structures corresponds to a search algorithm to find all  the instances of LETS structures in the Tanner graph of LDPC codes with different values of $d_\mathrm{v}$, $g$, $a_{max}$ and $b_{max}$.
First, based on the characterization table, the set of prime structures in the interest range of $a \leq a_{max}$ and $b \leq b_{max}$ are identified. For example, for variable-regular graphs with $d_\mathrm{v} = 3$, $g=6$ and $b_{max}=a_{max}=9$, all the prime structures of Table \ref{tab:3,6gh1} are needed. As another example, if, for a variable-regular graph with $d_{\mathrm{v}} =3$ and $g=6$, one is interested in $b_{max}=3$ and $a_{max}=8$, only 8 prime structures in the $(3,3)$, $(4,4)$, $(5,5)$, $(6,4)$, $(6,6)$ and $(7,5)$ classes are needed.
Suppose that $T$ prime structures are needed for $a \leq a_{max}$ and $b \leq b_{max}$. The instances of each prime structure are then enumerated in the Tanner graph of the code. 
Consider $\mathcal{IP}_1,\dots , \mathcal{IP}_T$, to be the sets of all the instances of $T$ prime structures in the graph. These sets are used as the input to the dot-based search algorithm to find 
all the instances of LETS structures in the interest range of $a \leq a_{max}$ and $b \leq b_{max}$.
The pseudo-code of the dot-based search algorithm is given in Algorithm \ref{tab:search al1}. 
\begin{algorithm}
\centering
\caption{{\bf (Dot-based Search Algorithm)} Finds all the instances of $(a,b)$ LETS structures of a variable-regular Tanner graph $G$ with girth $g$, for $a \leq a_{max}$, starting from
all the instances $\mathcal{IP}_1,\dots , \mathcal{IP}_T$, of the $T$ prime structures identified in the characterization table. The set $\mathcal{I}_t$ in the output
contains all the instances of LETS in $G$ generated from $\mathcal{IP}_t$ for $ t = 1, \ldots, T$.}
\label{tab:search al1}
\begin{algorithmic}[1]
\State  \textbf{Inputs:} $G$, $\{\mathcal{IP}_1,\dots , \mathcal{IP}_T\}$, $g$, $a_{max}$. 
\State  \textbf{Initializations:} $\mathcal{I}_t^a \gets \emptyset$, for $1 \leq t \leq T$, $g/2 \leq a \leq a_{max}$, $\mathcal{I}\gets \emptyset$.
\For {$t = 1, \dots, T$}
\parState {$\mathcal{I}_t^k \gets \mathcal{IP}_t$, where $k$ is the size of the instances in $\mathcal{IP}_t$. }
\State $a=k$. 
\While {$a < a_{max}$}
\State $\mathcal{I}_t^{a+1}$= \textbf{DotSrch}$(\mathcal{I}_t^{a},\mathcal{I})$. \label{reftoOneDotSrch}
\State $\mathcal{I}=\mathcal{I} \cup\mathcal{I}_t^{a+1}$.
\State $a=a+1$.
\EndWhile
\State $\mathcal{I}_t=\bigcup\limits _{a=k}^{a_{max}} \mathcal{I}_t^{a}$
\EndFor
\State \textbf{Outputs:} $\{\mathcal{I}_1,\dots ,\mathcal{I}_T\}$.
\end{algorithmic}
\end{algorithm}

In Algorithm~\ref{tab:search al1}, the set $\mathcal{I}_t^a $ is the set of instances of LETSs of size $a$ which are found starting from the instances $\mathcal{IP}_t$ of the $t$th prime structure, 
the set $\mathcal{I}_t $ is the union of sets $\mathcal{I}_t^a $ with different sizes $a$, and $\mathcal{I}$ is the set of instances of all LETSs which are found so far in the algorithm. 
The outputs of the search algorithm are the sets of instances of LETSs in all classes starting from  instances of $T$ prime structures.
Given a set of instances of LETS structures of size $a$, in Line \ref{reftoOneDotSrch} of Algorithm~\ref{tab:search al1}, Routine \ref{tab:search $dot$} finds a set of instances of LETS structures of size $a+1$, using {\em dot} expansions. 
For each instance $\mathcal{S}$ of a LETS structure in $\mathcal{I}_t^a$, some instances $\mathcal{I}_{tem}$ of LETS structures with size $a+1$ are found. 
The instances $\mathcal{I}_t^{a+1}$ in $\mathcal{I}_{tem}$, which were not already found and included in $\mathcal{I}$, are the output of this routine.
\begin{routine}
\centering
\caption{{\bf (DotSrch)} Expansion of a set of instances $\mathcal{I}_t^a$ of LETS structures of size $a$ using {\em dot} expansions to find a set of instances  of LETS structures of size $a+1$,
excluding the already found structures $\mathcal{I}$, and storing the rest in $\mathcal{I}_t^{a+1}$. 
$\mathcal{I}_t^{a+1}$= DotSrch $( \mathcal{I}_t^{a},\mathcal{I})$ }
\label{tab:search $dot$}
\begin{algorithmic}[1]
\State  \textbf{Initializations:} $\mathcal{I}_{tem} \gets \emptyset$.
\For{each instance of LETS structure $\mathcal{S}$ in $\mathcal{I}_t^a$}
\parState {Consider $\mathcal{V}$ to be the set of variable nodes in $V \setminus \mathcal{S}$, which have at least two connections with the check nodes in $\Gamma_{o}{(\mathcal{S})}$ and have no connection with the check nodes in $\Gamma_{e}{(\mathcal{S})}$. }
\vspace{-6pt}
\For {each variable node $v \in \mathcal{V}$}
\State $\mathcal{I}_{tem}=\mathcal{I}_{tem}\cup \{\mathcal{S} \cup v\}$.
\EndFor
\EndFor
\State $\mathcal{I}_t^{a+1} \gets \mathcal{I}_{tem}\setminus \mathcal{I}$.
\State \textbf{Output:} $ \mathcal{I}_t^{a+1} $.
\end{algorithmic}
\end{routine}

The dot-based search algorithm presented here is essentially the same as the LSS-based search algorithm of \cite{mehdi2014},\cite{yoones2015}.

\subsection{Shortcomings of Dot-based Search Algorithm}
\label{sec:problems}
To motivate the new characterization/search technique, we start by explaining the main practical issues in implementing a dot-based search algorithm
through a number of examples. In summary, in a given Tanner graph, to exhaustively search for $(a,b)$ LETS instances within the range of interest $a \leq a_{max}$ and $b \leq b_{max}$, 
one needs to initially enumerate instances of prime structures that are needed as parents of the LETS structures under consideration. Some of these prime structures
have relatively large $a$ and $b$ values. For practical codes, the multiplicity of instances of such prime structures
can be easily in the range of tens of millions. This imposes a huge computational burden and memory requirement on the dot-based search algorithm.
The dot-based search algorithm in this case is particularly inefficient, when such prime structures, themselves, are not of direct interest (as they have relatively large $a$ and $b$ values), and when 
the multiplicity of the instances of the LETS structures, that are children of such prime structures and of interest, is relatively low, or in some cases even zero. 

We start by the case of variable-regular graphs with $d_\mathrm{v} = 3$ and $g=6$. For such graphs, the dot-based characterization for the range of $a \leq 9$ and $b \leq 9$, was
presented in Table \ref{tab:3,6gh1}. In Table \ref{tab:3,6gh2}, we have extended the results of Table \ref{tab:3,6gh1} to cover LETS structures with 
$a=10, 11, 12$,  and $b \leq 5$. All the (non-cycle and cycle) prime structures of size less than 10 (those identified in Table~\ref{tab:3,6gh1} by $s_k, c_k, n_k$, with $k \leq 9$) are used to generate the LETS structures in Table \ref{tab:3,6gh2}. 
\begin{table}[]
\centering
\caption{dot-based characterization of non-isomorphic LETS structures of $(a,b)$ classes for variable-regular graphs with $d_\mathrm{v} = 3$, $g=6$, $10 \leq a \leq 12$ and $b \leq 5$ }
\label{tab:3,6gh2}
\begin{tabular}{||c|c| c| c|| }
\cline{1-4}
&$a = 10$&$a = 11$&$a = 12$\\
\cline{1-4}
b = 0&$  \begin{array}{@{}c@{}}s_6(12),s_7(3),c_7(2)\\ n_7(1), n_8(1)\end{array} $&-&$  \begin{array}{@{}c@{}}c_7 (43),s_8(15),c_8(12)\\n_8(9),s_9(1),c_9(1),n_9(4)\end{array} $\\
\cline{1-4}
b = 1&-&$  \begin{array}{@{}c@{}}s_6 (50),s_7 (26),c_7 (10),n_7 (7)\\s_8(4),c_8(6), n_8(9), n_9(2)\end{array} $&-\\ 
\cline{1-4}
b = 2&$  \begin{array}{@{}c@{}}s_6 (50),s_7 (24),c_7 (10),n_7 (7)\\ s_8(3), c_8(8), n_8(10), n_9(1)\end{array} $&-&$  \begin{array}{@{}c@{}}s_7 (350),s_8 (156),c_8 (93), n_8 (73) \\s_9(22),c_9(44), n_9(77),\textbf{NA(20)}\end{array} $\\
\cline{1-4}
b = 3&-&$  \begin{array}{@{}c@{}}s_7 (211),s_8 (65),c_8 (72),n_8(65)\\s_9(4), c_9(21), n_9(41),\textbf{NA(3)}\end{array} $&-\\
\cline{1-4}
b = 4&$  \begin{array}{@{}c@{}}s_7 (75),s_8 (20),c_8 (37),n_8(37),\\c_9(10), n_9(18),\textbf{NA(1)}\end{array} $&-&$  \begin{array}{@{}c@{}}s_8 (711),s_9(202),c_9(322) \\n_9(346),\textbf{NA(311)}\end{array} $\\
\cline{1-4}
b = 5&-&$  \begin{array}{@{}c@{}}s_8 (172),s_9(40),c_9(104) \\n_9(139),\textbf{NA(81)}\end{array} $&-\\
\cline{1-4}
\end{tabular}
\end{table}
For the LETS structures which are not generated using the {\em dot} expansion technique starting from any of these prime structures, the notation ``NA'' is used in Table~\ref{tab:3,6gh2}. 
Starting from the prime structures of size 10 and 11, one can also generate (characterize) NA LETS structures of Table~\ref{tab:3,6gh2}.

\begin{ex} The $(11,1)$ class is a potentially dominant LETS class of variable-regular graphs with $d_\mathrm{v}=3$ and $g=6$. Table  \ref{tab:3,6gh2} shows that among 114 non-isomorphic LETS structures in this class, 50, 26, 10, 7, 4, 6, 9 and 2 structures are generated starting from prime structures $s_6$, $s_7$, $c_7$, $n_7$, $s_8$, $c_8$, $n_8$ and $n_9$, respectively. This means prime structures with size up to $9$ need to be enumerated 
in the Tanner graph to find the instances of all the LETS structures of this class (if any exists). 
Fig. \ref{fig:example(11,1)} represents $\mathcal{S}=(\{f_1,f_2,\dots,f_{11}\}, E)$, a LETS structure in the $(11,1)$ class which is generated starting from the simple cycle prime structure, $s_8$, in the $(8,8)$ class $(\mathcal{S}'=(\{f_1,f_2,\dots,f_{8}\}, E'))$.
 \begin{figure}[] 
\centering
\includegraphics [width=0.22\textwidth]{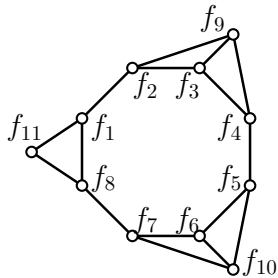}
\caption{ A LETS structure in the $(11,1)$ class of a variable-regular graph with $d_\mathrm{v}=3$ and $g=6$.}
\label{fig:example(11,1)}
\end{figure}
To find all the instances of $\mathcal{S}$ in the Tanner graph using a dot-based search algorithm, all the instances of $s_8$ in the graph need to be enumerated. 
\end{ex}

\begin{ex}
In the $(12,2)$ class of Table \ref{tab:3,6gh2}, in addition to LETS structures that are generated by the dot-based expansion starting from $c_7$, $s_8$, $c_8$, $n_8$, $s_9$, $c_9$ and $n_9$, there are 20 non-isomorphic 
LETS structures that are generated from prime structures of size 10 and 11 (denoted as NA in the table).  
Fig. \ref{fig:example(12,2)} shows the structure $\mathcal{S}=(\{f_1,f_2,\dots,f_{12}\}, E)$, as one of these NA structures.  This structure is generated starting from 
a non-cycle prime structure of size 10 in the $(10,8)$ class ($\mathcal{S}'=(\{f_1,f_2,\dots,f_{10}\}, E')$).
\begin{figure}[t] 
\centering
\includegraphics [width=0.37\textwidth]{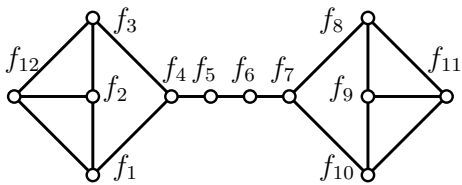}
\caption{A LETS structure in the $(12,2)$ class of a variable-regular graph with $d_\mathrm{v}=3$ and $g=6$.}
\label{fig:example(12,2)}
\end{figure}
Although a certain Tanner graph with $d_\mathrm{v}=3$ and $g=6$ may not contain any instance of $\mathcal{S}$, all the instances of the structure $\mathcal{S}'$ in the Tanner graph 
need to be enumerated for the dot-based search algorithm to guarantee the exhaustive coverage of the $(12,2)$ class. 
\end{ex}

\begin{ex}
Consider a $(3,6)$-regular LDPC code with $g=6$, and $n=20000$~\cite{mackayencyclopedia}.\footnote{This code is labeled as $\mathcal{C}_{10}$ in Section~\ref{sec:numerical}.}
\begin{table}[]
\centering
\caption{ Multiplicities of prime structures of size 8 and 9 in the code of Example~\ref{ex15}}
\label{tab:c20000}
\begin{tabular}{||c|r||c|r|| }
\cline{1-4}
Prime & Multiplicity& Prime& Multiplicity \\
structure&  &structure&  \\
\cline{1-4}
$c_8$ in $(8,4)$ class&0&$c_9$ in $(9,5)$ class&15\\
\cline{1-4}
$c_8$ in $(8,6)$ class&6430&$c_9$ in $(9,7)$ class&83392\\
\cline{1-4}
$s_8$ in  $(8,8)$ class&6219941&$s_9$ in $(9,9)$ class&55181079\\
\cline{1-4}
$n_8$ in  $(8,4)$ class&0&$n_9$ in $(9,5)$ class&1\\
\cline{1-4}
$n_8$ in  $(8,6)$ class&6754&$n_9$ in $(9,7)$ class&101777\\
\cline{1-4}
\end{tabular}
\end{table}
Suppose that one is interested in LETSs of this code in the range of  $a \leq12$ and $b \leq 5$. The multiplicities of all the instances of prime structures of size 8 and 9 in the Tanner graph of this code in 
the $(8,4)$, $(8,6)$, $(8,8)$, $(9,5)$, $(9,7)$ and $(9,9)$ classes are listed in Table \ref{tab:c20000}.
In the range $a \leq12$ and $b \leq 5$, the $(8,4)$, $(9,3)$, $(9,5)$, $(10,0)$, $(10,2)$, $(10,4)$, $(11,1)$, $(11,3)$, $(11,5)$, $(12,0)$, $(12,2)$, $(12,4)$ classes are the classes 
that can have LETS structures starting from these prime structures. To exhaustively search the Tanner graph for LETSs with $a \leq12$ and $b \leq 5$, the dot-based search algorithm
will have to enumerate all the instances of the prime structures listed in Table \ref{tab:c20000} (a total of more than 60 million). In all the above classes, however, 
this code has only 15, 1 and 8 instances of LETS structures in $(9,5)$, $(10,4)$ and $(11,5)$ classes, respectively, that are generated by the dot-based expansion, starting from the instances of these (more than
60 million) prime structures.
\label{ex15}
\end{ex}

\begin{ex} 
Consider a $(4,8)$-regular LDPC code with $g=6$, and $n=4000$~\cite{mackayencyclopedia}.\footnote{This code is labeled as $\mathcal{C}_{22}$ in Section~\ref{sec:numerical}.}
Suppose that one is interested in finding all the LETSs of this code in the range $a \leq 9$ and $b \leq 8$, using the dot-based search algorithm.
To do this, based on Table V of \cite{mehdi2014}, one must enumerate all the instances of simple cycle prime structure, $s_7$. 
The LETS structure in Fig. \ref{fig:example(9,8)} is an example of a structure in the $(9,8)$ class that has $s_7$ as its prime sub-structure.
It appears that while there are $123,111,331$ instances of  $s_7$ in the Tanner graph of this code, there is not a single instance of a LETS structure in all the classes of interest which has $s_7$ as its prime sub-structure, in the Tanner graph.

\begin{figure}[t] 
\centering
\includegraphics [width=0.24\textwidth]{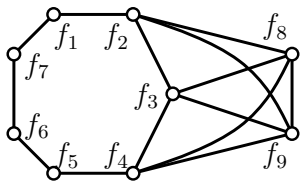}
\caption{A  LETS structure in the $(9,8)$ class of a variable-regular graph with $d_\mathrm{v}=4$ and $g=6$.}
\label{fig:example(9,8)}
\end{figure}
\end{ex}

\section{Dpl-based Characterization/Search Approach}
\label{sec:dpl-based}

To overcome the shortcomings of the dot-based search algorithm, explained in Subsection~\ref{sec:problems}, in this section, we first propose two new graph expansion
techniques, {\em path} and {\em lollipop} expansions, that are applied to smaller LETS structures to generate larger ones. Similar to {\em dot} expansion, the new expansions are also applied in the space of normal graphs. 
In Subsection~\ref{dpl}, we then use  {\em path} and {\em lollipop} expansions along with the {\em dot} expansion, 
to formulate a new characterization of LETS structures. In the context of characterizing a LETS structure as a sequence of embedded LETS structures that starts from
a simple cycle and expands, step by step, to reach the structure of interest, we introduce the concept of {\em minimal} characterizations, as those in which none
of the expansions in the sequence can be divided into smaller ones. The proposed characterization in Subsection~\ref{dpl} is minimal, and 
describes each and every LETS structure as an expansion of a simple cycle, using a combination of {\em dot}, {\em path} and {\em lollipop} expansion techniques, 
thus the name {\em dpl-based characterization}, or {\em dpl characterization}, in brief. In Subsection~\ref{dpl}, we also prove that any minimal characterization
is based only on the expansions {\em dot}, {\em path} and {\em lollipop}, and that the proposed characterization has some optimal properties 
as related to the corresponding search algorithm. The characterization results for some values of
$d_{\mathrm{v}}, g, a_{max}$, and $b_{max}$ are presented as characterization tables in Subsection~\ref{dpl-tables}.
For a given range $a \leq a_{max}$ and $b \leq b_{max}$, the maximum length of simple cycles participating in the {\em dpl} characterization is often less than
that for the $dot$ characterization. Finally, in Subsection~\ref{dpl-search}, we develop the search algorithm corresponding to the {\em dpl} characterization,
and in Subsection~\ref{complex}, we discuss the complexity of the search. 

\subsection{ \textit{Path} and \textit{Lollipop} Expansion Techniques}

Consider an $(a,b)$ LETS structure ${\cal S}$ of a $d_{\mathrm{v}}$-regular Tanner graph with $g \geq 6$. The {\em path} expansion of ${\cal S}$ is a new LETS structure ${\cal S}'$ of size $a+m$, that
is constructed by appending a  path of length $m+1$ to ${\cal S}$. The first and the last nodes of the  path are common with ${\cal S}$, and can be identical, in that case, the  path is closed.
Figures \ref{fig:paten}(a) and (b), show the {\em path} expansion of $\mathcal{S}$ using open and closed paths of length $m+1$, respectively. In these figures, the symbol \scalebox{1.2}{$\circ$} 
is used to represent the common node(s) between $\mathcal{S}$ and the path, and the symbol \scalebox{1.2}{$\bullet$} is used to represent the other $m$ nodes of the path. 
It is clear that for an {\em open-path} expansion, the degrees of the two nodes that are common with ${\cal S}$ must be strictly less that $d_{\mathrm{v}}$ (in $G({\cal S})$), and for the {\em closed-path} expansion,
the degree of the one common node must be strictly less than $d_{\mathrm{v}}-1$ in $G({\cal S})$. We use the notations $pa^o_m$ and $pa^c_m$ for open and closed paths of length $m+1$, respectively.  
The notation $pa_m$ is used to include both open and closed paths. 
\begin{figure}[] 
\centering
\includegraphics [width=0.55\textwidth]{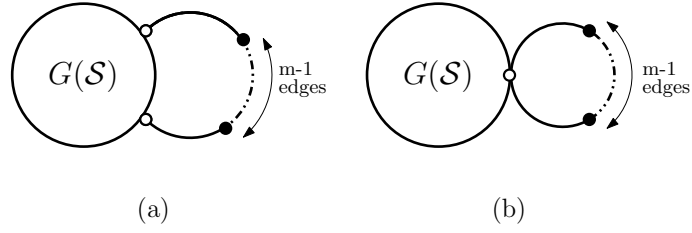}
\caption{Expansion of the LETS structure $\mathcal{S}$ with (a) an open $path$ of length $m+1, pa_m^o$ (b) a closed $path$ of length $m+1, pa_m^c$.}
\label{fig:paten}
\end{figure}

\begin{pro}
\label{pro:pa}
Suppose that $\mathcal{S}$ is a LETS structure in the $(a,b)$ class $(b \geq 2)$ for variable-regular Tanner graphs with variable degree $d_{\mathrm{v}}$. Applying the  path expansion $pa_m$, with $m \geq 2$, to $\mathcal{S}$ will 
result in LETS structure(s) in the $(a+m, b-2+m(d_\mathrm{v}-2))$ class. 
\end{pro}
\begin{proof}
The constraint $b \geq 2$ is to accommodate the 2 edges of the path that are connected to $\mathcal{S}$. 
We also note that $pa^o_1$ is the same as $dot_2$, and that the expansion $pa^c_1$ will result in a cycle of length 4. 
We thus focus on $m \geq 2$ in the statement of the proposition. 

The expanded structure $\mathcal{S}'$ has size $a'=a+m$. Based on Lemma~\ref{lem1}, the number of edges in $\mathcal{S}$ is $|E_\mathcal{S}|=(a d_\mathrm{v}-b)/2$. 
By the nature of expansion, the number of edges in $\mathcal{S}'$  is $|E_{\mathcal{S}'}|=|E_\mathcal{S}|+m+1$. By combining these and using Lemma~\ref{lem1}, we
have $b'=a'd_\mathrm{v}-2|E_{\mathcal{S}'}|=(a + m) d_\mathrm{v} - (a d_\mathrm{v} - b) -2m - 2=b-2+m(d_\mathrm{v}-2)$.
\end{proof}
The pseudo-code of the {\em path} expansion algorithm is given in Routine \ref{Onepath}.
\begin{routine}
\caption{{\bf (PathExpansion)} Expansion of a LETS structure $\mathcal{S}$ to all possible LETS structures $\mathcal{L}^p$ of size $|{\cal S}|+m$, using $pa_m$. $\mathcal{L}^p$ = PathExpansion$(\mathcal{S},m)$}
\label{Onepath}
 \begin{algorithmic} [1] 
\State \textbf{Initialization:}  Identify the subset of nodes in $\mathcal{S}$ with degree strictly less that $d_\mathrm{v}$, and store them in $F'$.
Identify the subset of  nodes in $\mathcal{S}$ with degree strictly less than $d_\mathrm{v}-1$, and store them in $F''$.
$\mathcal{L}^p  \gets \emptyset$.
\parState {Construct $\mathcal{N}_2(G)$ to include all the possible subsets of $F'$ of size 2.}
\For {each set in $\mathcal{N}_2(G)$ }
\parState {Generate a new LETS structure $\mathcal{S}'$ by applying $pa^o_m$ expansion to $\mathcal{S}$ such that the 2 nodes in $F'$ are the first and the last nodes of the new path of length $m+1$.}
\vspace{-6pt}
\State  $\mathcal{L}^p \gets \mathcal{L}^p \cup \mathcal{S}'$.
\EndFor
\For {each node in $F''$ }
\parState {Generate a new LETS structure $\mathcal{S}''$ by applying $pa^c_m$ expansion to $\mathcal{S}$ such that the  node in $F''$ is the common node between $\mathcal{S}$ and the new path of length $m+1$.}
\vspace{-6pt}
\State  $\mathcal{L}^p \gets \mathcal{L}^p \cup \mathcal{S}''$.
\EndFor
\State \textbf{Output:} $\mathcal{L}^p$.
 \end{algorithmic}
 \end{routine}

In Fig. \ref{fig:lolen}, the expansion of a LETS structure $\mathcal{S}$ using a lollipop walk of length $m+1$ is shown. 
\begin{figure}[t] 
\centering
\includegraphics [width=0.38\textwidth]{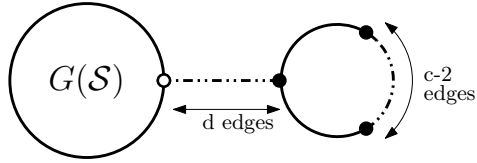}
\caption{Expansion of the LETS structure $\mathcal{S}$ with a lollipop walk of length $m+1=d+c$, $lo^c_m$.}
\label{fig:lolen}
\end{figure}
The notation $lo^c_m$ is used for a lollipop walk of length $m+1$ ($m$ is number of the nodes added to $\mathcal{S}$), 
which consists of a cycle of length $c$ $(c \geq 3)$ and a path of length $d$ $(d \geq 1)$.
Clearly, the common node between ${\cal S}$ and the {\em lollipop} expansion must have a degree strictly less than $d_{\mathrm{v}}$ in $G({\cal S})$.

\begin{pro}
\label{pro:lo}
Suppose that $\mathcal{S}$ is a LETS structure in the $(a,b)$ class $(b \geq 1)$ for variable-regular Tanner graphs with variable degree $d_\mathrm{v}$.
Applying the lollipop expansion $lo^c_m$, with  $ m \geq 3$ and $3 \leq c \leq m$, to $\mathcal{S}$ will result in LETS structure(s) in the $(a+m, b-2+m(d_\mathrm{v}-2))$ class. 
\end{pro}
\begin{proof}
Similar to that of Proposition~\ref{pro:pa}.
\end{proof} 
The pseudo-code for the {\em lollipop} expansion is given in Routine \ref{One$lollipop$}.
\begin{routine}
\caption{{\bf (LollipopExpansion)} Expansion of a LETS structure $\mathcal{S}$  to all possible LETS structures $\mathcal{L}^l_c$ of size $|{\cal S}|+m$, using $lo^c_m$. $\mathcal{L}^l_c$ = LollipopExpansion$(\mathcal{S},m,c)$}
\label{One$lollipop$}
 \begin{algorithmic} [1] 
\State \textbf{Initialization:}  Identify the subset of nodes in $\mathcal{S}$ with degree strictly less than $d_\mathrm{v}$ and store them in $F'$. $\mathcal{L}^l_c  \gets \emptyset$.

\For {each node $f$ in $F'$ }
\parState {Generate a new LETS structure, $\mathcal{S}'$, by expanding ${\cal S}$ using $lo^c_m$ expansion such that $f$ is the common node between the lollipop and ${\cal S}$.}
\vspace{-6pt}
\State  $\mathcal{L}^l_c \gets \mathcal{L}^l_c \cup \mathcal{S}'$.
\EndFor
\State \textbf{Output:} $\mathcal{L}^l_c$.
 \end{algorithmic}
 \end{routine}

In the following, we demonstrate, through a sequence of intermediate results, that any LETS structure of variable-regular Tanner graphs can be generated 
from simple cycles using a combination of the three expansion techniques: {\em dot}, {\em path} and {\em lollipop}.

\begin{lem}
\label{lem:simple}
Except the simple cycles, any LETS structure $\mathcal{S}'=(F',E')$ contains at least one LETS sub-structure $\mathcal{S}=(F,E)$ with $|F| < |F'|$.
\end{lem}

\begin{proof}
Since the structure $\mathcal{S}'$ is leafless, it must contain a cycle, and as a result, a simple cycle. The proof then follows from the fact that simple cycles are LETSs.
\end{proof}

\begin{lem}
\label{lem:$dot$a-1}
Suppose that a LETS structure $\mathcal{S}'=(F',E')$ has a LETS sub-structure $\mathcal{S}=(F,E)$, where $|F|=|F'|-1$. 
Structure $\mathcal{S}'$ can then be generated from $\mathcal{S}$ by using the  dot expansion.
\end{lem}

\begin{proof}
Node $f$ in $F'\backslash F$ must be connected to at least two nodes in $\mathcal{S}$. Based on the definition of {\em dot} expansion, structure $\mathcal{S}'$ 
can then be generated from $\mathcal{S}$ by using $dot_m$ expansion ($m \geq 2$), where the root of the expansion tree is $f$.
\end{proof} 

\begin{lem}
\label{lem-as}
Consider a LETS structure $\mathcal{S}'=(F',E')$ that is prime with respect to  dot expansion, and is not a simple cycle. Then, 
for the largest LETS sub-structure $\mathcal{S}=(F,E)$ of ${\cal S}'$, we have \\
(i) $|F| < |F'|-1$,\\
(ii) subgraph induced by $F' \backslash F$ is connected,\\
(iii) no node in $F'\setminus F$ can have more than one  edge connected to the nodes in ${\cal S}$.
\end{lem}

\begin{proof}
Proof of (i) follows from Lemmas~\ref{lem:simple} and~\ref{lem:$dot$a-1}, and the definition of a prime structure. For (ii),
assume that the subgraph is disconnected. Then removing the nodes in one of the disconnected parts from $F'$ results in a new LETS sub-structure 
of ${\cal S}'$, which is larger than ${\cal S}$. This contradicts the assumption that ${\cal S}$ is the largest LETS sub-structure of ${\cal S}'$.
For (iii), again by contradiction, if such a node exists, by removing  all the other nodes in $F'\setminus F$,
we obtain a LETS sub-structure of ${\cal S}'$ larger than ${\cal S}$.
\end{proof} 

\begin{pro}
\label{lem:main}
Suppose that $\mathcal{S}'=(F',E')$ is a prime structure of  dot expansion, but is not a simple cycle.  Let $\mathcal{S}=(F,E)$ be the largest LETS sub-structure of ${\cal S}'$. 
Structure $\mathcal{S}'$ can then be generated by expanding $\mathcal{S}$ using a  path or a  lollipop walk.
\end{pro}
\begin{proof}
Based on Lemma~\ref{lem-as}, in the following, we consider two cases of $|F|=|F'|-2$ and $|F| < |F'|-2$. 

{\em Case (1)}--$|F|=|F'|-2$: Based on Part (ii) of Lemma~\ref{lem-as}, the two nodes in $F'\setminus F$ must be adjacent. Moreover, based on Part (iii) of the same lemma,
each node in $F'\setminus F$ must have one connection to the nodes in $\mathcal{S}$. The only possible configurations satisfying both conditions 
are the expansions of ${\cal S}$ by the closed and open paths of length 3. 

{\em Case (2)}--$|F| < |F'|-2$: We first prove that in this case, the number of the edges connecting the nodes in $F'\setminus F$ to the nodes in $\mathcal{S}$ must be strictly less than 3.
Suppose that the number of such edges is at least $3$. This means that, based on Part (iii) of Lemma~\ref{lem-as}, there are at least three nodes in $F'\setminus F$ each connected with 
one edge to a node in $\mathcal{S}$. Select three such nodes, and call them $v_1, v_2$ and $v_3$. Based on Part (ii) of Lemma~\ref{lem-as}, there is a path between any pair of these nodes
in the subgraph induced by $F'\setminus F$. Find the shortest paths between any pair of these nodes in the subgraph. Without loss of generality, assume that among the three
shortest paths, the one between $v_1$ and $v_2$ has the smallest length. Clearly $v_3$ cannot be on this path. Add the nodes $v_1$, $v_2$ and all the nodes on the shortest path between $v_1$ and $v_2$ 
to ${\cal S}$. This results in a LETS sub-structure of ${\cal S}'$ that is larger than ${\cal S}$ and strictly smaller than ${\cal S}'$, which is a contradiction.
Therefore, the number of the edges connecting the nodes in $F'\setminus F$ to the nodes in $\mathcal{S}$ must be either two or one.

For the case where there are two edges connecting $F'\setminus F$ to $\mathcal{S}$, based on Part (iii) of Lemma~\ref{lem-as}, these edges must be incident 
to two distinct nodes in $F' \setminus F$. Call these nodes $v_1$ and $v_2$. There must be no cycles in the subgraph induced by $F'\setminus F$. 
Otherwise, one can find the shortest path between $v_1$ and $v_2$ in the subgraph and by removing all the nodes in $F'\setminus F$ except $v_1$, $v_2$ 
and all the nodes on the shortest path between them, obtain a LETS sub-structure of ${\cal S}'$ that is larger than ${\cal S}$ but strictly smaller than 
${\cal S}'$. Therefore, for this case, the only possible configurations are the expansions of ${\cal S}$ by closed and open paths of length more than 3.

For the case where there is only one edge connecting $F'\setminus F$ to ${\cal S}$, for ${\cal S}'$ to be a LETS structure, there must exist at least one
cycle in the subgraph induced by $F' \setminus F$. With a discussion similar to the case of two connecting edges,
one can prove that there must be exactly one cycle in the subgraph. 
For this case, the only possible configurations are the expansions of ${\cal S}$ by lollipop walks $lo^c_m$ with $m \geq 3$ (length more than four) and $c \geq 3$ (cycles longer than three).
\end{proof} 

Based on Proposition \ref{lem:main}, except the simple cycles, each prime structure of dot-based expansion can be generated by applying {\em path} or {\em lollipop} expansions to its largest LETS sub-structure.
This sub-structure is either a simple cycle or not. If not, then Lemmas~\ref{lem:simple} and~\ref{lem:$dot$a-1}, and Proposition~\ref{lem:main} show that we can obtain the sub-structure by
applying a combination of the three expansion techniques to a simple cycle. 

\begin{theo}
\label{thm1}
All LETS structures of variable-regular Tanner graphs for any variable degree $d_{\mathrm{v}}$, and in any $(a,b)$ class 
(that are not simple cycles), can be generated by applying a combination of depth-one tree (dot),  path and lollipop expansions
to simple cycles. 
\end{theo}

\begin{proof}
The only LETS structures that are out of the reach of {\em dot} expansions are prime structures discussed in Section \ref{sec: dot-based}. 
These prime structures are either simple cycles or not. Based on Proposition~\ref{lem:main}, and the above discussions, the prime structures 
that are not simple cycles can be generated using the three expansion techniques starting from simple cycles. 
\end{proof}

Compared to the dot-based characterization, Theorem~\ref{thm1} provides a new characterization of LETS structures. 
Instead of using one expansion technique ({\em dot}) and some prime structures (simple cycles plus some other prime structures) to characterize the LETS structures, 
the new characterization uses more expansion techniques ({\em dot}, {\em path} and {\em lollipop}) but has fewer prime structures (only simple cycles).
This new characterization, however,
does not provide us with a road-map to a more efficient search algorithm. The reason is that, if we continue to rely on dot-based search algorithm, we still need
to enumerate all the instances of the required prime structures as the input to the search algorithm. Having a characterization of the prime structures 
as expansions of simple cycles, although helpful in such an enumeration, has no effect on the required prime structures with relatively large $a$ and $b$ values,
and the (huge) number of the instances of such structures that need to be enumerated.
    
To use the new expansions in the characterization such that they translate to a more efficient search algorithm,
we need to use them for generating not only the prime structures but also the other LETS structures of interest.  
This will, consequently, reduce the prime structures that need to be enumerated to only a few short simple cycles.
The following example explains this idea and motivates the {\em dpl} characterization of next subsection.  

\begin{ex}
Consider the $(11,1)$ LETS structure ${\cal S}$ of Fig. \ref{fig:example(11,1)}. The prime sub-structure of $\mathcal{S}$ with respect to dot expansion is $s_8$. 
There are, however, different ways to generate $\mathcal{S}$ starting from $s_3$, using a combination of  dot, path and  lollipop expansion techniques.
One approach is explained in Fig.  \ref{fig:DPLexample1}. In this figure, $F_i$ is the set of nodes in $\mathcal{S}_i$, the LETS sub-structure of $\mathcal{S}$
in the $i$-th step of expansion. 
\begin{figure}[] 
\centering
\includegraphics [width=0.55\textwidth]{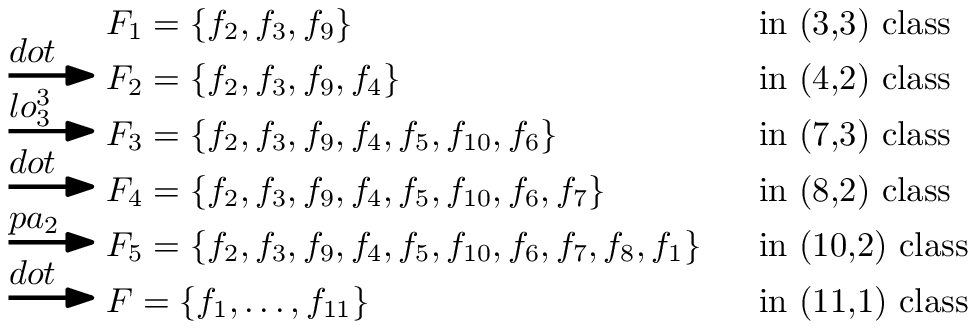}
\caption{Steps of generating the $(11,1)$ LETS structure of Fig.~\ref{fig:example(11,1)} starting from $s_3$. }
\label{fig:DPLexample1}
\end{figure}
A search algorithm based on the sequence of expansions presented in Fig.  \ref{fig:DPLexample1}, will start by enumerating all the instances of $s_3$ in the Tanner graph,
instead of finding the instances of $s_8$, to find all the instances of $\mathcal{S}$. 

As an example, consider the code discussed in Example~\ref{ex15}. This code has 161 and 6,219,941 instances of $s_3$ and $s_8$, respectively. 
In dot-based search algorithm, all the instances of $s_8$ (more than 6 million) need to be enumerated and then {\em dot} expansions need to be applied to each instance of $s_8$ three times successively
to obtain the instances of $\mathcal{S}$. For this code, however, no instance of ${\cal S}$ exists in the Tanner graph. 
On the other hand, following the road-map of Fig.  \ref{fig:DPLexample1}, by applying the {\em dot} expansion to the 161 instances of $s_3$, as the first step, one
obtains no instance of the structure in the $(4,2)$ class, and thus stops the search process for ${\cal S}$ (very quickly). 
 
In general, if all the LETS structures of interest, that have $s_8$ as their prime sub-structure in the {\em dot} characterization, are generated using  
prime structures with smaller size, the instances of $s_8$ are no longer needed to be enumerated in the search algorithm. 
\end{ex}

\subsection{Dpl-based Characterization}
\label{dpl}
Our goal in this subsection is to characterize all the non-isomorphic LETS structures of a variable-regular Tanner graph
with variable degree $d_{\mathrm{v}}$ and girth $g$, within all the $(a,b)$ classes with $a \leq a_{max}$ and $b \leq b_{max}$.
The characterization is based on describing each LETS structure ${\cal S}$ as a hierarchy of embedded LETS structures that starts
from a simple cycle, and expands to ${\cal S}$ in multiple steps, each step involving one of the three expansion techniques,
i.e., {\em dot}, {\em path} or {\em lollipop}. 
The basic idea behind the new characterization, i.e., dpl-based characterization, is to generate as many LETS structures as possible
within the interest range of $a$ and $b$ values by using smaller LETS structures that have $a$ and $b$ values also within the range.
When translated to a search algorithm, this has the benefit that we will only search for instances of the structures that are themselves 
of direct interest to us (have $a$ and $b$ values within the target range). This is unlike the case for the dot-based search algorithm,
where we would search for structures that are not of direct interest (have $a$ or $b$ values outside the target range), 
but happen to be parents of structures of interest in the $dot$ characterization. To achieve the goal of staying within the target 
range of $a$ and $b$ values, and have an exhaustive characterization for all the LETS structures, we need to use {\em path} and {\em lollipop}
expansions in addition to the {\em dot} expansion. For the reasons just explained, one of the important properties of the {\em dpl} characterization
is to generate a LETS structure in each level of expansion. With this constraint, i.e., all the sub-structures in the embedded sequence are LETSs, 
it is easy to see that the proposed characterization is {\em minimal}, in the sense that, none of the expansion steps can be divided into smaller expansions. 
We also prove that the three expansion techniques of $dot$, $path$ and $lollipop$, are the only expansions that are needed in minimal characterizations.

Given a target range of interest $a \leq a_{max}$ and $b \leq b_{max}$, we start from all the simple cycles
within this range. Based on Proposition~\ref{progh}, for a given $a$, a simple cycle of size $a$ has the largest $b$ value of $a(d_{\mathrm{v}}-2)$ among all the 
$(a,b)$ LETS structures.  
We thus initiate the characterization by including simple cycles $s_{g/2}, \ldots, s_{\lfloor b_{max}/(d_\mathrm{v}-2) \rfloor}$.
We then recursively apply {\em dot} expansions to these simple cycles to generate more LETS structures within the range of interest.
(We choose the $dot$ expansion in this step of the algorithm, since it can generate the majority of the LETS structures starting from the simple cycles.)
If at the end of this process, there are still some LETS structures within the range of interest that are not generated, we 
try to generate them using the already generated smaller LETS structures by applying the {\em path} or {\em lollipop} expansions.
Propositions~\ref{pro:pa} and \ref{pro:lo}, are our guide in this process. After the generation of any new LETS structure in this process,
we also recursively apply {\em dot} expansions to it to (possibly) generate more new LETS structures in the range of interest.
If there are still some structures that are not generated, we will increase the range of $b$ values to $b_{max}+1$ to 
include some new LETS structures that can generate the missing LETS structures through expansions. This process of
including LETS structures with larger $b$ values will continue until all the LETS structures of all the $(a,b)$ classes
within the range  $a \leq a_{max}$ and $b \leq b_{max}$, are generated. 

When the process of generating LETS structures of interest, as described above, is completed, for each LETS structure ${\cal S}$
in the range, we have the parent simple cycle $s_j$, for some $j \geq g/2$, and the sequence of expansions that are applied to $s_j$ to generate ${\cal S}$.
We however, recognize that this minimal characterization of ${\cal S}$ may not be unique, in the sense that, there may be other combinations of
parent and expansion sequences that result in the generation of ${\cal S}$. This is explained in the following example.
\begin{ex}
Fig. \ref{fig:minimal_example(7,3)} shows a $(7,3)$ LETS structure $\mathcal{S}=(\{f_1,f_2,\dots,f_{7}\}, E)$ of a variable-regular graph with $d_\mathrm{v}=3$ and $g=6$. 
Two minimal characterizations of ${\cal S}$ starting from the simple cycle $s_4$ in the $(4,4)$ class ($\mathcal{S}'=(\{f_1,f_2,\dots,f_{4}\}, E')$) are presented in Fig. \ref{fig:DPLexample2}.
\begin{figure}[] 
\centering
\includegraphics [width=0.22\textwidth]{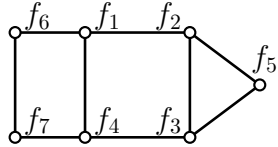}
\caption{ A LETS structure in the $(7,3)$ class of a variable-regular graph with $d_\mathrm{v}=3$ and $g=6$.}
\label{fig:minimal_example(7,3)}
\end{figure}
\begin{figure}[] 
\centering
\includegraphics [width=0.55\textwidth]{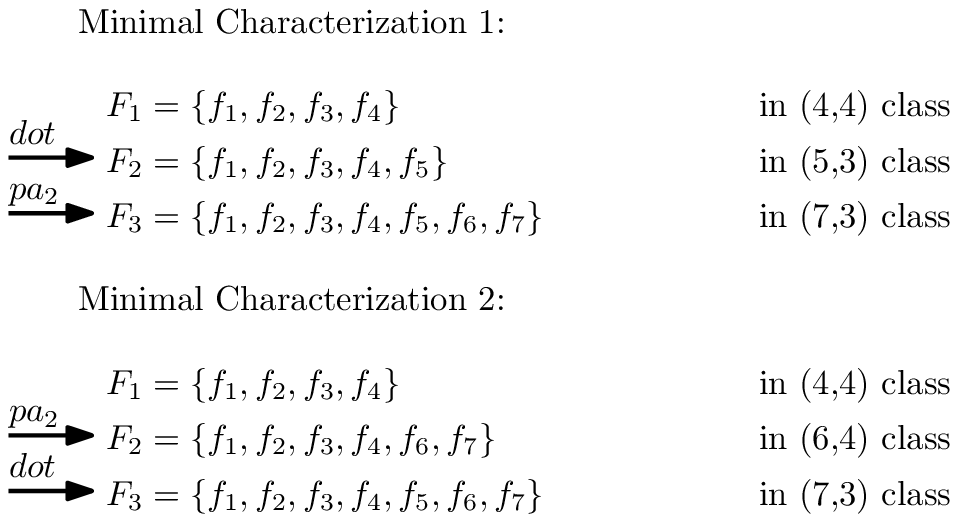}
\caption{Two minimal characterizations of the $(7,3)$ LETS structure of Fig.~\ref{fig:minimal_example(7,3)}, starting from $s_4$. }
\label{fig:DPLexample2}
\end{figure}
\label{exfy}
\end{ex}

To present the $dpl$ characterization
in a compact form, instead of providing the above sequence of parent/expansions for each LETS structure, for each $(a,b)$ class of interest,
we only provide the dpl-prime structures (simple cycles) that are required to obtain the LETS structures in that class, plus the
expansions that need to be applied to the structures within that class to obtain larger LETS structures within the range of interest. 
This presentation of characterization is particularly helpful for the dpl-based search algorithm. 

The pseudo-code of the {\em dpl} characterization algorithm is given in Algorithm \ref{alg3}.
In this algorithm, the table $\mathcal{EX}$ is used to store the sets of required expansion techniques in different classes. 
The notation $\mathcal{EX}_{(a,b)}$ is used to denote the expansion sets for the $(a,b)$ class. For example, $\mathcal{EX}_{(a,b)} = \{dot, pa_2, lo^3_3\}$ 
means that $dot, pa_2$ and $lo^3_3$ expansions will need to be applied to all the non-isomorphic LETS structures in the $(a,b)$ class.

\begin{algorithm}
\caption{{\bf (Dpl-based Characterization Algorithm)} Finds all the non-isomorphic LETS structures in $(a,b)$ classes with $a \leq a_{max}$ and $b \leq b_{max}$, for 
variable-regular Tanner graphs with variable degree $d_{\mathrm{v}}$ and girth $g$. The output parameter $K$ is the size of the largest simple cycle required in the characterization.
The LETS structures generated from the simple cycle $s_k$ are stored in $\mathcal{L}_k$, for $k = g/2, \ldots, K$, and the set of expansions 
for different classes are stored in $\mathcal{EX}$.}
\label{alg3}
\begin{algorithmic} [1] 
\State \textbf{Inputs:} $a_{max}$, $b_{max}$, $d_\mathrm{v}$ and $g$.
\State  \textbf{Initializations:}  $\mathcal{L} \gets \emptyset$,  $a=g/2$, $b'_{max}=b_{max}$, $\mathcal{EX}_{(a,b)} \gets $dot$,~\forall a \in \{g/2,\dots,a_{max}-1\},~\forall b \in \{2,\dots, b'_{max}\}$.
\State  $K=\lfloor b_{max}/(d_\mathrm{v}-2) \rfloor$.
\For {$k = g/2, \dots, K$} \label{a3}
\State  $\mathcal{L}_k$= \textbf{DotExpansion}$(s_k, \mathcal{L}, a_{max})$. \label{alg:first dot}
\State  $\mathcal{L}= \mathcal{L}\cup \mathcal{L}_k$.
\EndFor \label{a5}
\While{$a \leq a_{max}$} \label{rout:shorwhile}
\parIf {all the LETS structures in classes of size $a$  are not found (check by comparing the structures of size $a$ in $\mathcal{L}$ with those generated by {\em geng})}
\parFor {any incomplete $(a,b)$ class, where $1 \leq b \leq b_{max}$}
\For {$k = g/2, \dots, K$}
\State  $(\mathcal{L}_{tem},\mathcal{EX}^{tem})$=\textbf{PaLoExpan}$(a,b, \mathcal{L}_k, \mathcal{L})$. \label{alg:first palo}
\State  $\mathcal{L}'_{tem}$= \textbf{DotExpansion}$(\mathcal{L}_{tem},\mathcal{L},a_{max})$. \label{alg:second dot}
\State  $\mathcal{L}= \mathcal{L} \cup \mathcal{L}'_{tem}$, $\mathcal{L}_k= \mathcal{L}_k \cup \mathcal{L}'_{tem}$.
\State  $\mathcal{EX}= \mathcal{EX} \cup \mathcal{EX}^{tem}$.
\EndFor
\EndparFor
\EndparIf
\State  $a=a+1$.\label{a'=a'+1}
\EndWhile \label{rout:endwhile}
\algstore{bkbreak}
\end{algorithmic}
 \end{algorithm}
\begin{algorithm}[]
\begin{algorithmic}[1]
\algrestore{bkbreak}
\While {all the LETS structures in all the classes of interest are not found (check by comparing all the structures in $\mathcal{L}$ with those generated by {\em geng})}\label{rout:shorwhile2}
\State  $b'_{max}=b'_{max}+1$, $a=g/2$.
\While {$a < a_{max}$}
\State  $\mathcal{EX}_{(a,b'_{max})} \gets dot$.
\If{$(a, b'_{max})$ is the class of a simple cycle}
\State  $K=K+1$.
\State  $\mathcal{L}_K$= \textbf{DotExpansion}$(s_K, \mathcal{L}, a_{max})$. 
\State  $\mathcal{L}= \mathcal{L}\cup \mathcal{L}_K$.
\If{$\mathcal{L}_K = \emptyset$}
\State  $\mathcal{EX}_{(a,b'_{max})}= \mathcal{EX}_{(a,b'_{max})} \setminus dot$.
\EndIf
\Else
\For {$k = g/2, \dots, K$}
\State  $(\mathcal{L}_{tem},\mathcal{EX}^{tem})$=\textbf{PaLoExpan}$(a,b'_{max}, \mathcal{L}_k, \mathcal{L})$. \label{rout:forremexam}
\State  $\mathcal{L}'_{tem}$= \textbf{DotExpansion}$(\mathcal{L}_{tem},\mathcal{L},a_{max})$.
\State  $\mathcal{L}= \mathcal{L} \cup \mathcal{L}'_{tem}$, $\mathcal{L}_k= \mathcal{L}_k \cup \mathcal{L}'_{tem}$.
\State  $\mathcal{EX}= \mathcal{EX} \cup \mathcal{EX}^{tem}$.
\EndFor
\EndIf
\State  $a=a+1$.
\EndWhile
\EndWhile\label{rout:endwhile2}
\State  \textbf{Outputs:}  $K$, $\{\mathcal{L}_{g/2}, \dots, \mathcal{L}_{K}\}$, $\mathcal{EX}$.
\end{algorithmic}
 \end{algorithm}
The algorithm presented in Routine \ref{rout:PaLoExp} is responsible for generating LETS structures that are out of the reach of {\em dot} expansions,
by using {\em path} and {\em lollipop} expansions.
\begin{routine}
 \caption{{\bf (PaLoExpansion)} Finding new LETS structures in the $(a,b)$ class which are out of the reach of {\em dot} expansion by applying {\em path} and {\em lollipop} expansions 
to LETS structures already found in ${\cal L}_k$, excluding the already found structures ${\cal L}$, 
and storing all the new structures in $\mathcal{L}_{tem}$. The expansions corresponding to the new structures are stored in $\mathcal{EX}^{tem}$. $(\mathcal{L}_{tem},\mathcal{EX}^{tem})$= PaLoExpansion $(a,b, \mathcal{L}_k, \mathcal{L})$}
\label{rout:PaLoExp}
 \begin{algorithmic} [1] 
\State \textbf{Initialization:}  $\mathcal{EX}^{tem} \gets \emptyset$, $\mathcal{L}_{tem} \gets \emptyset$.
\State $a'=a-2$.
\While {$a' \geq g/2$}
\State $m=a-a'$. 
\State $\mathcal{L}^p_{tem} \gets \emptyset$, $\mathcal{L}_{c,tem} \gets \emptyset, 3 \leq c \leq m$.
\parFor { any LETS structure $\mathcal{S}$ in the  $(a', b'= b+2-m(d_\mathrm{v}-2))$ class of $\mathcal{L}_k$}
\State $\mathcal{L}^p$ = \textbf{PathExpansion}$(\mathcal{S},m)$.
\State $\mathcal{L}^p_{tem}= \mathcal{L}^p_{tem} \cup \mathcal{L}^p$.
\For {any possible $c$ }
\State $\mathcal{L}^l_c$ = \textbf{LollipopExpansion}$(\mathcal{S},m,c)$.
\State $\mathcal{L}_{c,tem}= \mathcal{L}_{c,tem} \cup \mathcal{L}^l_c$.
\EndFor
\EndparFor
\If {$\{\mathcal{L}^p_{tem} \setminus \mathcal{L}\} \neq \emptyset$ and $\{\mathcal{L}^p_{tem} \setminus \mathcal{L}_{tem}\} \neq \emptyset$}
\State $\mathcal{L}_{tem}=\mathcal{L}_{tem} \cup \{\mathcal{L}^p_{tem} \setminus \mathcal{L}\}$.
\State $\mathcal{EX}^{tem}_{(a', b')} \gets  \mathcal{EX}^{tem}_{(a', b')}  \cup pa_m $.
\EndIf
\For {any possible $c$ }
\If {$\{\mathcal{L}_{c,tem} \setminus \mathcal{L}\} \neq \emptyset$ and $\{\mathcal{L}_{c,tem} \setminus \mathcal{L}_{tem}\} \neq \emptyset$}
\State $\mathcal{L}_{tem}=\mathcal{L}_{tem} \cup \{\mathcal{L}_{c,tem} \setminus \mathcal{L}\}$.
\State $\mathcal{EX}^{tem}_{(a', b')} \gets  \mathcal{EX}^{tem}_{(a', b')}  \cup lo^c_m$.
\EndIf
\EndFor
\State $a'=a'-1$.
\EndWhile
\State \textbf{Outputs:} $\mathcal{L}_{tem},\mathcal{EX}^{tem}$.
\end{algorithmic}
 \end{routine}

The following lemma proves that minimal characterizations are based only on the three expansions $dot$, $path$ and $lollipop$.

\begin{lem}
Consider a characterization of a LETS structure ${\cal S}$, based on a hierarchy of embedded LETS structures that starts from a simple cycle
and generates ${\cal S}$ through a series of graph expansions. Assume that the characterization is minimal, in the sense that,
none of the expansions can be broken into a sequence of smaller expansions such that the resulting sub-structures are still LETSs. 
Then, any graph expansion in the series, corresponding to a minimal characterization, is either a $dot$, a $path$ or a $lollipop$
expansion.
\label{lemcu}
\end{lem}
\begin{proof}
Consider structures $\mathcal{S}'=(F',E')$ and $\mathcal{S}''=(F'',E'')$ as two successive embedded LETS structures in a minimal characterization of  ${\cal S}$ ($\mathcal{S}''$ 
is generated by applying a graph expansion to $\mathcal{S}'$).
If $|F'|=|F''|-1$, then it is easy to see that $\mathcal{S}''$ is obtained from $\mathcal{S}'$ by a $dot$ expansion. For the case of $|F'|=|F''|-2$,
using discussions similar to those in the proof of Case (1) of Proposition~\ref{lem:main}, one can prove that the only possible expansions to generate 
$\mathcal{S}''$ from $\mathcal{S}'$ are the closed and open paths of length $3$. Finally, for the case of $|F'| < |F''|-2$, following similar steps as those 
 in the proof of Case (2) of Proposition~\ref{lem:main}, we can show that the only possible expansions to generate 
$\mathcal{S}''$ from $\mathcal{S}'$ are closed and open paths of length more than $3$, or lollipop walks $lo^c_m$ with $m \geq 3$ (length more than four) and $c \geq 3$ (cycles longer than three).
\end{proof}

Algorithm~\ref{alg3} obtains minimal characterizations for LETS structures. The following theorem shows that the algorithm 
performs this task efficiently and effectively, by making sure that all the intermediate LETS structures generated in each characterization are within
the target range of $a$ and $b$ values, if possible.

\begin{theo}
For given values of $d_{\mathrm{v}}, g, a_{max}$ and $b_{max}$, Algorithm~\ref{alg3} provides an optimal minimal characterization of
all non-isomorphic $(a,b)$ LETS structures of variable-regular graphs with variable degree $d_{\mathrm{v}}$ and girth $g$, with
$a \leq a_{max}$ and $b \leq b_{max}$, in the sense that, to generate all such structures starting from simple cycles and 
using  dot,  path, and  lollipop expansions, the maximum length $K$ of the required simple cycles, and
the maximum range $b'_{max}$ of the $b$ values of the LETS structures that are generated in the characterization process, are
minimized by Algorithm~\ref{alg3}.
\label{thm-ue}
\end{theo}
\begin{proof}
Assume that Algorithm~\ref{alg3} fails to characterize all the non-isomorphic $(a,b)$ LETS structures with $a \leq a_{max}$ and $b \leq b_{max}$, in Lines \ref{a3}--\ref{rout:endwhile} of the algorithm. 
Consider one such structure ${\cal S}$. The failure of Algorithm~\ref{alg3} to characterize ${\cal S}$ implies, 
based on Lemma~\ref{lemcu} and Propositions ~\ref{pro:dot}, \ref{pro:pa} and \ref{pro:lo}, that there is no minimal characterization of ${\cal S}$
that starts from a simple cycle in the range of interest with all the sub-structures also in the range. 
To obtain a minimal characterization for ${\cal S}$, one has thus no choice other than increasing the range of $b$ values beyond $b_{max}$.
In  Lines \ref{rout:shorwhile2}-\ref{rout:endwhile2}, Algorithm~\ref{alg3} performs this task step by step, increasing the range of $b$ values only by one, at each step,
ensuring that the range increase is the minimum required to obtain a minimal characterization for all the LETS structures in the range.  
\end{proof}

\subsection{Dpl-based Characterization Tables}
\label{dpl-tables}
Given the values of $d_\mathrm{v}$ and $g$, and a target range $a \leq a_{max}$ and $b \leq b_{max}$, Algorithm~\ref{alg3}
can be used to characterize all the LETS structures of variable-regular Tanner graphs in $(a,b)$ classes of interest. For a given graph (code),
the target values $a_{max}$ and $b_{max}$, however, need to be selected such that the classes that have  the main contribution in
the error floor of the code, the so-called {\em dominant} classes, are included in the range. The proper selection of the range 
would depend on the code's rate and block length. In this section, a number of characterization tables for different values of 
$d_\mathrm{v}$, $g$, $a_{max}$ and $b_{max}$ are provided. The information regarding these tables are summarized in Table~\ref{tab:range}.
Table~\ref{tab:range} shows that there are in total 11 characterization tables (Tables~\ref{tab:3,6c1}-\ref{tab:5,6}), covering variable node degrees 3, 4, 5, and girths 
6 and 8. The fifth and sixth rows of Table~\ref{tab:range} show $a_{max}$ and $b_{max}$ for each characterization table, respectively.
Row seven of the table contains the maximum $b$ value of the LETS structures that need to be generated 
for the exhaustive coverage of the classes of interest. This corresponds to $b'_{max}$ in Algorithm~\ref{alg3}.
This row should be compared to  the last row of the table that shows the similar parameter when {\em dot} characterization is used.
The comparison of the two rows demonstrates the advantage of the dpl-based  approach compared to the dot-based approach
in covering the same range of LETS classes but by generating LETS structures with smaller $b$ values.
\begin{ex}
In dot-based (LSS-based) approach, to find all the LETS structures of  variable regular graphs 
with $d_\mathrm{v}=3$ and $g=6$, in the interest range of $a \leq 10$ and $b \leq 3$, all the LETS structures within the range
$a \leq 8,~b \leq 8$, need to be generated. However, in the proposed approach, only the LETS structures with $ b \leq 5$ need to be generated 
to guarantee the generation of all the LETS structures in the interest range. This corresponds to Characterization Table~\ref{tab:3,6c3}.
\label{ex-iu}
\end{ex}
\vspace{-20pt}
Rows eight and nine of Table~\ref{tab:range} contain the simple cycles and the expansion techniques that are used in each characterization table, respectively.
\begin{ex}
Consider the case discussed in Example~\ref{ex-iu}. Table~\ref{tab:range} shows that for this case, to characterize all the LETS structures in the range of interest,
simple cycles of length 3, 4 and 5, and expansions $dot, pa_2, pa_3, lo^3_3, lo^3_4$ and $lo^4_4$ are needed.
\end{ex}  
We note that while the study of the relative harmfulness of different LETS structures is beyond the scope of this work, our experimental results,
as well as those reported in the literature, show that the dominant LETS structures/classes depend not only on $d_\mathrm{v}$ and $g$, 
but also on the rate $R$ and block length $n$ of the code. As a result, the values of interest for $a_{max}$ and $b_{max}$ would 
also depend on all these parameters ($d_\mathrm{v}, g, R$ and $n$). In the fourth row of Table~\ref{tab:range}, we have indicated the ranges of rates and block lengths for which a characterization table
can be useful, i.e., for a code within the specified range of rate and block length, the table is likely to cover the dominant classes of LETSs.
Since the characterization tables are used as guidelines for dpl-based search algorithm to find instances of LETS structures in a given code (graph),
one is interested in selecting the smallest values for $a_{max}$ and $b_{max}$ to minimize the search complexity while ensuring that the dominant LETS structures are 
covered (found) in the search process. While this has been the main motivation for providing Characterization Tables~\ref{tab:3,6c1}-\ref{tab:5,6},
we make no claim that these tables are necessarily optimal in the sense just explained, nor we assert that such universally optimal tables
even exist for codes with the same $d_{\mathrm{v}}, g$, rate and block length. These tables should thus be only treated as suggestions rather than
definite guidelines.   

\begin{ex}
Consider the Characterization Table~\ref{tab:3,6c3}. The information in Table~\ref{tab:range} indicates that Table~\ref{tab:3,6c3} is likely to contain the dominant LETS classes 
of medium-rate codes ($0.3 < R < 0.7$) with short to medium block length ($n < 8000$), with $d_{\mathrm{v}} = 3$, and $g=6$.
\end{ex} 

\begin{table*}[]
\centering
\begin{threeparttable}
\setlength{\tabcolsep}{2 pt}
\caption{A List of (and Some Information about) Characterization Tables Provided in This Paper}
\label{tab:range}
\begin{tabular}{||c|c|c| c| c| c| c| c|c| c| c| c|| }
\cline{1-12}
Table&\ref{tab:3,6c1}&\ref{tab:3,6c2}&\ref{tab:3,6c3}&\ref{tab:3,6c4}&\ref{tab:3,8c1}&\ref{tab:3,8c2}&\ref{tab:4,6c1}&\ref{tab:4,6c2}&\ref{tab:4,6c3}&\ref{tab:4,8}&\ref{tab:5,6}\\
\cline{1-12}
$d_\mathrm{v}$&3&3&3&3&3&3&4&4&4&4&5\\
\cline{1-12}
$g$&6&6&6&6&8&8&6&6&6&8&6\\
\cline{1-12}
Rate\tnote{\textdagger}~-Length\tnote{\textdaggerdbl }~~&h-s&h-s,h-m&m-s,m-m&m-m,m-l&h-l,m-s&m-m,m-l&h-s&h-m,h-l&m-m,m-l&m-m,m-l&h-l,m-s\\
&h-m&h-l,m-s&&l-s,l-m&&l-s&&m-s&l-s&l-s&m-m\\
\cline{1-12}
$a_{max}$&6&8&10&12&10&12&6&8&10&11&9\\
\cline{1-12}
$b_{max}$&3&3&3&5&4&4&4&6&10&10&11\\
\cline{1-12}
$b'_{max} $(dpl)&4&4&5&6&5&6&6&8&12&12&13\\
\cline{1-12}
Cycle Primes&$s_3,s_4$&$s_3,s_4$&$s_3,s_4,s_5$&
$s_3,s_4,s_5,s_6$&$s_4,s_5$&$s_4,s_5,s_6$&
$s_3$&$s_3,s_4$&$s_3,s_4,s_5$&$s_3,s_4$&$s_3,s_4$\\
\cline{1-12}
Expansion&&$dot$&$dot,pa_2$&
$dot,pa_2,pa_3$&$dot,pa_2$&$dot,pa_2$&
&$dot$&$dot,pa_2,pa_3$&
$dot,pa_2$&$dot$\\
Techniques&$dot$&$pa_2$&
$pa_3$&
$pa_4,pa_5$&
$pa_3$&$pa_3,pa_4$&$dot$&
$pa_2$&$pa_4,pa_5,lo_3^3$&
$pa_3$&$pa_2$\\
&$pa_2$&$pa_3$&$lo_3^3,lo_4^3$&$lo_3^3,lo_4^3,lo_4^4,lo_5^3$&
$pa_4$&$lo_4^4$&&$pa_3$&
$lo_4^3,lo_4^4,lo_5^3$&$pa_4$&$pa_3$\\
&&$lo^3_3$&$lo_4^4$&$lo_6^3,lo_6^4,lo_6^5,lo_6^6$&$lo_4^4$&$lo_5^4,lo_5^5$&&$lo_3^3$&
$lo_5^4,lo_5^5$&$lo_4^4$&$lo_3^3$\\
\cline{1-12}
$b'_{max}$($dot$)&5&6&8&9&8&9&6&12&16&16&18\\
\cline{1-12}
\end{tabular}
\begin{tablenotes}
\item[\textdagger] For Rate: l = low $(R < 0.3)$, m = medium $(0.3 < R <0.7)$, h = high $(R >0.7)$.
\item[\textdaggerdbl ] For Length: s = short $( n < 2000)$, m = medium $(2000 < n < 8000)$, l = large $(n > 8000)$.
  \end{tablenotes}
\end{threeparttable}
\end{table*}

\begin{table}[]
\centering
\caption{Characterization (cycle prime sub-structures and expansion techniques) of Non-Isomorphic LETS Structures of $(a,b)$ Classes for Variable-Regular Graphs with $d_\mathrm{v}=3$ and $g=6$  for $a \leq a_{max}=6$ and $b \leq b_{max}=3$ }
\label{tab:3,6c1}
\begin{tabular}{||c|c|c| c| c|| }
\cline{1-5}
&$a = 3$&$a = 4$&$a = 5$&$a = 6$\\
\cline{1-5}
b = 0&-&$  \begin{array}{@{}c@{}}s_3(1)\\ \hdashline -\end{array} $&-&$  \begin{array}{@{}c@{}}s_4(2)\\ \hdashline -\end{array} $\\
\cline{1-5}
b = 1&-&-&$  \begin{array}{@{}c@{}}s_3(1)\\ \hdashline -\end{array} $&-\\
\cline{1-5}
b = 2&-&$  \begin{array}{@{}c@{}}s_3(1)\\ \hdashline dot,pa_2\end{array} $&-&$  \begin{array}{@{}c@{}}s_3(1),s_4(3)\\ \hdashline -\end{array} $\\
\cline{1-5}
b = 3&$\begin{array}{@{}c@{}}\boldsymbol{s_3(1)} \\ \hdashline dot \end{array} $&-&$  \begin{array}{@{}c@{}}s_4(2)\\ \hdashline dot\end{array} $& -\\
\cline{1-5}
\hline
\hline
b = 4&-&$  \begin{array}{@{}c@{}}\boldsymbol{s_4(1)}\\ \hdashline dot\end{array} $&-&$  \begin{array}{@{}c@{}}\\ \hdashline -\end{array} $\\
\cline{1-5}
\end{tabular}
\end{table}

 \begin{table}[]
\centering
\caption{Characterization (cycle prime sub-structures and expansion techniques) of Non-Isomorphic LETS Structures of $(a,b)$ Classes for Variable-Regular Graphs with $d_\mathrm{v}=3$ and $g=6$  for $a \leq a_{max}=8$ and $b \leq b_{max}=3$ }
\label{tab:3,6c2}
\begin{tabular}{||c|c|c| c| c| c| c|| }
\cline{1-7}
&$a = 3$&$a = 4$&$a = 5$&$a = 6$&$a = 7$&$a = 8$\\
\cline{1-7}
b = 0&-&$  \begin{array}{@{}c@{}}s_3(1)\\ \hdashline -\end{array} $&-&$  \begin{array}{@{}c@{}}s_4(2)\\ \hdashline -\end{array} $& -& $  \begin{array}{@{}c@{}}s_3(3),s_4(2)\\ \hdashline -\end{array} $\\
\cline{1-7}
b = 1&-&-&$  \begin{array}{@{}c@{}}s_3(1)\\ \hdashline lo^3_3\end{array} $&-&$  \begin{array}{@{}c@{}}s_3(3),s_4(1)\\ \hdashline -\end{array} $&-\\
\cline{1-7}
b = 2&-&$  \begin{array}{@{}c@{}}s_3(1)\\ \hdashline dot,pa_2,pa_3\end{array} $&-&$  \begin{array}{@{}c@{}}s_3(1),s_4(3)\\ \hdashline dot,pa_2\end{array} $& -& $  \begin{array}{@{}c@{}}s_3(11),s_4(8)\\ \hdashline -\end{array} $\\
\cline{1-7}
b = 3&$\begin{array}{@{}c@{}}\boldsymbol{s_3(1)} \\ \hdashline dot, pa_3,lo^3_3 \end{array} $&-&$  \begin{array}{@{}c@{}}s_4(2)\\ \hdashline dot\end{array} $& -& $  \begin{array}{@{}c@{}}s_3(6),s_4(4)\\ \hdashline dot\end{array} $&-\\
\cline{1-7}
\hline
\hline
b = 4&-&$  \begin{array}{@{}c@{}}\boldsymbol{s_4(1)}\\ \hdashline dot,pa_2\end{array} $&-&$  \begin{array}{@{}c@{}}\\ \hdashline dot\end{array} $& -& $  \begin{array}{@{}c@{}}\\ \hdashline -\end{array} $\\
\cline{1-7}
\end{tabular}
\end{table}

In each characterization table, columns and rows correspond to different values of $a$ and $b$, respectively. For each $(a,b)$ class of LETSs, the top entries in the table (separated by a dashed line from the bottom entries)
show the prime simple cycles that are parents of the LETS structures within that class, and the multiplicity of non-isomorphic structures in the class with those parents (multiplicity is given within brackets).    
The bottom entries show the expansion techniques applied to all the LETS structures within the class. 
For example, $s_i(x), s_j(y)$, as upper entries in a class means that there are $x+y$ non-isomorphic LETS structures in that class, $x$ number of such structures are generated from $s_i$ and $y$ of them are generated from $s_j$. 
Having $dot, pa_m, lo^c_m$, as lower entries in a class means that all the possible $dot_m$ expansions, both closed and open {\em path} expansions $(pa^o_m , pa^c_m)$, and $lo^c_m$  are used 
to expand all the LETS structures of that class. Having the symbol ``-", as the only entry for a class means that it is impossible to have LETS structures in that class. Having the symbol ``-" as the 
only bottom entry means that no expansion technique is applied to the structures in that class.
In each characterization table, the simple prime cycles that are needed to generate all the structures within the table are boldfaced.

\begin{ex}
Consider the LETS class $(7,5)$ of Table \ref{tab:3,6c4}. There are 6 non-isomorphic LETS structures in this class. Three, two and one of these structures are generated starting from $s_3, s_4$ and  $s_5$, respectively. 
To obtain all the LETS structures in the table, one needs to apply $dot, pa_2, pa_3$ and $lo_3^3$ expansions to each one of the 6 structures in this class. 
This results in the generation of LETS structures in $(8,2), (8,4), (9,5)$ and $(10,6)$ classes (based on Propositions \ref{pro:dot}, \ref{pro:pa} and \ref{pro:lo}).
\end{ex} 

\begin{table}[!htbp]
\centering
\setlength{\tabcolsep}{2pt}
\caption{Characterization (cycle prime sub-graphs and expansion techniques) of Non-Isomorphic LETS Structures of $(a,b)$ Classes for Variable-Regular Graphs with $d_\mathrm{v}=3$ and $g=6$  for $a \leq a_{max}=10$ and $b \leq b_{max}=3$}
\label{tab:3,6c3}
\begin{tabular}{||c|c| c| c| c| c| c| c| c|| }
\cline{1-9}
&$a = 3$&$a = 4$&$a = 5$&$a = 6$&$a = 7$&$a = 8$&$a = 9$&$a = 10$\\
\cline{1-9}
b = 0&-&$  \begin{array}{@{}c@{}}s_3(1)\\ \hdashline -\end{array} $&-&$  \begin{array}{@{}c@{}}s_4(2)\\ \hdashline -\end{array} $& -& $  \begin{array}{@{}c@{}}s_3(3),s_4(2)\\ \hdashline -\end{array} $&-&$  \begin{array}{@{}c@{}}s_3(13),s_4(5),s_5(1)\\ \hdashline -\end{array} $\\
\cline{1-9}
b = 1&-&-&$  \begin{array}{@{}c@{}}s_3(1)\\ \hdashline lo^3_4,lo^4_4\end{array} $&-&$  \begin{array}{@{}c@{}}s_3(3),s_4(1)\\ \hdashline lo^3_3\end{array} $&-&$  \begin{array}{@{}c@{}}s_3(15),s_4(4)\\ \hdashline - \end{array} $&-\\
\cline{1-9}
b = 2
&-&$  \begin{array}{@{}c@{}}s_3(1)\\ \hdashline dot,pa_2,pa_3\end{array} $
&-&$  \begin{array}{@{}c@{}}s_3(1),s_4(3)\\ \hdashline dot,pa_2,pa_3\end{array} $
&-& $ \begin{array}{@{}c@{}}s_3(12),s_4(7)\\ \hdashline dot,pa_2\end{array} $
&-&$  \begin{array}{@{}c@{}}s_3(75),s_4 (36),s_5(2)\\ \hdashline -\end{array} $\\
\cline{1-9}
b = 3&$  \begin{array}{@{}c@{}}\boldsymbol{s_3(1)}\\ \hdashline dot,pa_3\\lo^3_3,lo^3_4,lo^4_4\end{array} $&-
&$  \begin{array}{@{}c@{}}s_4(2)\\ \hdashline dot,pa_3\end{array} $& -
& $  \begin{array}{@{}c@{}}s_3(6),s_4(4)\\ \hdashline dot,pa_2\end{array} $&-
&$  \begin{array}{@{}c@{}}s_3(40)\\s_4(21),s_5(2)\\ \hdashline dot\end{array} $&-\\
\cline{1-9}
\hline
\hline
b = 4&-&$  \begin{array}{@{}c@{}}\boldsymbol{s_4(1)}\\ \hdashline dot,pa_2\end{array} $&-&$  \begin{array}{@{}c@{}}\\ \hdashline dot,pa_2\end{array} $& -& $  \begin{array}{@{}c@{}}\\ \hdashline dot\end{array} $&-&$  \begin{array}{@{}c@{}}\\ \hdashline -\end{array} $\\
\cline{1-9}
b = 5&-&-&$  \begin{array}{@{}c@{}}\boldsymbol{s_5(1)}\\ \hdashline pa_2\end{array} $& -& $  \begin{array}{@{}c@{}}\\ \hdashline dot\end{array} $&-&$  \begin{array}{@{}c@{}}\\ \hdashline -\end{array} $&-\\
\cline{1-9}
\end{tabular}
\end{table}

\begin{table}[!htbp]
\centering
\begin{threeparttable}[b]
\setlength{\tabcolsep}{.3pt}
\caption{Characterization (cycle prime sub-structures and expansion techniques) of Non-Isomorphic LETS Structures of $(a,b)$ Classes for Variable-Regular Graphs with $d_\mathrm{v}=3$ and $g=6$  for $a \leq a_{max}=12$ and $b \leq b_{max}=5$}
\label{tab:3,6c4}
\begin{tabular}{||c|c| c| c| c| c| c| c| c| c| c|| }
\cline{1-11}
&$a = 3$&$a = 4$&$a = 5$&$a = 6$&$a = 7$&$a = 8$&$a = 9$&$a = 10$&$a = 11$&$a = 12$\\
\cline{1-11}
b = 0&-&$ \begin{array}{@{}c@{}}s_3(1)\\ \hdashline -\end{array} $
&-&$  \begin{array}{@{}c@{}} s_4(2)\\ \hdashline -\end{array} $
& -& $  \begin{array}{@{}c@{}} s_3(3)\\ s_4(2)\\ \hdashline -\end{array} $&-&$  \begin{array}{@{}c@{}} s_3 (13)\\ s_4(5)\\ s_5 (1)\\ \hdashline -\end{array} $&-&$  \begin{array}{@{}c@{}} s_3 (63)\\ s_4 (20)\\s_5(2)\\ \hdashline -\end{array} $\\
\cline{1-11}
b = 1&-&-
&$  \begin{array}{@{}c@{}}s_3(1)\\ \hdashline lo^3_6,lo^4_6\\lo^5_6,lo^6_6\end{array} $&-
&$  \begin{array}{@{}c@{}} s_3(3), s_4(1) \\ \hdashline lo^3_5,lo^4_5,lo^5_5\end{array} $&-&$  \begin{array}{@{}c@{}} s_3(15)\\s_4(4)\\ \hdashline -\end{array} $&-&$  \begin{array}{@{}c@{}} s_3(91)\\ s_4(22)\\ s_5 (1)\\ \hdashline -\end{array} $&-\\
\cline{1-11}
b = 2&-&$ \begin{array}{@{}c@{}}s_3(1)\\ \hdashline dot,pa_2,pa_3\\pa_4,pa_5,lo^5_5\\lo^3_6,lo^4_6,lo^5_6,lo^6_6\end{array}$
&-&$  \begin{array}{@{}c@{}} s_3(1)\\ s_4(3)\\ \hdashline dot,pa_4\\lo^5_5,pa_5\end{array} $
& -& $ \begin{array}{@{}c@{}}s_3(14)\\ s_4(5)\\ \hdashline dot,pa_4\\lo^3_4,lo^4_4 \end{array}$&-&$  \begin{array}{@{}c@{}} s_3 (85)\\ s_4 (27)\\ s_5(1)\\ \hdashline dot\end{array} $&-&$  \begin{array}{@{}c@{}} s_3 (641)\\ s_4 (184)\\ s_5 (10)\\ \hdashline -\end{array} $\\
\cline{1-11}
b = 3&$\begin{array}{@{}c@{}}\boldsymbol{s_3(1)}\\ \hdashline dot,pa_2,pa_3\\pa_4,pa_5 ,lo^3_3\\lo^3_4,lo^4_4,lo^3_5,lo^4_5 \end{array} $&-
&$  \begin{array}{@{}c@{}} s_4(2)\\ \hdashline dot\\lo^4_4,lo^4_5\end{array} $& -
& $  \begin{array}{@{}c@{}} s_3(6), s_4(4)\\ \hdashline dot,pa_4,lo^3_3\\lo^3_4,lo^4_4\end{array} $&-
&$  \begin{array}{@{}c@{}} s_3(44)\\ s_4(17)\\ s_5 (2)\\ \hdashline dot,pa_3,lo^3_3\end{array} $&-&$  \begin{array}{@{}c@{}} s_3 (355)\\ s_4 (120)\\ s_5 (7)\\ \hdashline dot\end{array} $&-\\
\cline{1-11}
b = 4&-&$ \begin{array}{@{}c@{}}\boldsymbol{s_4(1)}\\ \hdashline dot,pa_2\\pa_3,pa_4\end{array}$
&-& $  \begin{array}{@{}c@{}} s_3(2)\\ s_4(2)\\ \hdashline dot,pa_2\\pa_3,pa_4 \end{array} $
& -&$  \begin{array}{@{}c@{}} s_3(14)\\s_4(10)\\ s_5(1)\\ \hdashline dot\\pa_2,pa_3\end{array} $&-&$  \begin{array}{@{}c@{}} s_3 (129)\\ s_4 (63)\\ s_5 (6)\\ \hdashline dot,pa_2\end{array} $&-&$  \begin{array}{@{}c@{}} s_3 (1315)\\s_4 (524)\\ s_5 (52)\\s_6(1) \\ \hdashline -\end{array} $\\
\cline{1-11}
b = 5&-&-
&$  \begin{array}{@{}c@{}}\boldsymbol{ s_5(1)}\\ \hdashline dot,pa_2\\pa_3\end{array} $& -
& $  \begin{array}{@{}c@{}} s_3(3)\\s_4(2), s_5(1)\\ \hdashline dot,pa_2\\pa_3,lo^3_3\end{array} $&-
&$  \begin{array}{@{}c@{}} s_3 (30)\\s_4 (18)\\s_5(4)\\ \hdashline dot,pa_2\end{array} $&-&$  \begin{array}{@{}c@{}} s_3 (328)\\ s_4 (180)\\ s_5 (27),s_6(1)\\ \hdashline dot\end{array} $&-\\
\cline{1-11}
\hline
\hline
b = 6&-&-&-& $  \begin{array}{@{}c@{}}\boldsymbol{ s_6(1)}\\ \hdashline dot,pa_2 \end{array} $& -&$  \begin{array}{@{}c@{}}\\ \hdashline dot,pa_2\end{array} $&-&$  \begin{array}{@{}c@{}}\\ \hdashline dot\end{array} $&-&$  \begin{array}{ c } \\ \hdashline -\end{array} $\\
\cline{1-11}
\end{tabular}
\end{threeparttable}
\end{table}

As discussed in Subsection~\ref{dpl}, to cover a given range of $a$ and $b$ values exhaustively, it is sometimes required to generate LETS structures with $b$ values
larger than $b_{max}$ and up to some $b'_{max} > b_{max}$. In the characterization tables, the LETS classes with $b > b_{max}$ are separated from the rest
by a double-line. For classes with $b > b_{max}$, the multiplicity of non-isomorphic LETS structures are not reported in the tables. 
 
\begin{rem}
Note that in Lines \ref{rout:shorwhile2}-\ref{rout:endwhile2} of Algorithm \ref{alg3}, if an expansion technique does not generate a new LETS structure within the range of interest, it will not be added to $\mathcal{EX}$.
For example, consider the characterization of the non-isomorphic LETS structures of $(a,b)$ classes for variable-regular graphs with $d_\mathrm{v}=3$ and $g=6$, in the range $a \leq a_{max}=10$ and $b \leq b_{max}=3$ (Table \ref{tab:3,6c3}). In Line \ref{rout:forremexam} of Algorithm \ref{alg3}, $pa_4$ is applied to the simple cycle $s_3$ in the $(3,3)$ class and generates one new LETS structure in the $(7,5)$ class. However, the new LETS structure in the $(7,5)$ class does not generate any new non-isomorphic LETS structure
in the interest range of $a \leq 10$ and $b \leq 3$. Therefore, $pa_4$ expansion can be removed from the list of expansion techniques applied to $s_3$ in the $(3,3)$ class. 
\end{rem}

\begin{table}[!htbp]
\centering
\caption{  Characterization (cycle prime sub-structures and expansion techniques) of Non-Isomorphic LETS Structures of $(a,b)$ Classes for Variable-Regular Graphs with $d_\mathrm{v}=3$ and $g=8$ for $a \leq a_{max}=10$ and $b \leq b_{max}=4$}
\label{tab:3,8c1}
\begin{tabular}{||c| c| c| c| c| c| c| c|| }
\cline{1-8}
&$a = 4$&$a = 5$&$a = 6$&$a = 7$&$a = 8$&$a = 9$&$a = 10$\\
\cline{1-8}
b = 0&-&-&$ \begin{array}{@{}c@{}}s_4(1)\\ \hdashline -\end{array}  $& -& $ \begin{array}{c }s_4(2)\\ \hdashline -\end{array}  $& -&  $ \begin{array}{@{}c@{}}s_4 (5), s_5(1)\\ \hdashline -\end{array}  $\\
\cline{1-8}
b = 1&-&-&-&$ \begin{array}{@{}c@{}}s_4(1)\\ \hdashline -\end{array}  $&-&$ \begin{array}{c }s_4(4)\\ \hdashline -\end{array}  $&-\\
\cline{1-8}
b = 2&-&-&$ \begin{array}{@{}c@{}}s_4(1)\\ \hdashline dot, pa_3,pa_4\end{array}  $& -& $ \begin{array}{c }s_4(5)\\ \hdashline dot\end{array}  $&-&$ \begin{array}{@{}c@{}}s_4 (27), s_5(1)\\ \hdashline -\end{array}  $\\
\cline{1-8}
b = 3&-&$ \begin{array}{@{}c@{}}s_4(1)\\ \hdashline dot,pa_4,lo^4_4\end{array}  $& -& $ \begin{array}{@{}c@{}}s_4(3)\\ \hdashline dot\end{array}  $&-&$ \begin{array}{@{}c@{}}s_4 (16),s_5(1)\\ \hdashline dot\end{array}  $&-\\
\cline{1-8}
b = 4&$ \begin{array}{@{}c@{}}\boldsymbol{s_4(1)}\\ \hdashline dot ,pa_2, pa_3\end{array}  $&-& $ \begin{array}{@{}c@{}}s_4(2)\\ \hdashline dot ,pa_2, pa_3\end{array}  $& -&$ \begin{array}{@{}c@{}}s_4(9),s_5(1)\\ \hdashline dot ,pa_2 \end{array}  $&-& $ \begin{array}{@{}c@{}}s_4(57), s_5(6)\\ \hdashline - \end{array}  $\\
\cline{1-8}
\hline
\hline
b = 5&-&$ \begin{array}{@{}c@{}}\boldsymbol{s_5(1)}\\ \hdashline pa_2\end{array}  $& -& $ \begin{array}{@{}c@{}}\\ \hdashline  dot,pa_2\end{array}  $&-&$ \begin{array}{@{}c@{}}\\ \hdashline dot\end{array}  $&-\\
\cline{1-8}
\end{tabular}
\end{table}

\begin{rem}
In Lines \ref{rout:shorwhile2}-\ref{rout:endwhile2} of Algorithm \ref{alg3}, if there are still some LETS structures within the range of interest that are not generated, 
the range of $b$ values will be increased to include some new LETS structures that can generate the missing LETS structures through expansions.
We note that, in this process, some LETS structures, which are already generated in Lines \ref{rout:shorwhile}-\ref{rout:endwhile} of Algorithm \ref{alg3}, may be regenerated. 
In such cases, some path or  lollipop expansions, which were added in Lines \ref{rout:shorwhile}-\ref{rout:endwhile} of Algorithm \ref{alg3}, could be removed. 
For example, consider the characterization of the non-isomorphic LETS structures of $(a,b)$ classes for variable-regular graphs with  $d_\mathrm{v}=3$ and $g=6$, in the range $a \leq a_{max}=6$ and $b \leq b_{max}=3$ (Table \ref{tab:3,6c1}). 
In Lines \ref{rout:shorwhile}-\ref{rout:endwhile} of Algorithm \ref{alg3}, $pa_2$ is applied to the simple cycle $s_3$ in the $(3,3)$ class and generates one new LETS structure in the $(5,3)$ class. 
However, since there are still some structures missing by the time that the algorithm reaches Line \ref{rout:shorwhile2}, the range of $b$ values will be increased to $b'_{max}=4$. 
In Lines \ref{rout:shorwhile2}-\ref{rout:endwhile2} of Algorithm \ref{alg3}, by applying  the dot expansion to the simple cycle $s_4$, both of the non-isomorphic LETS structures in the $(5,3)$ class will be generated. 
Therefore, $pa_2$ expansion can be removed from the list of expansion techniques applied to $s_3$ in the $(3,3)$ class. 
\end{rem}

\begin{table}[!htbp]
\setlength{\tabcolsep}{3 pt}
\centering
\caption{  Characterization (cycle prime sub-structures and expansion techniques) of Non-Isomorphic LETS Structures of $(a,b)$ Classes for Variable-Regular Graphs with $d_\mathrm{v}=3$ and $g=8$ for $a \leq a_{max}=12$ and $b \leq b_{max}=4$}
\label{tab:3,8c2}
\begin{tabular}{||c| c| c| c| c| c| c| c|c| c|| }
\cline{1-10}
&$a = 4$&$a = 5$&$a = 6$&$a = 7$&$a = 8$&$a = 9$&$a = 10$&$a = 11$&$a = 12$\\
\cline{1-10}
b = 0&-&-&$ \begin{array}{@{}c@{}}s_4(1)\\ \hdashline -\end{array}  $& -& $ \begin{array}{@{}c@{}}s_4(2)\\ \hdashline -\end{array}  $& -&  $ \begin{array}{@{}c@{}}s_4 (5), s_5(1)\\ \hdashline -\end{array}  $& - & $ \begin{array}{@{}c@{}}s_4(20), s_5(2)\end{array}  $\\
\cline{1-10}
b = 1&-&-&-&$ \begin{array}{@{}c@{}}s_4(1)\\ \hdashline lo^4_5,lo^5_5\end{array}  $&-&$ \begin{array}{@{}c@{}}s_4(4)\\ \hdashline -\end{array}  $&-&$ \begin{array}{@{}c@{}}s_4(22),s_5(1)\\ \hdashline -\end{array}  $& -\\
\cline{1-10}
b = 2&-&-&$ \begin{array}{@{}c@{}}s_4(1)\\ \hdashline dot, pa_3\\pa_4\end{array}  $& -& $ \begin{array}{@{}c@{}}s_4(5)\\ \hdashline dot,pa_3\\pa_4,lo^4_4\end{array}  $&-&$ \begin{array}{@{}c@{}}s_4 (27), s_5(1)\\ \hdashline dot\end{array}  $& -& $ \begin{array}{@{}c@{}}s_4(179)\\ s_5(11)\\ \hdashline -\end{array}  $\\
\cline{1-10}
b = 3&-&$ \begin{array}{@{}c@{}}s_4(1)\\ \hdashline dot,pa_4\\lo^4_4,lo^4_5\end{array}  $& -& $ \begin{array}{@{}c@{}}s_4(3)\\ \hdashline dot,pa_4\end{array}  $&-&$ \begin{array}{@{}c@{}}s_4 (16)\\s_5(1)\\ \hdashline dot,pa_3\end{array}  $&-&$ \begin{array}{@{}c@{}}s_4 (115)\\ s_5(7)\\ \hdashline dot\end{array}  $& -\\
\cline{1-10}
b = 4&$ \begin{array}{@{}c@{}}\boldsymbol{s_4(1)}\\ \hdashline dot,pa_2 \\ pa_3\end{array}  $&-& $ \begin{array}{@{}c@{}}s_4(2)\\ \hdashline dot ,pa_2 \\ pa_3,lo^4_4\end{array}  $& -&$ \begin{array}{@{}c@{}}s_4(9),s_5(1)\\ \hdashline dot ,pa_2\end{array}  $&-& $ \begin{array}{@{}c@{}}s_4(57), s_5(6)\\ \hdashline dot ,pa_2 \end{array}  $& -&$ \begin{array}{@{}c@{}}s_4 (481)\\ s_5(48) , s_6 (10)\\  \hdashline -\end{array}  $\\
\cline{1-10}
\hline
\hline
b = 5&-&$ \begin{array}{@{}c@{}}\boldsymbol{s_5(1)}\\ \hdashline pa_2\end{array}  $& -& $ \begin{array}{@{}c@{}}\\ \hdashline  dot,pa_2,pa_3\end{array}  $&-&$ \begin{array}{@{}c@{}}\\ \hdashline dot,pa_2\end{array}  $&-&$ \begin{array}{@{}c@{}}\\ \hdashline dot\end{array}  $& -\\
\cline{1-10}
b = 6&-&-& $ \begin{array}{@{}c@{}}\boldsymbol{s_6(1)}\\ \hdashline pa_2\end{array}  $& -&$ \begin{array}{@{}c@{}}\\ \hdashline dot,pa_2 \end{array}  $&-& $ \begin{array}{@{}c@{}}\\ \hdashline dot  \end{array}  $& -&$ \begin{array}{@{}c@{}}\\ \hdashline -\end{array}  $\\
\cline{1-10}
\end{tabular}
\end{table}
 
\begin{table}[!htbp]
\centering
\caption{ Characterization (cycle prime sub-structures and expansion techniques) of Non-Isomorphic LETS Structures of $(a,b)$ Classes for Variable-Regular Graphs with $d_\mathrm{v}=4$ and $g=6$ for $a \leq a_{max}=6$ and $b \leq b_{max}=4$} 
\label{tab:4,6c1}
\begin{tabular}{||c|c| c| c| c|| }
\cline{1-5}
&$a = 3$&$a = 4$&$a = 5$&$a = 6$\\
\cline{1-5}
b = 0&-&-&$  \begin{array}{@{}c@{}}s_3(1)\\ \hdashline -\end{array}  $&$  \begin{array}{@{}c@{}}s_3(1)\\ \hdashline -\end{array}  $ \\
\cline{1-5}
b = 1&-&-&-&-\\
\cline{1-5}
b = 2&-&-&$  \begin{array}{@{}c@{}}s_3(1)\\ \hdashline dot\end{array}  $&$  \begin{array}{@{}c@{}}s_3(3)\\ \hdashline -\end{array}  $ \\
\cline{1-5}
b = 3&-&-&-&-\\
\cline{1-5}
b = 4&-&$  \begin{array}{@{}c@{}}s_3(1)\\ \hdashline dot\end{array}  $&$  \begin{array}{@{}c@{}}s_3(2)\\ \hdashline dot\end{array}  $&$  \begin{array}{@{}c@{}}s_3(7)\\ \hdashline -\end{array}  $\\
\cline{1-5}
\hline
\hline
b = 5&-&-&-&-\\
\cline{1-5}
b = 6&$  \begin{array}{@{}c@{}}\boldsymbol{s_3(1)}\\ \hdashline dot\end{array}  $&$  \begin{array}{@{}c@{}}\\ \hdashline dot\end{array}  $&$  \begin{array}{@{}c@{}}\\ \hdashline dot\end{array}  $&$  \begin{array}{@{}c@{}}\\ \hdashline - \end{array}  $\\
\cline{1-5}
\end{tabular}
\end{table}

\begin{table}[!htbp]
\centering
\caption{ Characterization (cycle prime sub-structures and expansion techniques) of Non-Isomorphic LETS Structures of $(a,b)$ Classes for Variable-Regular Graphs with $d_\mathrm{v}=4$ and $g=6$ for $a \leq a_{max}=8$ and $b \leq b_{max}=6$} 
\label{tab:4,6c2}
\begin{tabular}{||c|c| c| c| c| c| c|| }
\cline{1-7}
&$a = 3$&$a = 4$&$a = 5$&$a = 6$&$a = 7$&$a = 8$\\
\cline{1-7}
b = 0&-&-&$  \begin{array}{@{}c@{}}s_3(1)\\ \hdashline -\end{array}  $&$  \begin{array}{@{}c@{}}s_3(1)\\ \hdashline -\end{array}  $& $  \begin{array}{@{}c@{}}s_3(2)\\ \hdashline -\end{array}  $& $  \begin{array}{@{}c@{}}s_3(5),s_4(1)\\ \hdashline -\end{array}  $ \\
\cline{1-7}
b = 1&-&-&-&-&-&-\\
\cline{1-7}
b = 2&-&-&$  \begin{array}{@{}c@{}}s_3(1)\\ \hdashline dot,pa_3\end{array}  $&$  \begin{array}{@{}c@{}}s_3(3)\\ \hdashline dot\end{array}  $& $  \begin{array}{@{}c@{}}s_3(9)\\ \hdashline dot\end{array}  $& $  \begin{array}{@{}c@{}}s_3(34),s_4(1) \\ \hdashline - \end{array}  $ \\
\cline{1-7}
b = 3&-&-&-&-&-&-\\
\cline{1-7}
b = 4&-&$  \begin{array}{@{}c@{}}s_3(1)\\ \hdashline dot,lo^3_3\end{array}  $&$  \begin{array}{@{}c@{}}s_3(2)\\ \hdashline dot\end{array}  $&$  \begin{array}{@{}c@{}}s_3(7)\\ \hdashline dot,pa_2\end{array}  $& $  \begin{array}{@{}c@{}}s_3(27),s_4(1)\\ \hdashline dot\end{array}  $& $  \begin{array}{@{}c@{}}s_3(122),s_4(2) \\ \hdashline -\end{array}  $\\
\cline{1-7}
b = 5&-&-&-&-&-&-\\
\cline{1-7}
b = 6&$  \begin{array}{@{}c@{}}\boldsymbol{s_3(1)}\\ \hdashline dot\end{array}  $&$  \begin{array}{@{}c@{}}s_3(1)\\ \hdashline dot,pa_2\end{array}  $&$  \begin{array}{@{}c@{}}s_3(3)\\ \hdashline dot,pa_2\end{array}  $&$  \begin{array}{@{}c@{}}s_3(10),s_4(1)\\ \hdashline dot \end{array}  $& $  \begin{array}{@{}c@{}}s_3(43),s_4(1) \\ \hdashline dot\end{array}  $& $  \begin{array}{@{}c@{}}s_3(226),s_4(5) \\ \hdashline -\end{array}  $\\
\cline{1-7}
\hline
\hline
b = 7&-&-&-&-&-&-\\
\cline{1-7}
b = 8&-&$  \begin{array}{@{}c@{}}\boldsymbol{s_4(1)}\\ \hdashline dot\end{array}  $&$  \begin{array}{@{}c@{}}\\ \hdashline dot\end{array}  $&$  \begin{array}{@{}c@{}}\\ \hdashline dot \end{array}  $& $  \begin{array}{@{}c@{}} \\ \hdashline dot\end{array}  $& $  \begin{array}{@{}c@{}} \\ \hdashline -\end{array}  $\\
\cline{1-7}
\end{tabular}
\end{table}

\begin{table}[!htbp]
\centering
\setlength{\tabcolsep}{2pt}
\caption{ Characterization (cycle prime sub-structures and expansion techniques) of Non-Isomorphic LETS Structures of $(a,b)$ Classes for Variable-Regular Graphs with $d_\mathrm{v}=4$ and $g=6$ for $a \leq a_{max}=10$ and $b \leq b_{max}=10$}
\label{tab:4,6c3}
\begin{tabular}{||c|c| c| c| c| c| c| c| c|| }
\cline{1-9}
&$a = 3$&$a = 4$&$a = 5$&$a = 6$&$a = 7$&$a = 8$&$a = 9$&$a = 10$\\
\cline{1-9}
b = 0&-&-
&$  \begin{array}{@{}c@{}}s_3(1)\\ \hdashline -\end{array}  $
&$  \begin{array}{@{}c@{}}s_3(1)\\ \hdashline -\end{array}  $
& $  \begin{array}{@{}c@{}}s_3(2)\\ \hdashline -\end{array}  $
& $  \begin{array}{@{}c@{}}s_3(5),s_4(1)\\ \hdashline -\end{array}  $
& $  \begin{array}{@{}c@{}}s_3(16)\\ \hdashline -\end{array}  $ 
& $  \begin{array}{@{}c@{}}s_3(57),s_4(2)\\ \hdashline -\end{array}  $\\
\cline{1-9}
b = 1&-&-&-&-&-&-&-&-\\
\cline{1-9}
b = 2&-&-&
$  \begin{array}{@{}c@{}}s_3(1)\\ \hdashline dot,pa_5\end{array}  $&
$  \begin{array}{@{}c@{}}s_3(3)\\ \hdashline dot\end{array}  $
& $  \begin{array}{@{}c@{}}s_3(9)\\ \hdashline dot\end{array}  $
& $  \begin{array}{@{}c@{}}s_3(34),s_4(1) \\ \hdashline dot\end{array}  $
& $  \begin{array}{@{}c@{}}s_3 (152),s_4(2)\\ \hdashline dot\end{array}  $
&  $  \begin{array}{@{}c@{}}s_3(840),s_4(5) \\ \hdashline -\end{array}  $ \\
\cline{1-9}
b = 3&-&-&-&-&-&-&-&-\\
\cline{1-9}
b = 4
&-
&$  \begin{array}{@{}c@{}}s_3(1)\\ \hdashline dot,pa_4\\lo^3_5,lo^4_5,lo^5_5\end{array}  $
&$  \begin{array}{@{}c@{}}s_3(2)\\ \hdashline dot\end{array}  $
&$  \begin{array}{@{}c@{}}s_3(7)\\ \hdashline dot,pa_3\\pa_4,lo^4_4\end{array}  $
& $  \begin{array}{@{}c@{}}s_3(26)\\s_4(1),s_5(1)\\ \hdashline dot\end{array}  $
& $  \begin{array}{@{}c@{}}s_3(122)\\s_4(2) \\ \hdashline dot\end{array}  $
&  $  \begin{array}{@{}c@{}}s_3(656)\\s_4(7) \\ \hdashline dot\end{array}  $
&  $  \begin{array}{@{}c@{}}s_3(4140)\\s_4(33) \\ \hdashline -\end{array}  $ \\
\cline{1-9}
b = 5&-&-&-&-&-&-&-&-\\
\cline{1-9}
b = 6
&$  \begin{array}{@{}c@{}}\boldsymbol{s_3(1)}\\ \hdashline dot,pa_2,pa_3\\lo^3_3,lo^3_4,lo^4_4 \end{array}$ 
&$  \begin{array}{@{}c@{}}s_3(1)\\ \hdashline dot\\pa_3,pa_4\end{array}  $
&$  \begin{array}{@{}c@{}}s_3(3)\\ \hdashline dot,pa_3\\pa_4,lo^3_4\end{array}  $
&$  \begin{array}{@{}c@{}}s_3(10)\\s_4(1)\\ \hdashline dot,pa_2,pa_3 \end{array}  $
& $  \begin{array}{@{}c@{}}s_3(42)\\s_4(2) \\ \hdashline dot,pa_3\end{array}  $
& $  \begin{array}{@{}c@{}}s_3(224)\\s_4(7) \\ \hdashline dot,pa_2\end{array}  $
& $  \begin{array}{@{}c@{}}s_3(1360)\\s_4(19)\\ \hdashline dot \end{array}  $
&  $  \begin{array}{@{}c@{}}s_3(9382)\\s_4(111) \\ \hdashline -\end{array}  $  \\
\cline{1-9}
b = 7&-&-&-&-&-&-&-&-\\
\cline{1-9}
b = 8
&-
&$  \begin{array}{@{}c@{}}\boldsymbol{s_4(1)}\\ \hdashline dot,pa_2\end{array}  $
&$  \begin{array}{@{}c@{}}s_3(2),s_4(1)\\ \hdashline dot,pa_2\\pa_3\end{array}  $
&$  \begin{array}{@{}c@{}}s_3(7)\\s_4(2),s_5(1)\\ \hdashline dot,pa_2 \end{array}  $
& $  \begin{array}{@{}c@{}}s_3(39)\\s_4(4),s_5(1)\\ \hdashline dot,pa_2\end{array}  $
& $  \begin{array}{@{}c@{}}s_3(236)\\s_4(14) \\ \hdashline dot,pa_2\end{array}  $
& $  \begin{array}{@{}c@{}}s_3(1561)\\s_4(52)\\ \hdashline dot \end{array}  $
&  $  \begin{array}{@{}c@{}}s_3(11719)\\s_4(286) \\ \hdashline -\end{array}  $  \\
\cline{1-9}
b = 9&-&-&-&-&-&-&-&-\\
\cline{1-9}
b = 10
&-
&-
&$  \begin{array}{@{}c@{}}\boldsymbol{s_5(1)}\\ \hdashline dot,pa_2\end{array}  $
&$  \begin{array}{@{}c@{}}s_3(3)\\s_4(2)\\ \hdashline dot,pa_2 \end{array}  $
& $  \begin{array}{@{}c@{}}s_3(21)\\s_4(6)\\ \hdashline dot \end{array}  $
& $  \begin{array}{@{}c@{}}s_3(142)\\s_4(19),s_5(1)  \\ \hdashline dot\end{array}  $
& $  \begin{array}{@{}c@{}}s_3(1060)\\s_4(82),s_5(2) \\ \hdashline dot \end{array}  $
&  $  \begin{array}{@{}c@{}}s_3(8767)\\s_4(476),s_5(1) \\ \hdashline -\end{array}  $  \\
\cline{1-9}
\hline
\hline
b = 11&-&-&-&-&-&-&-&-\\
\cline{1-9}
b = 12
&-
&-
&-
&$  \begin{array}{@{}c@{}}\\ \hdashline - \end{array}  $
& $  \begin{array}{@{}c@{}}\\ \hdashline dot \end{array}  $
& $  \begin{array}{@{}c@{}}\\ \hdashline dot\end{array}  $
& $  \begin{array}{@{}c@{}}\\ \hdashline dot \end{array}  $
&  $  \begin{array}{@{}c@{}} \\ \hdashline -\end{array}  $ \\
\cline{1-9}
\end{tabular}
\end{table}

\begin{rem}
In characterization tables for the same values of $d_\mathrm{v}$ and $g$, but different ranges of $a$ and $b$ values, there may be some classes,
for which, in different tables, different expansion techniques are used. For example, for the class of $(6,2)$ LETS structures in Tables \ref{tab:3,6c2} and \ref{tab:3,6c3},
the expansions are different, and are respectively, $\{dot, pa_2\}$ and $\{dot, pa_2, pa_3\}$. Similarly, there may be cases where a LETS structure is generated starting from a different
prime simple cycle. For example, comparison of the entries of the $(8,2)$ class in Tables \ref{tab:3,6c2} and \ref{tab:3,6c3} reveals that such LETS structures exist in this class.
\end{rem}

\begin{table}[!htbp]
\centering
\caption{ Characterization (cycle prime sub-structures and expansion techniques) of Non-Isomorphic LETS Structures of $(a,b)$ Classes for Variable-Regular Graphs with $d_\mathrm{v}=4$ and $g=8$ for $a \leq a_{max}=11$ and $b \leq b_{max}=12$}
\label{tab:4,8}
\begin{tabular}{||c| c| c| c| c| c| c| c| c|| }
\cline{1-9}
&$a = 4$&$a = 5$&$a = 6$&$a = 7$&$a = 8$&$a = 9$&$a = 10$&$a = 11$\\
\cline{1-9}
b = 0&-&-&-& -& $  \begin{array}{@{}c@{}}s_4(1)\\ \hdashline -\end{array}  $& -& $  \begin{array}{@{}c@{}}s_4(2)\\ \hdashline -\end{array}  $ & $  \begin{array}{@{}c@{}}s_4(2)\\ \hdashline -\end{array}  $\\
\cline{1-9}
b = 1&-&-&-&-&-&-&-&-\\
\cline{1-9}
b = 2&-&-&-& -& $  \begin{array}{@{}c@{}}s_4(1)\\ \hdashline dot\end{array}  $& $  \begin{array}{@{}c@{}}s_4(2)\\ \hdashline dot\end{array}  $& $  \begin{array}{@{}c@{}}s_4(5)\\ \hdashline dot\end{array}  $ & $  \begin{array}{@{}c@{}}s_4(19)\\ \hdashline -\end{array}  $\\
\cline{1-9}
b = 3&-&-&-&-&-&-&-&-\\
\cline{1-9}
b = 4&-&-&-& $  \begin{array}{@{}c@{}}s_4(1)\\ \hdashline dot,pa_4,lo_4^4\end{array}  $& $  \begin{array}{@{}c@{}}s_4(2)\\ \hdashline dot\end{array}  $& $  \begin{array}{@{}c@{}}s_4(7)\\ \hdashline dot\end{array}  $& $  \begin{array}{c@{}c@{}}s_4(33)\\ \hdashline dot\end{array}  $ & $  \begin{array}{@{}c@{}}s_4(164)\\ \hdashline -\end{array}  $\\
\cline{1-9}
b = 5&-&-&-&-&-&-&-&-\\
\cline{1-9}
b = 6&-&-&$  \begin{array}{@{}c@{}}s_4(1)\\ \hdashline dot\end{array}  $& $  \begin{array}{@{}c@{}}s_4(1)\\ \hdashline dot,pa_3\end{array}  $& $  \begin{array}{@{}c@{}}s_4(5)\\ \hdashline dot,pa_3\end{array}  $& $  \begin{array}{@{}c@{}}s_4(19)\\ \hdashline dot,pa_2\end{array}  $& $  \begin{array}{@{}c@{}}s_4(111)\\ \hdashline dot\end{array}  $ & $  \begin{array}{@{}c@{}}s_4(706)\\ \hdashline -\end{array}  $\\
\cline{1-9}
b = 7&-&-&-&-&-&-&-&-\\
\cline{1-9}
b = 8&$  \begin{array}{@{}c@{}}\boldsymbol{s_4(1)}\\ \hdashline dot,pa_2,pa_3\end{array}  $&$  \begin{array}{@{}c@{}}s_4(1)\\ \hdashline dot\end{array}  $&$  \begin{array}{@{}c@{}}s_4(2)\\ \hdashline dot,pa_2\end{array}  $& $  \begin{array}{@{}c@{}}s_4(3)\\ \hdashline dot\end{array}  $& $  \begin{array}{@{}c@{}}s_4(14)\\ \hdashline dot,pa_2\end{array}  $& $  \begin{array}{@{}c@{}}s_4(50)\\ \hdashline dot,pa_2\end{array}  $& $  \begin{array}{@{}c@{}}s_4(286)\\ \hdashline dot\end{array}  $ & $  \begin{array}{@{}c@{}}s_4(1902)\\ \hdashline -\end{array}  $\\
\cline{1-9}
b = 9&-&-&-&-&-&-&-&-\\
\cline{1-9}
b = 10&-&$  \begin{array}{@{}c@{}}\boldsymbol{s_5(1)}\\ \hdashline pa_2 \end{array}  $&$  \begin{array}{@{}c@{}}s_4(2)\\ \hdashline dot\end{array}  $& $  \begin{array}{c }s_4(6)\\ \hdashline dot\end{array}  $& $  \begin{array}{@{}c@{}}s_4(19)\\ \hdashline dot\end{array}  $&  $  \begin{array}{@{}c@{}}s_4(82)\\ \hdashline dot\end{array}  $ & $  \begin{array}{@{}c@{} }s_4(475),s_5(1)\\ \hdashline dot\end{array}  $& $  \begin{array}{@{}c@{}}s_4(3223)\\ \hdashline -\end{array}  $\\
\cline{1-9}
\hline
\hline
b = 11&-&-&-&-&-&-&-&-\\
\cline{1-9}
b = 12&-&-&$  \begin{array}{@{}c@{}}\\ \hdashline - \end{array}  $& $  \begin{array}{@{}c@{}}\\ \hdashline dot\end{array}  $& $  \begin{array}{@{}c@{}}\\ \hdashline dot\end{array}  $& $  \begin{array}{@{}c@{}}\\ \hdashline dot\end{array}  $& $  \begin{array}{@{}c@{}}\\ \hdashline dot\end{array}  $ &  $  \begin{array}{@{}c@{}}\\ \hdashline -\end{array}  $\\
\cline{1-9}
\end{tabular}
\end{table}

\begin{rem}
When Algorithm~\ref{alg3} will have to go beyond LETS structures with $b \leq b_{max}$, and to use simple cycles of size larger than
$\lfloor b_{max}/(d_{\mathrm{v}} - 2) \rfloor$, it may happen that the new cycle does not result in any new LETS structure in the range of interest. In such cases, the new cycle is not included as a prime structure in the table.
An example of this can be seen in Table~\ref{tab:4,6c3}, where $s_6$ does not generate any new non-isomorphic LETS structure
in the interest range of $a \leq 10$ and $b \leq 10$, and is thus not added as a prime structure in the table.
\end{rem}

\begin{table}[!htbp]
\centering
\caption{ Characterization (cycle prime sub-structures and expansion techniques) of Non-Isomorphic LETS Structures of $(a,b)$ Classes for Variable-Regular Graphs with $d_\mathrm{v}=5$ and $g=6$ for $a \leq a_{max}=9$ and $b \leq b_{max}=11$}
\label{tab:5,6}
\begin{tabular}{||c|c| c| c| c| c| c |c|| }
\cline{1-8}
&$a = 3$&$a = 4$&$a = 5$&$a = 6$&$a = 7$&$a = 8$&$a = 9$\\
\cline{1-8}
b = 0&-&-&-&$  \begin{array}{@{}c@{}}s_3(1)\\ \hdashline -\end{array}  $& -& $  \begin{array}{@{}c@{}}s_3(3)\\ \hdashline -\end{array}  $&-\\
\cline{1-8}
b = 1&-&-&-&-&$  \begin{array}{@{}c@{}}s_3(1)\\ \hdashline -\end{array}  $&-&$  \begin{array}{@{}c@{}}s_3(28)\\ \hdashline -\end{array}  $\\
\cline{1-8}
b = 2&-&-&-&$  \begin{array}{@{}c@{}}s_3(1)\\ \hdashline dot,pa_3\end{array}  $& -& $  \begin{array}{@{}c@{}}s_3(16)\\ \hdashline dot\end{array}  $&-\\
\cline{1-8}
b = 3&-&-&-&-&$  \begin{array}{@{}c@{}}s_3(6)\\ \hdashline dot\end{array}  $&-&$  \begin{array}{@{}c@{}}s_3(289)\\ \hdashline -\end{array}  $\\
\cline{1-8}
b = 4&-&-&-&$  \begin{array}{@{}c@{}}s_3(2)\\ \hdashline dot,pa_3,lo^3_3\end{array}  $& -& $  \begin{array}{@{}c@{}}s_3(75)\\ \hdashline dot\end{array}  $&-\\
\cline{1-8}
b = 5&-&-&$  \begin{array}{@{}c@{}}s_3(1)\\ \hdashline dot,pa_3,lo^3_3\end{array}  $&-&$  \begin{array}{@{}c@{}}s_3(18)\\ \hdashline dot,pa_2\end{array}  $&-&$  \begin{array}{@{}c@{}}s_3(1356),s_4(1)\\ \hdashline - \end{array}  $\\
\cline{1-8}
b = 6&-&-&-&$  \begin{array}{@{}c@{}}s_3(5)\\ \hdashline dot\end{array}  $& -& $  \begin{array}{@{}c@{}}s_3(223)\\ \hdashline dot\end{array}  $&-\\
\cline{1-8}
b = 7&-&-&$  \begin{array}{@{}c@{}}s_3(1)\\ \hdashline dot,pa_2\end{array}  $&-&$  \begin{array}{@{}c@{}}s_3(37)\\ \hdashline dot,pa_2\end{array}  $&-&$  \begin{array}{@{}c@{}}s_3(3786),s_4(1) \\ \hdashline -\end{array}  $\\
\cline{1-8}
b = 8&-&$  \begin{array}{@{}c@{}}s_3(1)\\ \hdashline dot\end{array}  $&-&$  \begin{array}{@{}c@{}}s_3 (8)\\ \hdashline dot,pa_2\end{array}  $& -&$  \begin{array}{@{}c@{}}s_3(460),s_4(1) \\ \hdashline dot\end{array}  $&-\\
\cline{1-8}
b = 9&$  \begin{array}{@{}c@{}}\boldsymbol{s_3(1)}\\ \hdashline dot,pa_2\end{array}  $&-&$  \begin{array}{@{}c@{}}s_3(2)\\ \hdashline dot,pa_2\end{array}  $&-&$  \begin{array}{@{}c@{}}s_3(62)\\ \hdashline dot\end{array}  $&-&$  \begin{array}{@{}c@{}}s_3(7086),s_4(6)\\ \hdashline -\end{array}  $\\
\cline{1-8}
b = 10&-&$  \begin{array}{@{}c@{}}s_3(1)\\ \hdashline dot\end{array}  $&-&$  \begin{array}{@{}c@{}}s_3(12)\\ \hdashline dot\end{array}  $& -& $  \begin{array}{@{}c@{}}s_3(691), s_4(2)\\ \hdashline dot\end{array}  $&-\\
\cline{1-8}
b = 11&-&-&$  \begin{array}{@{}c@{}}s_3(3)\\ \hdashline dot\end{array}  $&-&$  \begin{array}{@{}c@{}}s_3(81),s_4 (1)\\ \hdashline dot\end{array}  $&-&$  \begin{array}{@{}c@{}}s_3(9527),s_4(19)\\ \hdashline -\end{array}  $\\
\cline{1-8}
\hline
\hline
b = 12&-&$  \begin{array}{@{}c@{}}\boldsymbol{s_4(1)}\\ \hdashline dot\end{array}  $&-&$  \begin{array}{@{}c@{}}\\ \hdashline dot\end{array}  $& -& $  \begin{array}{@{}c@{}}\\ \hdashline dot\end{array}  $&-\\
\cline{1-8}
b = 13&-&-&$  \begin{array}{@{}c@{}}\\ \hdashline dot\end{array}  $&-& $  \begin{array}{@{}c@{}}\\ \hdashline dot\end{array}  $& -&$  \begin{array}{@{}c@{}}\\ \hdashline -\end{array}  $\\
\cline{1-8}
\end{tabular}
\end{table}
 
\subsection{Dpl-based Search Algorithm}
\label{dpl-search}
The characterization of LETS structures described in Subsections~\ref{dpl} and \ref{dpl-tables} corresponds to an efficient search algorithm to find the instances of all $(a,b)$ LETS structures of a 
variable-regular Tanner graph for a given range of $a$ and $b$ values in a guaranteed fashion. The efficiency of the search is implied
by Theorem~\ref{thm-ue}. The pseudo-code of the proposed search algorithm is given in Algorithm \ref{tab:search al2}. 
The algorithm is similar to Algorithm \ref{tab:search al1}, except that, in addition to the {\em dot} expansion, {\em path} and {\em lollipop} expansions are also used
to find the instances of some of the LETS structures. The search algorithm starts by the enumeration of simple cycles of length up to $K$. These are the cycles that are 
identified in the characterization table as the required prime structures. For the sake of completeness, the pseudo-code for cycle enumeration is provided 
in Routine~\ref{tab:search cycle}.  

\begin{algorithm}
\centering
\caption{{\bf (Dpl-based Search Algorithm)} Finds all the instances of $(a,b)$ LETS structures of a variable-regular Tanner graph $G$ with girth $g$, for $a \leq a_{max}$ and $b \leq b_{max}$. The inputs are 
the maximum size $K$ of the simple cycles,  and the table $\mathcal{EX}$ of expansion techniques required in {\em dpl} characterization. The outputs are the sets ${\cal I}_k, k=g/2,\ldots,K$, which contain all the instances 
of LETS structures of size up to $a_{max}$ that are expansions of simple cycles of length $k$ in $G$.}
\label{tab:search al2}
\begin{algorithmic}[1]
\State  \textbf{Inputs:} $K$, $\mathcal{EX}$, $g$, $a_{max}$, $b_{max}$. 
\State  \textbf{Initializations:} $\mathcal{I} \gets \emptyset$.
\For {$k=g/2,\ldots,K$}
\For {$a = k, \ldots, a_{max}$}
\State $\mathcal{I}_k^{a}\gets \emptyset$.
\EndFor
\EndFor
\For {$k = g/2, \dots, K$}
\State $\mathcal{I}_k^k$= \textbf{CycSrch}$(s_k)$,  $a=k$. 
\While {$a < a_{max}$} \label{shorouWhile3}
\State $\mathcal{I}_{tem}^{a+1}$= \textbf{DotSrch}$( \mathcal{I}_k^{a},\mathcal{I})$.
\For {any $pa_m $ in the characterization table}
\State $\mathcal{I}_{tem}^{a+m}$=\textbf{PathSrch}$(\mathcal{I}_k^a, \mathcal{I},  \mathcal{EX}, m)$.
\EndFor
\For {any $lo^c_m$ in the characterization table}
\State $\mathcal{I}'$=\textbf{LolliSrch}$(\mathcal{I}_k^a, \mathcal{I}_c^{c}, \mathcal{I},  \mathcal{EX}, m)$.
\State $\mathcal{I}_{tem}^{a+m}=\mathcal{I}_{tem}^{a+m}\cup \mathcal{I}'$.
\EndFor
\For{$ t=a+1,\dots, a_{max}$}
\State $\mathcal{I}_k^{t}=\mathcal{I}_k^{t} \cup \mathcal{I}_{tem}^{t}$.
\State $\mathcal{I}=\mathcal{I} \cup \mathcal{I}_{tem}^{t}$.
\EndFor
\State $a=a+1$.
\EndWhile \label{EndWhile3}
\State $\mathcal{I}_k=\bigcup\limits _{a=k}^{a_{max}} \mathcal{I}_k^{a}$
\EndFor
\State \textbf{Outputs:} $\{ \mathcal{I}_{g/2},\dots,\mathcal{I}_K\}$.
\end{algorithmic}
\end{algorithm}
\begin{routine}
\centering
\caption{{\bf (CycSrch)} Finds all the instances $\mathcal{I}_k^k$ of simple cycles of length $k$ in a given Tanner graph. $\mathcal{I}_k^k$= CycSrch$(s_k)$  }
\label{tab:search cycle}
\begin{algorithmic}[1]
\State  \textbf{Initializations:} $\mathcal{I}_k^k \gets \emptyset$.
\For {each variable node $v_l$ in the Tanner graph}
\For{each check node $c_i$ in the neighborhood of $v_l$}
\parState {Find all the paths $\mathcal{PA}_{i,l}$ of length $k-1$ in the Tanner graph, starting from $c_i$ that do not contain $v_l$.} 
\vspace{-6pt}
\EndFor
\parFor {any pair of check nodes $c_i$ and $c_j$ in the neighborhood of $v_l$, where $i \neq j$}
\For {any path $pa \in \mathcal{PA}_{i,l}$, and any path $pa' \in \mathcal{PA}_{j,l}$} \label{cyc-comp}
\parIf {the two paths end with the same node and that node is their only common node}
\parState {Let $v$ and $v'$ denote the last variable nodes of $pa$ and $pa'$, respectively.}
\parIf{variable nodes in $pa \setminus v$ and $pa' \setminus v'$ do not have any common check node}
\parState {$\mathcal{S}=v_l \cup$\{set of variable nodes in $pa \cup pa'$\}.}
\State $\mathcal{I}_k^k=\mathcal{I}_k^k \cup \{\mathcal{S}\}$.
\EndparIf
\EndparIf
\EndFor
\EndparFor
\EndFor
\State \textbf{Output:} $\mathcal{I}_k^k$.
\end{algorithmic}
\end{routine}
After the enumeration of all instances of a simple cycle, these instances are expanded in the while loop (Lines \ref{shorouWhile3}-\ref{EndWhile3}) to find instances of other LETS structures up to size $a_{max}$, $\mathcal{I}_k$. 
Notation $\mathcal{I}^a_k$ is used for the set of instances of size $a$ found by starting from the instances of the simple cycle of length $k$,
and $\mathcal{I}$  is the set of instances of all LETSs which are found so far in the algorithm.
Finding the instances of LETS structures using {\em dot}, {\em path} and {\em lollipop} expansion techniques are explained in Routines \ref{tab:search $dot$}, \ref{tab:search $path$} and \ref{tab:search lolli}, respectively.
\begin{routine}
\centering
\caption{{\bf (PathSrch)} Expanding the LETS instances of size $a$ in $\mathcal{I}_k^{a}$ using $pa_m$ to find instances of LETS structures of size $a+m$, excluding the already found instances $\mathcal{I}$,
and storing the rest in $\mathcal{I}'$. Only instances  in $\mathcal{I}_k^{a}$ whose structures in accordance to the expansion table $ \mathcal{EX}$ need to be expanded by $pa_m$ are expanded.  
$\mathcal{I}'$= PathSrch$( \mathcal{I}_k^{a}, \mathcal{I},  \mathcal{EX}, m)$ }
\label{tab:search $path$}
\begin{algorithmic}[1]
\State  \textbf{Initializations:} $\mathcal{I}_{tem} \gets \emptyset$.
\parFor{each LETS instance ${\cal S}$ in $\mathcal{I}_k^a$ that based on $\mathcal{EX}$ requires $pa_m$}
\parFor {each unsatisfied check node $c_k \in \Gamma_{o}{(\mathcal{S})}$}
\parState {Find all the paths $\mathcal{PA}_k$ of length $m$ in the Tanner graph, starting from $c_k$. }
\EndparFor
\parFor {any pair of unsatisfied check nodes $c_k$ and $c_j$ in $\Gamma_{o}{(\mathcal{S})}$, where $k \neq j$}
\parFor {any path $pa \in \mathcal{PA}_k$ and any path $pa' \in \mathcal{PA}_j$}
\parIf{$pa$ and $pa'$ end with the same node and this node is their only common node, and if variable and check nodes of $pa$ and $pa'$ are not in $G(\mathcal{S})$}
\parState {$\mathcal{I}_{tem}\gets \mathcal{I}_{tem}\cup \{\mathcal{S}~\cup$ \{set of variable nodes in $pa \cup pa'$\}\}.}
\EndparIf
\EndparFor
\EndparFor
\EndparFor
\State $\mathcal{I}' \gets \mathcal{I}_{tem}\setminus \mathcal{I}$.
\State \textbf{Output:} $\mathcal{I}'$.
\end{algorithmic}
\end{routine}

\begin{routine}
\centering
\caption{{\bf (LolliSrch)} Expanding the LETS instances of size $a$ in $\mathcal{I}_k^{a}$ using $lo^c_m$ to find instances of LETS structures of size $a+m$, excluding the already found instances $\mathcal{I}$,
and storing the rest in $ \mathcal{I}'$. Only instances  in $\mathcal{I}_k^{a}$ whose structures in accordance to the expansion table $ \mathcal{EX}$ need to be expanded by $lo^c_m$  are expanded. 
The set of instances of simple cycles of length $c$, ${\cal I}_c^c$, is also an input. $ \mathcal{I}'$= LolliSrch $( \mathcal{I}_k^a,\mathcal{I}_c^c, \mathcal{I}, \mathcal{EX}, m)$}
\label{tab:search lolli}
\begin{algorithmic}[1]
\State  \textbf{Initializations:} $\mathcal{I}_{tem} \gets \emptyset$,  $d=m+1-c$.
\parFor{each LETS instance ${\cal S}$ in $\mathcal{I}_k^a$ that based on $\mathcal{EX}$ requires $lo^c_m$}
\parState {Find all the paths ${\cal PA}$ of length $2(d-1)$ in the Tanner graph, starting from check nodes $c'$ in $\Gamma_{o}{(\mathcal{S})}$ that have no common nodes with $G({\cal S})$ other than $c'$. }
\parFor {each structure $\mathcal{S}' \in \mathcal{I}_c^c$, for which $G({\cal S}')$ has no common node with $G({\cal S})$, let $\Gamma_{o}(\mathcal{S}')$  denote the set of unsatisfied check nodes of $G(\mathcal{S}')$ and }
\vspace{-6pt}
\For {each path, $pa \in \mathcal{PA}$} \label{lolli-comp}
\parIf {$pa$ ends with a check node $c'$ in $\Gamma_{o}(\mathcal{S}')$ and if $c'$ is the only common node between $pa$ and $\Gamma_{o}(\mathcal{S}')$}
\parState {$\mathcal{I}_{tem}\gets \mathcal{I}_{tem}\cup \{\mathcal{S}~\cup$ \{set of variable nodes in $pa \cup G(\mathcal{S}'$\}\}.}
\EndparIf
\EndFor
\EndparFor
\EndparFor
\State $\mathcal{I}' \gets \mathcal{I}_{tem}\setminus \mathcal{I}$.
\State \textbf{Output:} $\mathcal{I}'$.
\end{algorithmic}
\end{routine}

\subsection{Complexity of the dpl-based Search Algorithm}
\label{complex}

The complexity of the search algorithm depends, in general, on the multiplicity of different instances of LETS structures and the expansion techniques used in different classes.
The total complexity of the proposed algorithm can be divided into two parts: 1) complexity of finding the instances of prime simple cycle structures, and 
2) complexity of expanding the instances of prime structures to find all the instances of every $(a,b)$ LETS structure within the interest range of $a$ and $b$ values. 
In this subsection, we assume that the LDPC codes are regular.

\subsubsection{Complexity of finding the instances of prime structures}
The proposed dpl-based search algorithm uses Routine \ref{tab:search cycle} to find the instances of simple cycles. 
Based on this approach, the order of complexity for finding the instances of a simple cycle of size $k$ is  $\mathcal{O}(n (d_\mathrm{v} d_\mathrm{c})^{k})$.
This complexity increases linearly, polynomially and exponentially with the block length, $n$, the check and variable degrees, $d_\mathrm{c}, d_\mathrm{v}$, and the size of the cycle, respectively. 
The complexity of finding the instances of a simple cycle thus becomes quickly impractical as the cycle length increases. 

Let $N_k$ denote the multiplicity of the instances of the simple cycle $s_k$ of length $k$ in the Tanner graph. Clearly, 
the complexity of finding the instances of a LETS structure that is an immediate child of $s_k$
is proportional to $N_k$. Moreover, all the instances of the prime cycles need to be stored for further processing throughout the search algorithm. The memory required to store all the instances of $s_k$ is also
proportional to $N_k$. To the best of our knowledge, there is no theoretical result on how $N_k$ scales with $n$, $k$, or the degree distribution of the Tanner graph. 
Empirical results in \cite{mehdi1}, however, suggest that $N_k$ is rather independent of $n$ and increases rather rapidly as the variable and check degrees, or $k$ are increased.
We note that, in comparison with the dot-based (LSS-based) search algorithm, the proposed dpl-based search algorithm requires smaller size prime cycles to be enumerated. This
translates to less computational complexity and memory requirements as will be discussed in details in Section~\ref{sec:numerical}. 

\subsubsection{Complexity of expansions}

Based on the characterization tables, a few expansion techniques (those in $\mathcal{EX}_{(a,b)}$) are used to expand the instances of LETS structures within each $(a,b)$ class. 
The complexity of this part of the algorithm depends on the expansion techniques, and the multiplicity of LETSs in each $(a,b)$ class, $N_{a,b}$. 
To the best of our knowledge, there is no theoretical result on how $N_{a,b}$ changes with $a$ and $b$, or with different code parameters.  

In the following, we discuss the complexity of the three expansion techniques.  
The expansion of an instance of a LETS structure in an $(a,b)$ class using {\em dot} expansion is explained in Routine \ref{tab:search $dot$}. 
For each instance, $bd_\mathrm{c}$ variable nodes should be checked to find the set $\mathcal{V}$ in Routine \ref{tab:search $dot$}.
The complexity for the expansion of all the instances of LETSs in an $(a,b)$ class using {\em dot} is thus $\mathcal{O}(N_{a,b}b^2d_\mathrm{c}^2)$. 
The expansion of an instance of a LETS structure in an $(a,b)$ class using the {\em path} expansion technique is explained in Routine \ref{tab:search $path$}.
Based on the approach of Routine \ref{tab:search $path$}, it is easy to see that the complexity of expanding all the instance of LETS structures in an $(a,b)$ class using $pa_m$ is $\mathcal{O}(N_{a,b}b^2(d_\mathrm{v}d_\mathrm{c})^m)$.
Similarly, the complexity of expanding all the instance of LETS structures in an $(a,b)$ class using $lo^c_m$, based on Routine \ref{tab:search lolli}, is $\mathcal{O}(N_{a,b} N_c b c (d_\mathrm{v} d_\mathrm{c})^d)$,
where $N_c$ is the multiplicity of simple cycles of size $c$ in the Tanner graph, and $d=m+1-c$.

Finally, we note that the memory required to store all the LETS instances of an $(a,b)$ class is $\mathcal{O}(a N_{a,b})$.  
The main advantage of using the dpl-based search algorithm, compared to dot-based search, is to avoid dealing with classes with large $a$ and $b$ values. 
Searching for LETS instances in these classes and expanding them imposes huge computational complexity and memory requirements on the search algorithm. 
Instead, in the dpl-based search algorithm, expansions are used in classes with relatively small $N_{a,b}$ values. 
Our experimental results in Section \ref{sec:numerical} show that the dpl-based search algorithm is significantly faster and requires much less memory compared to the 
$dot$ (LSS)-based search algorithm, for different codes with a wide range of variable degrees, rates, block lengths and girths. 

\section{Numerical Results}
\label{sec:numerical}
We have applied the dot-based and dpl-based search algorithms to a large number of variable-regular random and structured LDPC codes with a wide range of variable degrees, rates and block lengths. These codes are listed in Table \ref{tab:codelist}. 
 For all the runtimes reported in this paper, a desktop computer with 2.4-GHz CPU and 8-GB memory is used.
\begin{table}[]
\centering
\caption{List of LDPC Codes Used in This Paper}
\label{tab:codelist}
\begin{tabular}{||c|c|c|c|c|c|c|c|c|c|c|c|c|| }
\cline{1-13}
Code &$\mathcal{C}_1$&$\mathcal{C}_2$&$\mathcal{C}_3$&$\mathcal{C}_4$&
$\mathcal{C}_5$&$\mathcal{C}_6$&$\mathcal{C}_7$&$\mathcal{C}_8$&
$\mathcal{C}_9$&$\mathcal{C}_{10}$&$\mathcal{C}_{11}$&$\mathcal{C}_{12}$\\
\cline{1-13}
$n$&169&361&529&1057&16383&504&816&1008&4000&20000&155&504\\
\cline{1-13}
$R$&0.78&0.84&0.87&0.77&0.87&0.5&0.5&0.5&0.5&0.5&0.4&0.5\\
\cline{1-13}
$d_\mathrm{v}$&3&3&3&3&3&3&3&3&3&3&3&3\\
\cline{1-13}
$g$&6&6&6&6&6&6&6&6&6&6&8&8\\
\cline{1-13}
Ref.&\cite{fan2000}&\cite{fan2000}&\cite{fan2000}&\cite{mackayencyclopedia}
&\cite{mackayencyclopedia}&\cite{mackayencyclopedia}&\cite{mackayencyclopedia}
&\cite{mackayencyclopedia}&\cite{mackayencyclopedia}&\cite{mackayencyclopedia}
&\cite{michael2001class}&\cite{hu2005regular}\\
\cline{1-13}
\hline
\hline
Code&$\mathcal{C}_{13}$&$\mathcal{C}_{14}$&$\mathcal{C}_{15}$&
$\mathcal{C}_{16}$&$\mathcal{C}_{17}$&$\mathcal{C}_{18}$&
$\mathcal{C}_{19}$&$\mathcal{C}_{20}$&$\mathcal{C}_{21}$&
$\mathcal{C}_{22}$&$\mathcal{C}_{23}$&$\mathcal{C}_{24}$\\
\cline{1-13}
$n$&1008&2640&49&169&289&361&529&4095&16383&4000&8000&816\\
\cline{1-13}
$R$&0.5&0.5&0.49&0.71&0.77&0.8&0.83&0.82&0.87&0.5&0.5&0.5\\
\cline{1-13}
$d_\mathrm{v}$&3&3&4&4&4&4&4&4&4&4&4&5\\
\cline{1-13}
$g$&8&8&6&6&6&6&6&6&6&6&6&6\\
\cline{1-13}
Ref.&\cite{hu2005regular}&\cite{mackayencyclopedia}&\cite{fan2000}&
\cite{fan2000}&\cite{fan2000}&\cite{fan2000}&\cite{fan2000}&
\cite{mackayencyclopedia}&\cite{mackayencyclopedia}&\cite{mackayencyclopedia}&
\cite{mackayencyclopedia}&\cite{mackayencyclopedia}\\
\cline{1-13}
\end{tabular}
\end{table}
Except codes $\mathcal{C}_1$-$\mathcal{C}_3$,  $\mathcal{C}_{11}$, $\mathcal{C}_{14}$ and $\mathcal{C}_{15}$-$\mathcal{C}_{19}$, which are structured codes, the other LDPC codes are all randomly constructed.
For structured codes, their structure is not used to simplify the search, and all the runtimes are based on exhaustive search of the code's Tanner graph without taking into account the automorphisms that exist for structured codes.

Tables \ref{tab:3hi1}-\ref{tab:56} list the multiplicity of instances of LETSs in different classes for these codes, found by the dpl-based search algorithm, starting from the instances of simple cycles.
Each row of a table corresponds to a LETS class, and for each class, the total number of instances of LETSs, EASs and FEASs are listed. We note that for the codes with $d_\mathrm{v}=3$, the multiplicities of LETSs and EASs are the same.
We have also listed the break down of 
the LETS multiplicity of each class based on the prime structure involved in the LETS structure. For example, in Table~\ref{tab:3hi1}, for Code~$\mathcal{C}_1$, and for the $(6,2)$ class,
we notice that there are a total of $142,974$ instances of LETS in this class, from which $54,756$ are children of $s_3$, and the remaining $88,218$ are children of $s_4$.
The symbol ``-" as an entry of an $(a,b)$ class for a specific $s_k$ means that, based on the characterization table, $s_k$ is not a prime sub-structure of any LETS structure in the $(a,b)$ class.
In the last two rows of each table, the runtime of the dpl-based and the dot-based search algorithms is reported for comparison. The symbol ``-" for the runtime of the dot-based search algorithm means 
that the search takes more than one day, or needs more than 8-GB memory to be completed. For fair comparison, both algorithms are run to find all the instances of $(a,b)$ LETS structures within the 
range of interest for $a$ and $b$ values, exhaustively.

Table \ref{tab:3hi1} lists the multiplicity of instances of LETSs and FEASs in all the nonempty classes within the range $a \leq 6$ and $b \leq 3$, for the codes $\mathcal{C}_1$, $\mathcal{C}_2$ and $\mathcal{C}_3$. These are high-rate array-based codes \cite{fan2000} with $d_\mathrm{v} = 3$ and $g = 6$. Absorbing and fully absorbing sets of these structured codes were studied in \cite{dolecek2010analysis}. 
The search algorithm to find the instances of LETSs in the range $a \leq 6$ and $b \leq 3$, is based on the information provided in Table \ref{tab:3,6c1}. 
\begin{table*}[th]
\centering
\caption{Multiplicities of LETS and FEAS structures of  codes $\mathcal{C}_1$, $\mathcal{C}_2$ and $\mathcal{C}_3$ within the range $a \leq 6$ and $b \leq 3$}
\setlength{\tabcolsep}{2pt}
\label{tab:3hi1}
\begin{tabular}{||c|c|c||c|c|||c|c||c|c|||c|c||c|c|| } 
\cline{1-13}
 &\multicolumn{4}{c|||}{$\mathcal{C}_{1}$}& \multicolumn{4}{c|||}{$\mathcal{C}_{2}$}&\multicolumn{4}{c||}{$\mathcal{C}_{3}$}\\
\cline{1-13}
$(a,b)$&\multicolumn{2}{c||}{Primes} &Total &Total&\multicolumn{2}{c||}{Primes}&Total&Total&\multicolumn{2}{c||}{Primes}&Total&Total \\
\cline{2-3}\cline{6-7}\cline{10-11}
class&$s_3$&$s_4$&LETS&FEAS&$s_3$&$s_4$&LETS&FEAS&$s_3$&$s_4$&LETS&FEAS\\
\cline{1-13}
(3,3)&2028&-&2028&0&6498&-&6498&0&11638&-&11638&0\\
\cline{1-13}
(4,2)&3042&-&3042&3042&9747&-&9747&9747&17457&-&17457&17457\\
\cline{1-13}
(5,3)&-&83148&83148&0&-&422370&422370&0&-&942678&942678&0\\
\cline{1-13}
(6,0)&-&3718&3718&3718&-&18411&18411&18411&-&40733&40733&40733\\
\cline{1-13}
(6,2)&54756&88218&142974&142974&292410&458109&750519&750519&663366&1029963&1693329&1693329\\
\cline{1-13}
$dpl$&\multicolumn{4}{c|||}{38 sec.}&\multicolumn{4}{c|||}{196 sec.} &\multicolumn{4}{c||}{424 sec.}\\
\cline{1-13}
$dot$&\multicolumn{4}{c|||}{181 sec.}&\multicolumn{4}{c|||}{1447 sec.} &\multicolumn{4}{c||}{5749 sec.}\\
\cline{1-13}
\end{tabular}
\end{table*}
Since the variable degree is 3, all the  instances of LETSs found by the search algorithm are EASs. Therefore,  FEASs were found by examining the LETSs, and testing them for the definition of a fully elementary absorbing set. 
While it takes the proposed search algorithm 38, 196 and 424 seconds to find all the instances of  LETS structures of Table \ref{tab:3,6c1} for these three codes, respectively, it took the dot-based search algorithm 181, 1447 and 5749 seconds to find the same set of LETSs for the three codes, respectively. 

Table \ref{tab:3hi2} lists the multiplicity of  instances of LETSs and FEASs in all the nonempty classes within the range $a \leq 8$ and $b \leq 3$, for the codes $\mathcal{C}_4$ and $\mathcal{C}_5$. Both codes have $d_\mathrm{v} = 3$ and $g = 6$. 
\begin{table}[]
\centering
\caption{Multiplicities of LETS and FEAS structures of  codes $\mathcal{C}_4$ and $\mathcal{C}_5$ within the range $a \leq 8$ and $b \leq 3$}
\label{tab:3hi2}
\begin{tabular}{||c|c|c||c|c|||c|c||c|c|| }
\cline{1-9}
&\multicolumn{4}{c|||}{$\mathcal{C}_4$}& \multicolumn{4}{c||}{$\mathcal{C}_5$}\\
\cline{1-9}
$(a,b)$ &\multicolumn{2}{c||}{Primes}&Total&Total &\multicolumn{2}{c||}{Primes}&Total&Total \\
\cline{2-3} \cline{6-7}
class& $s_3$&$s_4$&LETS&FEAS& $s_3$&$s_4$&LETS&FEAS\\
\cline{1-9}
(3,3)&2288&-&2288&1740&14291&-&14291&13491\\
\cline{1-9}
(4,2)&296&-&296&267&404&-&404&394\\
\cline{1-9}
(5,1)&29&-&29&29&10&-&10&10\\
\cline{1-9}
(5,3)&-&17320&17320&12847&-&43904&43904&41382\\
\cline{1-9}
(6,2)&650&2435&3085&2813&343&1284&1627&1590\\
\cline{1-9}
(7,1)&246&23&269&269&34&3&37&37\\
\cline{1-9}
(7,3)&146694&89023&235717&175759&156742&93724&250466&235926\\
\cline{1-9}
(8,0)&16&0&16&16&0&0&0&0\\
\cline{1-9}
(8,2)&25452&13658&39110&35524&5913&3084&8997&8824\\
\cline{1-9}
$dpl$&\multicolumn{4}{c|||}{3 min.}&\multicolumn{4}{c||}{27 min.}\\
\cline{1-9}
$dot$&\multicolumn{4}{c|||}{-}&\multicolumn{4}{c||}{-}\\
\cline{1-9}
\end{tabular}
\end{table}
Code ${\mathcal C}_4$ is a high-rate code ($R=0.77$) with short block length ($n=1057$), and code $\mathcal{C}_5$ is a high-rate code ($R=0.87$) with large block length $(n=16,383)$. 
The search to find the  instances of LETSs in the range $a \leq 8$ and $b \leq 3$, is based on the information provided in Table \ref{tab:3,6c2}. 
It takes the proposed search algorithm only about 3 and 27 minutes to find LETSs reported in this table for codes $\mathcal{C}_4$ and $\mathcal{C}_5$, respectively. 
To obtain the results of Table~\ref{tab:3,6c2}, the dpl-based search algorithm only needs to enumerate the instances  of $s_3$ and $s_4$ in the Tanner graph of the codes. The dot-based search algorithm,
on the other hand, needs to enumerate the instances of $s_3$, $s_4$, $s_5$, $s_6$, $c_6$, $n_6$, $s_7$, $c_7$, and $n_7$, to find all the instances of  LETSs in the range $a \leq 8$ and $b \leq 3$, in a guaranteed fashion.
It took the dot-based search algorithm 43 minutes to only find the 16,710,009  instances of  $s_5$, and more than a day to find all the instances of LETSs reported in this table for code $\mathcal{C}_4$. 

The examination of Table \ref{tab:3hi2} shows  that for two classes with the same size, the class with a larger $b$ value has more instances of LETSs. For example, for $\mathcal{C}_5$, 
the  instances of  LETS structures in the $(7,1)$ and $(7,3)$ classes are 37 and 250,466, respectively. 
Following a similar trend, this code probably has millions of instances of LETSs in the $(7,5)$ class, and it would take the search algorithm a long time to find them.  
These trapping sets, however, would not be dominant in the presence of those in the $(7,1)$ and $(7,3)$ classes. This highlights the importance of having a flexible and 
efficient search algorithm that can be easily tailored to different values of  $a_{max}$ and $b_{max}$. 

Multiplicities of instances of LETSs and FEASs in all the nonempty classes within the range $a \leq 10$ and $b \leq 3$, for  Codes $\mathcal{C}_6$, $\mathcal{C}_7$ and $\mathcal{C}_8$  are listed in Table \ref{tab:3mod}.
\begin{table}[]
\centering
\caption{Multiplicities of LETS and FEAS Structures of Codes $\mathcal{C}_6$, $\mathcal{C}_7$ and $\mathcal{C}_8$ within the range $a \leq 10$ and $b \leq 3$}
\label{tab:3mod}
\begin{tabular}{||c|c|c|||c|c|||c|c|| }
\cline{1-7}
 &\multicolumn{2}{c|||}{$\mathcal{C}_6$}& \multicolumn{2}{c|||}{$\mathcal{C}_7$} &\multicolumn{2}{c||}{$\mathcal{C}_8$}\\
\cline{1-7}
$(a,b)$&Total &Total& 
 Total &Total&  Total&Total \\
class&LETS&FEAS&LETS&FEAS&LETS&FEAS\\
\cline{1-7}
(3,3)&169&159&132&126&165&153\\
\cline{1-7}
(4,2)&5&5&3&3&6&6\\
\cline{1-7}
(5,3)&214&180&90&86&100&92\\
\cline{1-7}
(6,2)&20&20&2&2&5&5\\
\cline{1-7}
(7,3)&418&374&110&108&116&110\\
\cline{1-7}
(8,2)&24&23&1&1&3&3\\
\cline{1-7}
(9,1)&1&1&0&0&0&0\\
\cline{1-7}
(9,3)&1127&992&195&166&169&161\\
\cline{1-7}
(10,2)&71&69&15&15&4&4\\
\cline{1-7}
$dpl$&\multicolumn{2}{c|||}{20 sec.}&\multicolumn{2}{c|||}{18 sec.}& \multicolumn{2}{c||}{20 sec.}\\
\cline{1-7}
$dot$&\multicolumn{2}{c|||}{948 sec.}&\multicolumn{2}{c|||}{918 sec.}& \multicolumn{2}{c||}{942 sec.}\\
\cline{1-7}
\end{tabular}
\end{table} 
These are three $(3,6)$-regular LDPC codes, all with $g = 6$, and with short block lengths: 504, 816, and 1008, respectively. 
The search algorithm to find all the instances of LETSs in  the range $a \leq 10$ and $b \leq 3$, is based on the information provided in Table \ref{tab:3,6c3}. 
For each code, all the  instances of  LETSs and FEASs in this range are enumerated in less than 20 seconds. In comparison, the runtime for dot-based search is larger by a factor of about 45.
It is worth mentioning that \cite{zhang2011efficient} only reports the multiplicity of FEAS structures in the $(3,3), (4,2)$ and $(5,3)$ classes of $\mathcal{C}_7$ and $\mathcal{C}_8$. 
The results of \cite{zhang2011efficient} match the results reported in the first three rows of Table \ref{tab:3mod}.

Table \ref{tab:3long1} lists the multiplicity of  instances of LETSs and FEASs in all the nonempty classes within the range $a \leq 12$ and $b \leq 5$, for Codes $\mathcal{C}_7$ and $\mathcal{C}_8$.
\begin{table*}[]
\centering
\begin{threeparttable}
\caption{Multiplicities of LETS and FEAS structures of Codes $\mathcal{C}_7$ and $\mathcal{C}_8$ within the range $a \leq 12$ and $b \leq 5$}
\setlength{\tabcolsep}{2pt}
\label{tab:3long1}
\begin{tabular}{||c|c|c|c|c||c|c|c|||c|c|c|c||c|c|c|| }
\cline{1-15}
 &\multicolumn{7}{c|||}{$\mathcal{C}_7$} &\multicolumn{7}{c||}{$\mathcal{C}_8$}\\
\cline{1-15}
$(a,b)$ &\multicolumn{4}{c||}{Primes}&Total &Total &\cite{kyung2012finding} &\multicolumn{4}{c||}{Primes}&Total&Total &\cite{kyung2012finding}  \\
\cline{2-5}\cline{9-12}
class&$s_3$&$s_4$&$s_5$&$s_6$&LETS&FEAS&FAS&$s_3$&
$s_4$&$s_5$&$s_6$&LEAS&FEAS&FAS\\
\cline{1-15}
(3,3)&132&-&-&-&132&126&126&165&-&-&-&165&153&153\\
\cline{1-15}
(4,2)&3&-&-&-&3&3&3&6&-&-&-&6&6&6\\
\cline{1-15}
(4,4)&-&1491&-&-&1491&1350&1350&-&1252&-&-&1252&1130&1130\\
\cline{1-15}
(5,3)&-&90&-&-&90&86&86&-&100&-&-&100&\textbf{92}&\textbf{8388} \\
\cline{1-15}
(5,5)&-&-&9169&-&9169&7286&7286&-&-&10019&-&10019&8388&8388\\
\cline{1-15}
(6,2)&0&2&-&-&2&2&2&1&4&-&-&5&5&5\\
\cline{1-15}
(6,4)&1084&1379&-&-&2463&\textbf{2235}&\textbf{2236}&1011&874&-&-&1885&\textbf{1666}&\textbf{1668}\\
\cline{1-15}
(7,3)&54&56&-&-&110&108&108&74&42&-&-&116&\textbf{110}&\textbf{111}\\
\cline{1-15}
(7,5)&13372&15957&4077&-&33406&26688&nr \tnote{\textdagger}&13462&12148&4126&-&29736&24876&24876\\
\cline{1-15}
(8,2)&1&0&-&-&1&1&1&2&1&-&-&3&3&3\\
\cline{1-15}
(8,4)&2145&2020&34&-&4199&3769&nr&1793&1136&32&-&2961&2617&nr\\
\cline{1-15}
(9,3)&111&89&0&-&195&166&nr&122&47&-&-&169&161&nr\\
\cline{1-15}
(9,5)&39306&39304&5768&-&84378&67558&nr&35996&23149&4642&-&63787&53596&nr\\
\cline{1-15}
(10,2)&8&7&0&-&15&15&15&3&1&0&-&4&4&4\\
\cline{1-15}
(10,4)&4389&3446&130&-&7965&7277&nr&3163&1609&97&-&4869&4421&nr\\
\cline{1-15}
(11,3)&184&103&3&-&290&270&nr&154&64&1&-&219&209&nr\\
\cline{1-15}
(11,5)&104491&95600&11066&116&211273&168214&nr&78596&47962&7619&59&134236&112213&nr\\
\cline{1-15}
(12,2)&10&5&-&-&15&15&15&5&1&0&-&6&6&6\\
\cline{1-15}
(12,4)&10163&7207&467&1&17838&15970&nr&6140&2669&231&1&9041&8168&nr\\
\cline{1-15}
$dpl$&\multicolumn{7}{c|||}{16 min.} &\multicolumn{7}{c||}{14 min.}\\
\cline{1-15}
$dot$&\multicolumn{7}{c|||}{-} &\multicolumn{7}{c||}{-}\\
\cline{1-15}
\end{tabular}
\begin{tablenotes}
    \item[\textdagger] ``nr" stands for \textit{not reported}.
  \end{tablenotes}
\end{threeparttable}
\end{table*}
It takes the proposed search algorithm 16 and 14 minutes to find all the  instances of LETSs in the range of $a \leq 12$, $b \leq 5$, for Codes $\mathcal{C}_7$ and $\mathcal{C}_8$, respectively. 
This is versus about 20 seconds that takes the algorithm to find all the instances of LETSs of these codes in the range of $a \leq 10$, $b \leq 3$. 
In Table \ref{tab:3long1}, we have also reported the multiplicity of  instances of FASs obtained by the exhaustive search algorithm of \cite{kyung2012finding}. 
There are only four cases in Table \ref{tab:3long1}, where the multiplicity of FEAS classes obtained here differ from the multiplicity of FAS classes reported in \cite{kyung2012finding}. 
These cases are boldfaced in the table. We believe that the source of these differences are either typographical error, or some FASs that are not elementary. In addition, it took the exhaustive search of
\cite{kyung2012finding} more than 5 hours to generate the fully absorbing sets of each of the two codes on a desktop
computer with an Intel Core2 2.4 GHz CPU with 2GB memory \cite{kyung2012finding}.

Table \ref{tab:3long2} lists the multiplicity of instances of LETSs in all the nonempty classes within the range $a \leq 12$ and $b \leq 5$, for Codes $\mathcal{C}_9$ and $\mathcal{C}_{10}$. 
These codes are $(3,6)$-regular, both with $g=6$, and with medium and large block lengths: $4000$ and $20,000$, respectively. The search algorithm to find 
instances of LETSs in the range $a \leq 12$ and $b \leq 5$, is based on the information provided in Table \ref{tab:3,6c4}. For each code, all the instances of LETSs in this range are 
enumerated in less than $9$ minutes. From these examples, it can be seen that the dpl-based search algorithm is able to enumerate the instances of 
dominant LETSs of codes with large block length in a guaranteed fashion, efficiently.  Interestingly, despite the large difference between the block lengths 
of the two codes, the runtimes of the dpl-based search for the two codes are not very different.  In comparison, for the dot-based algorithm to search  the instances of LETS structures of these codes with $a \leq 12$ and $b \leq 5$, in an exhaustive fashion, 
the algorithm requires the enumeration of $s_{10}$ and $s_{11}$ structures. This  imposes huge computational complexity and memory requirements on the dot-based search algorithm (more than a day of runtime and 8-GB memory). 
On the other hand, the dpl-search only needs to enumerate prime cycles up to $s_6$. Moreover, in the dpl-based search, for finding the instances of potentially dominant LETSs, 
most of the expansion techniques with relatively high complexity (large $m$) are needed to be applied to the instances of LETS in classes with small $b$, where such classes happen to be usually empty.
For example, in $\mathcal{C}_{10}$, there is no LETS in the $(4,2)$, $(5,1)$, $(6,2)$, $(7,1)$ and $(8,2)$ classes.  

\begin{table}[]
\centering
\caption{Multiplicities of LETS and FEAS structures of Codes $\mathcal{C}_9$ and $\mathcal{C}_{10}$ within the range $a \leq 12$ and $b \leq 5$}
\label{tab:3long2}
\begin{tabular}{||c|c|c|c|c||c|c|||c|c|c|c||c|c|| }
\cline{1-13}
 &\multicolumn{6}{c|||}{$\mathcal{C}_9$}&\multicolumn{6}{c||}{$\mathcal{C}_{10}$}\\
\cline{1-13}
$(a,b)$ &\multicolumn{4}{c||}{Primes}&Total& Total &\multicolumn{4}{c||}{Primes}& Total& Total \\
\cline{2-5}\cline{8-11}
class&$s_3$&$s_4$&$s_5$&$s_6$&LETS&FEAS&$s_3$&$s_4$&$s_5$&$s_6$&LETS&FEAS\\
\cline{1-13}
(3,3)&171&-&-&-&171&169&161&-&-&-&161&161\\
\cline{1-13}
(4,2)&1&-&-&-&1&1&-&-&-&-&0&0\\
\cline{1-13}
(4,4)&-&1219&-&-&1219&1192&-&1260&-&-&1260&1256\\
\cline{1-13}
(5,3)&-&21&-&-&21&21&-&2&-&-&2&2\\
\cline{1-13}
(5,5)&-&-&9935&-&9935&9526&-&-&10046&-&10046&9952\\
\cline{1-13}
(6,4)&278&198&-&-&476&456&49&46&-&-&95&95\\
\cline{1-13}
(7,3)&5&5&-&-&10&10&0&0&-&-&0&0\\
\cline{1-13}
(7,5)&3708&2911&1042&-&7661&7359&729&609&202&-&1540&1528\\
\cline{1-13}
(8,4)&109&74&0&-&183&173&2&2&0&-&4&4\\
\cline{1-13}
(9,3)&2&2&0&-&4&4&0&0&0&-&0&0\\
\cline{1-13}
(9,5)&2316&1360&325&-&4001&3847&79&71&17&-&167&165\\
\cline{1-13}
(10,4)&45&29&0&-&74&71&1&0&0&-&1&1\\
\cline{1-13}
(11,3)&0&1&0&-&1&1&0&0&0&-&0&0\\
\cline{1-13}
(11,5)&1371&781&123&3&2278&2189&2&6&0&0&8&8\\
\cline{1-13}
(12,4)&26&10&0&0&36&36&0&0&0&0&0&0\\
\cline{1-13}
$dpl$&\multicolumn{6}{c|||}{7 min.} &\multicolumn{6}{c||}{9 min.}\\
\cline{1-13}
$dot$&\multicolumn{6}{c|||}{-} &\multicolumn{6}{c||}{-}\\
\cline{1-13}
\end{tabular}
\end{table}

Multiplicities of instances of LETSs and FEAS in all the nonempty classes within the range $a \leq 12$ and $b \leq 4$, for Codes $\mathcal{C}_{11}$, $\mathcal{C}_{12}$ and $\mathcal{C}_{13}$  are listed in Table \ref{tab:38}. 
These codes have $d_\mathrm{v}=3$ and $g=8$. The search algorithm to find all the instances of LETSs in the range of interest is based on the information provided in Table \ref{tab:3,8c2}.
\begin{table*}[]
\centering
\begin{threeparttable}
\caption{Multiplicities of LETS and FEAS structures of Codes $\mathcal{C}_{11}$, $\mathcal{C}_{12}$ and $\mathcal{C}_{13}$ within the range $a \leq 12$ and $b \leq 4$}
\setlength{\tabcolsep}{1pt}
\label{tab:38}
\begin{tabular}{||c|c|c|c||c|c|c|||c|c|c||c|c|c|||c|c|c||c|c|c|| }
\cline{1-19}
 &\multicolumn{6}{c|||}{$\mathcal{C}_{11}$} &\multicolumn{6}{c|||}{$\mathcal{C}_{12}$}&\multicolumn{6}{c||}{$\mathcal{C}_{13}$}\\
\cline{1-19}
$(a,b)$&\multicolumn{3}{c||}{Primes}&Total &Total& \cite{kyung2012finding}&\multicolumn{3}{c||}{Primes}&Total&Total&\cite{kyung2012finding}&\multicolumn{3}{c||}{Primes}&Total&Total&\cite{kyung2012finding}\\
\cline{2-4}\cline{8-10}\cline{14-16}
class& $s_4$&$s_5$&$s_6$&LETS&FEAS&FAS& $s_4$&$s_5$&$s_6$&LETS&FEAS&FAS& $s_4$&$s_5$&$s_6$&LETS&FEAS&FAS\\
\cline{1-19}
(4,4)&465&-&-&465&0&0&802&-&-&802&760&760&2&-&-&2&2&2\\
\cline{1-19}
(5,3)&155&-&-&155&155&155&14&-&-&14&14&14&0&-&-&0&0&0\\
\cline{1-19}
(6,4)&930&-&-&930&0&0&985&-&-&985&849&849&1&-&-&1&1&1\\
\cline{1-19}
(7,3)&930&-&-&930&0&0&57&-&-&57&47&47&0&-&-&0&0&0\\
\cline{1-19}
(8,2)&465&-&-&465&465&465&5&-&-&5&4&4&0&-&-&0&0&0\\
\cline{1-19}
(8,4)&4650&465&-&5115&1395&1395&2414&215&-&2629&2270&2270&2&47&-&49&46&nr\\
\cline{1-19}
(9,1)&0&-&-&0&0&0&1&-&-&1&1&1&0&-&-&0&0&0\\
\cline{1-19}
(9,3)&1860&0&-&1860&930&930&152&4&-&156&146&146&0&1&-&1&1&nr\\
\cline{1-19}
(10,2)&1395&0&-&1395&1395&1395&6&0&-&6&6&6&0&0&-&0&0&0\\
\cline{1-19}
(10,4)&25575&3720&-&29295&17670&17670&7311&1557&-&8868&7399&nr&0&173&-&173&148&nr\\
\cline{1-19}
(11,3)&6200&0&-&6200&5270&5270&551&77&-&628&577&nr&0&9&-&9&9&nr\\
\cline{1-19}
(12,2)&930&0&-&930&930&930&25&2&-&27&26&26&0&0&-&0&0&0\\
\cline{1-19}
(12,4)&144615&30225&6045&180885&115320&nr\tnote{\textdagger}&24857&7062&630&32549&
26936&nr&1&461&50&512&466&nr \\
\cline{1-19}
$dpl$&\multicolumn{6}{c|||}{442 sec. / 28 sec. \tnote{*}}& \multicolumn{6}{c|||}{561 sec. / 38 sec. }&\multicolumn{6}{c||}{157 sec. / 17 sec. }\\
\cline{1-19}
$dot$&\multicolumn{6}{c|||}{- / 392 sec.}& \multicolumn{6}{c|||}{- / 1563 sec. }&\multicolumn{6}{c||}{- / 1466 sec. }\\
\cline{1-19}
\end{tabular}
\begin{tablenotes}
    \item[\textdagger] ``nr" stands for \textit{not reported}.
\item[*] Runtime to find all the instances of LETS structures reported in this table/runtime to find instances in the first 10 classes reported in this table.
  \end{tablenotes}
\end{threeparttable}
\end{table*}

Code $\mathcal{C}_{11}$ is the Tanner $(155,64)$ code \cite{michael2001class} with $d_\mathrm{c}=5$. The error-prone structures of this code have been widely investigated~\cite{mehdi2012},\cite{zhang2013controlling},\cite{kyung2012finding},\cite{declercq2013finite}, \cite{mehdi2014}. It is well-known that the LETS structures in the  $(8,2)$ and $(10,2)$ classes are the dominant LETS structures of this code. It takes the dpl-based search algorithm 442 seconds 
to find all the instances of LETS structures reported in Table \ref{tab:38} for this code. However, it took the dot-based algorithm 45 minutes to find almost all the instances of LETS structures of $\mathcal{C}_{11}$ within this range (prime structures of size $10$ are not considered), and more than a day to find the exhaustive list of instances of LETS structures in this table.
If we limit the range of search to $a \leq 10$ and $b \leq 4$, it takes the proposed algorithm only 28 seconds to find all the instances of LETS structures of $\mathcal{C}_{11}$ within this range, based on the information provided in Table \ref{tab:3,8c1}. However, it took the dot-based algorithm $392$ seconds to find all the instances of LETS structures of $\mathcal{C}_{11}$ within the range $a \leq 10$ and $b \leq 4$.
   
Codes $\mathcal{C}_{12}$ and $\mathcal{C}_{13}$ are $(3,6)$-regular LDPC codes constructed by the PEG algorithm \cite{hu2005regular}. 
The dpl-based search algorithm finds all the instances of LETSs for codes $\mathcal{C}_{12}$ and $\mathcal{C}_{13}$ within the range $a \leq 12$ and $b \leq 4$, in  561 and 157 seconds, respectively. 
Again, the exhaustive search of  LETSs for both codes within the range $a \leq 12, b \leq 4$, is out of the reach of the dot-based search algorithm.
If we limit the range of search to $a \leq 10$ and $b \leq 4$, it takes the proposed algorithm only 38 and 17 seconds to find all the instances of LETSs of $\mathcal{C}_{12}$ and $\mathcal{C}_{13}$ within this range, respectively.  It, however, takes the dot-based algorithm $1563$ and $1466$ seconds to find all the instances of LETSs of $\mathcal{C}_{12}$ and $\mathcal{C}_{13}$, respectively, within the range $a \leq 10$ and $b \leq 4$. 

In Table \ref{tab:38}, we have also reported the multiplicity of  instances of FASs obtained by the exhaustive search algorithm of \cite{kyung2012finding}. It can be seen that the results of \cite{kyung2012finding} match those obtained here. One should, however, note that while the range of classes reported in Table \ref{tab:38} is much larger than the range of classes reported in \cite{kyung2012finding}, it took the exhaustive
search of \cite{kyung2012finding} more than 5 hours, compared to only a few minutes for the proposed algorithm, to generate the fully
absorbing sets of each of these codes on a desktop computer with an
Intel Core2 2.4 GHz CPU with 2GB memory \cite{kyung2012finding}.

Code $\mathcal{C}_{14}$ is the Margulis $(2640, 1320)$  code~\cite{margulis}, with $g=8$. It is well-known that the most dominant LETS structures of this code are in (12,4) and (14,4) LETS classes. 
The proposed search algorithm finds all the instances of LETSs in Table \ref{tab:marg} for this code in about 42 minutes. It is worth mentioning that \cite{kyung2012finding} only reports the multiplicity of FAS structures in the $(4,4), (5,5)$ classes of $\mathcal{C}_{14}$. 
The results of \cite{kyung2012finding} match the results reported in the first two rows of Table \ref{tab:marg}.

\begin{table}[]
\centering
\caption{Multiplicities of LETS and FEAS structures of  $\mathcal{C}_{14}$ within the range $a \leq 14$ and $b \leq 6$}
\label{tab:marg}
\begin{tabular}{||c|c|c|c||c|c|| }
\cline{1-6}
 &\multicolumn{5}{c||}{$\mathcal{C}_{14}$}\\
\cline{1-6}
$(a,b)$&\multicolumn{3}{c||}{Primes}&Total &Total\\
\cline{2-4}
class& $s_4$&$s_5$&$s_6$&LETS&FEAS\\
\cline{1-6}
(4,4)&1320&-&-&1320&1320\\
\cline{1-6}
(5,5)&-&11088&-&11088&11088\\
\cline{1-6}
(6,6)&-&-&106920&106920&104280\\
\cline{1-6}
(7,5)&5280&2640&-&7920&7920\\
\cline{1-6}
(8,6)&80520&51480&14520&146520&138600\\
\cline{1-6}
(9,5)&0&2640&-&2640&2640\\
\cline{1-6}
(10,6)&100320&39600&7920&147840&129360\\
\cline{1-6}
(11,5)&5280&0&0&5280&0\\
\cline{1-6}
(12,4)&1320&0&0&1320&1320\\
\cline{1-6}
(12,6)&73040&52800&9240&135080&128480\\
\cline{1-6}
(13,5)&2640&0&0&2640&0\\
\cline{1-6}
(14,4)&1320&0&0&1320&1320\\
\cline{1-6}
$dpl$&\multicolumn{5}{c||}{42 min.}\\
\cline{1-6}
$dot$&\multicolumn{5}{c||}{-}\\
\cline{1-6}
\end{tabular}
\end{table}

Table \ref{tab:3hi} lists the multiplicity of instances of LETSs, EASs and FEASs in all the nonempty classes within the range $a \leq 6$ and $b \leq 4$, for Codes $\mathcal{C}_{15}$-$\mathcal{C}_{19}$. 
\begin{table*}[]
\centering
\caption{Multiplicities of LETS, EAS and FEAS structures of Codes  $\mathcal{C}_{15}$, $\mathcal{C}_{16}$, $\mathcal{C}_{17}$, $\mathcal{C}_{18}$, $\mathcal{C}_{19}$ within the range $a \leq 6$ and $b \leq 4$}
\setlength{\tabcolsep}{1pt}
\label{tab:3hi}
\begin{tabular}{||c|c|c|c|||c|c|c|||c|c|c|||c|c|c|||c|c|c|| } 
\cline{1-16}
&\multicolumn{3}{c|||}{$\mathcal{C}_{15}$} &\multicolumn{3}{c|||}{$\mathcal{C}_{16}$}&\multicolumn{3}{c|||}{$\mathcal{C}_{17}$}&\multicolumn{3}{c|||}{$\mathcal{C}_{18}$}&\multicolumn{3}{c||}{$\mathcal{C}_{19}$}\\
\cline{1-16}
$(a,b)$ &Total&Total&Total &Total&Total&Total &Total&Total&Total&Total&Total&Total&Total&Total&Total\\
class&LETS&EAS&FEAS&LETS&EAS&FEAS&LETS&EAS&FEAS&LETS&EAS&FEAS&LETS&EAS&FEAS\\
\cline{1-16}
(4,4)&294&294&0&0&0&0&0&0&0&0&0&0&0&0&0\\
\cline{1-16}
(5,4)&1176&0&0&0&0&0&0&0&0&12996&12996&0&0&0&0\\
\cline{1-16}
(6,2)&588&588&588&0&0&0&0&0&0&0&0&0&0&0&0\\
\cline{1-16}
(6,4)&4116&1764&0&30420&30420&0&23120&23120&18496&110466&32490&25992&58190&58190&46552\\
\cline{1-16}
$dpl$&\multicolumn{3}{c|||}{1 min.}&\multicolumn{3}{c|||}{2 min.}&\multicolumn{3}{c|||}{4 min.}&\multicolumn{3}{c|||}{7 min.}&\multicolumn{3}{c||}{11 min.}\\
\cline{1-16}
$dot$&\multicolumn{3}{c|||}{1 min.}&\multicolumn{3}{c|||}{2 min.}&\multicolumn{3}{c|||}{4 min.}&\multicolumn{3}{c|||}{7 min.}&\multicolumn{3}{c||}{11 min.}\\
\cline{1-16}
\end{tabular}
\end{table*}
These are five high-rate array-based codes with $g=6$, $d_\mathrm{v}=4$, and  $d_\mathrm{c}=7, 13, 17, 19, 23$, respectively. It was shown in \cite{dolecek2010analysis} 
that array-based codes with $d_\mathrm{v}=4$ and $d_\mathrm{c}>7$, do not contain any absorbing sets (ASs) in the $(4,4)$ class. This is matched with our results for $(4,4)$ EASs. 
Moreover, from Table \ref{tab:3hi}, one can see that the array-based code with $d_\mathrm{v}=4$ and $d_\mathrm{c}=7$, has $294$ EASs, but no FEAS in the $(4,4)$ class. 
Lemmas 6 and 7  in \cite{dolecek2010analysis} indicate that for the array-based codes with $d_\mathrm{v}=4$ and $d_\mathrm{c} >19$, there is no $(5,b)$ and $(6,2)$ absorbing sets. 
Based on Table V of \cite{mehdi2014}, for a variable-regular code with $d_\mathrm{v}=4$ and $g=6$, the $(5,0)$, $(5,2)$ and $(5,4)$ classes are the only classes 
with size 5 that can have elementary absorbing sets (EASs). The results of Table \ref{tab:3hi} thus reveal that $\mathcal{C}_{18}$ with $d_\mathrm{c}=19$ is the only array-based code with EASs of size 5,
and that the only class of size 5 for which this code has some EASs is the $(5,4)$ class.  
Also, our results in Table~\ref{tab:3hi}  complement those of \cite{dolecek2010analysis}, and demonstrate that the array-based codes with $d_\mathrm{v}=4$ and $d_\mathrm{c} > 7$ do not contain any EASs in the $(6,2)$ class. 
Since the {\em dot} expansion is the only expansion technique used in the dpl-based search algorithm for the array-based codes, this algorithm is the 
same as the dot-based search algorithm for these codes.

Table \ref{tab:4hi} lists the multiplicity of instances of LETSs, EASs and FEASs in all the nonempty classes within the range $a \leq 8$ and $b \leq 6$, for Codes $\mathcal{C}_{20}$ and $\mathcal{C}_{21}$. These are two high-rate codes with $d_\mathrm{v}=4$, $g=6$ and block lengths $n=4096$ and $n=16,383$, respectively. The search algorithm to find all the instances of LETSs in the range $a \leq 8$ and $b \leq 6$, is based on the information provided in Table \ref{tab:4,6c2}. 
\begin{table}[]
\centering
\caption{Multiplicities of LETS, EAS and FEAS structures of Codes $\mathcal{C}_{20}$ and $\mathcal{C}_{21}$ within the range $a \leq 8$ and $b \leq 6$}
\label{tab:4hi}
\begin{tabular}{||c| c| c||c|c|c|||c|c||c|c|c|| }
\cline{1-11}
&\multicolumn{5}{c|||}{$\mathcal{C}_{20}$} &\multicolumn{5}{c||}{$\mathcal{C}_{21}$}\\
\cline{1-11}
$(a,b)$&\multicolumn{2}{c||}{Primes}&Total &Total&Total&\multicolumn{2}{c||}{Primes}&Total &Total&Total\\
\cline{2-3}\cline{7-8}
class& $s_3$&$s_4$&LETS&EAS&FEAS& $s_3$&$s_4$&LETS&EAS&FEAS\\
\cline{1-11}
(3,6)&43531&-&43531&0&0&119233&-&119233&0&0\\
\cline{1-11}
(4,4)&15&-&15&15&13&10&-&10&10&6\\
\cline{1-11}
(4,6)&21383&-&21383&0&0&29001&-&29001&0&0\\
\cline{1-11}
(5,4)&9&-&9&7&4&5&-&5&1&1\\
\cline{1-11}
(5,6)&14622&-&14622&0&0&9796&-&9796&0&0\\
\cline{1-11}
(6,4)&10&-&10&7&5&2&-&2&2&0\\
\cline{1-11}
(6,6)&11553&1403&12956&913&246&3815&530&4345&307&161\\
\cline{1-11}
(7,4)&10&0&10&6&5&2&0&2&0&0\\
\cline{1-11}
(7,6)&10941&1837&12778&1779&476&1859&337&2196&288&167\\
\cline{1-11}
(8,4)&12&1&13&11&6&1&0&1&0&0\\
\cline{1-11}
(8,6)&11479&2659&14138&2612&741&970&188&1158&191&102\\
\cline{1-11}
$dpl$&\multicolumn{5}{c|||}{68 min.} &\multicolumn{5}{c||}{133 min.}\\
\cline{1-11}
$dot$&\multicolumn{5}{c|||}{-} &\multicolumn{5}{c||}{-}\\
\cline{1-11}
\end{tabular}
\end{table}
All the instances of LETS structures in this range for $\mathcal{C}_{20}$ and $\mathcal{C}_{21}$ are enumerated in about 68 and 133 minutes, respectively. 
For the dot-based search algorithm to find all the instances of LETSs in the range of interest in a guaranteed fashion, it requires to enumerate $s_6$ structures. 
This enumeration alone imposes huge computational complexity and memory requirements to the dot-based search algorithm (more than a day of runtime and 8-GB of memory).

Table \ref{tab:4c1} lists the multiplicity of instances of LETSs, EASs and FEASs in all the nonempty classes within the range $a \leq 10$ and $b \leq 10$, for $\mathcal{C}_{22}$ and $\mathcal{C}_{23}$. 
These are $(4,8)$-regular LDPC codes, with $g=6$, and with block lengths of $n=4000$ and $n=8000$, respectively. The search algorithm to find all the instances of LETSs in this range is based on the information provided in Table \ref{tab:4,6c3}.
\begin{table*}[]
\centering
\caption{Multiplicities of LETS, EAS and FEAS structures of Codes $\mathcal{C}_{22}$ and $\mathcal{C}_{23}$ within the range $a \leq 10$ and $b \leq 10$}
\setlength{\tabcolsep}{2pt}
\label{tab:4c1}
\begin{tabular}{||c|c|c|c||c|c|c|||c|c|c||c|c|c|| }
\cline{1-13}
&\multicolumn{6}{c|||}{$\mathcal{C}_{22}$} &\multicolumn{6}{c||}{$\mathcal{C}_{23}$}\\
\cline{1-13}
$(a,b)$&\multicolumn{3}{c||}{Primes}&Total &Total&Total&\multicolumn{3}{c||}{Primes}&Total &Total&Total\\
\cline{2-4}\cline{8-10}
class& $s_3$&$s_4$&$s_5$&LETS&EAS&FEAS& $s_3$&$s_4$&$s_5$&LETS&EAS&FEAS\\
\cline{1-13}
(3,6)&1563&-&-&1563&0&0&1620&-&-&1620&0&0\\
\cline{1-13}
(4,6)&91&-&-&91&0&0&55&-&-&55&0&0\\
\cline{1-13}
(4,8)&-&24269&-&24269&0&0&-&24107&-&24107&0&0\\
\cline{1-13}
(5,6)&2&-&-&2&0&0&2&-&-&2&0&0\\
\cline{1-13}
(5,8)&4080&560&-&4640&0&0&2011&292&-&2303&0&0\\
\cline{1-13}
(5,10)&-&-&402513&402513&0&0&-&-&406289&406289&0&0\\
\cline{1-13}
(6,8)&495&28&65&588&0&0&164&5&27&196&0&0\\
\cline{1-13}
(6,10)&111540&74004&-&185544&0&0&50073&37172&-&96525&0&0\\
\cline{1-13}
(7,8)&51&2&0&53&0&0&10&0&0&10&0&0\\
\cline{1-13}
(7,10)&33205&6589&-&39794&0&0&9009&1689&-&10698&0&0\\
\cline{1-13}
(8,8)&5&0&-&5&0&0&0&0&-&0&0&0\\
\cline{1-13}
(8,10)&6248&754&4&7006&0&0&1030&84&2&1116&0&0\\
\cline{1-13}
(9,8)&1&0&-&1&0&0&0&0&-&0&0&0\\
\cline{1-13}
(9,10)&1038&107&0&1145&0&0&96&5&0&101&0&0\\
\cline{1-13}
(10,10)&170&15&0&185&1&1&13&0&0&13&0&0\\
\cline{1-13}
$dpl$&\multicolumn{5}{c}{63 min.}& &\multicolumn{5}{c}{51 min.}&\\
\cline{1-13}
$dot$&\multicolumn{5}{c}{-}& &\multicolumn{5}{c}{-}&\\
\cline{1-13}
\end{tabular}
\end{table*}
It can be observed that $\mathcal{C}_{22}$ and $\mathcal{C}_{23}$ have one and no instance of EAS in this range, respectively. 
The dominant LETSs of  $\mathcal{C}_{23}$, which are included in those reported in Table \ref{tab:4c1}, are used in \cite{sina1} to 
accurately estimate the error floor under quantized iterative decoders. The accurate estimation reveals that error-prone 
structures of this code are not absorbing sets. If we limit the range of interest to $a \leq 9$ and $b \leq 8$, it  takes the proposed search 
algorithm only 4 minutes to find all the LETSs of $\mathcal{C}_{22}$ within this range. 
It is worth mentioning that to find all the instances of LETSs of $\mathcal{C}_{22}$ within the range $a \leq 9$ and $b \leq 8$, by the dot-based search algorithm, 
all the instances of $s_6$, $s_7$ and $s_8$, need to be enumerated. It takes more than 7 hours to only enumerate instances of $s_7$  in the Tanner graph of this code,
and enumerating instances of $s_8$ is not even feasible on the desktop computer. 

Table \ref{tab:56} lists the multiplicity of instances of LETSs, EASs and FEASs in all the nonempty classes within the range $a \leq 9$ and $b \leq 11$, for $\mathcal{C}_{24}$. 
This is a $(5,10)$-regular LDPC code with $g=6$, and block length $n=816$. 
\begin{table}[]
\centering
\caption{Multiplicities of LETS, EAS and FEAS Structures of Code $\mathcal{C}_{24}$ within the range $a \leq 9$ and $b \leq 11$}
\label{tab:56}
\begin{tabular}{||c|c|c||c|c|c|| }
\cline{1-6}
 &\multicolumn{5}{c||}{$\mathcal{C}_{24}$}\\
\cline{1-6}
$(a,b)$&\multicolumn{2}{c||}{Primes}&Total &Total&Total\\
\cline{2-3}
class& $s_3$&$s_4$&LETS&EAS&FEAS\\
\cline{1-6}
(3,9)&7876&-&7876&0&0\\
\cline{1-6}
(4,8)&30&-&30&30&30\\
\cline{1-6}
(4,10)&8446&-&8446&0&0\\
\cline{1-6}
(5,9)&98&-&98&36&34\\
\cline{1-6}
(5,11)&14873&-&14873&0&0\\
\cline{1-6}
(6,8)&1&-&1&1&1\\
\cline{1-6}
(6,10)&321&-&321&118&116\\
\cline{1-6}
(7,9)&5&-&5&1&1\\
\cline{1-6}
(7,11)&1132&9&1141&371&337\\
\cline{1-6}
(8,8)&1&0&1&1&1\\
\cline{1-6}
(8,10)&31&0&31&20&18\\
\cline{1-6}
(9,9)&4&0&4&1&1\\
\cline{1-6}
(9,11)&165&2&167&76&68\\
\cline{1-6}
$dpl$&\multicolumn{5}{c||}{6 min.} \\
\cline{1-6}
$dot$&\multicolumn{5}{c||}{-} \\
\cline{1-6}
\end{tabular}
\end{table}
All the instances of LETSs of this code in the range $a \leq 9$ and $b \leq 11$, are enumerated in about $6$ minutes. 
It can be seen that for this code, as a result of relatively large $d_\mathrm{v}$, no LETS with $b < a$ exists in the range of interest. 
Moreover, there is no FEAS for this code in the range $b < a$. It is reported in \cite{sina1} that using dominant LETSs of $\mathcal{C}_{24}$, included in Table \ref{tab:56}, 
one can obtain a more accurate estimation of the error floor in comparison to using only EASs. 

To the best of our knowledge, there are only two exhaustive search algorithms in the literature for finding the error-prone structures of variable-regular LDPC codes 
which cover a relatively wide range of trapping set classes, rates, block lengths  and variable degrees: the algorithm of \cite{kyung2012finding}, and that of \cite{mehdi2014}, \cite{yoones2015}. 
The work in \cite{kyung2012finding} is limited to only fully absorbing sets (FASs). Also, the regular LDPC codes which were studied in \cite{kyung2012finding} are all codes with  $d_\mathrm{v}=3$, short  block lengths $(n <2000)$ and rates equal to or less than $0.5$. The variable-regular LDPC codes studied in \cite{kyung2012finding} are Codes $\mathcal{C}_{7}$, $\mathcal{C}_{8}$, $\mathcal{C}_{11}$, $\mathcal{C}_{12}$, $\mathcal{C}_{13}$ and $\mathcal{C}_{14}$ in this paper. While, it takes the dpl-based search algorithm less than 42 minutes to find all the dominant 
LETSs and FEASs of each of these codes, the runtime to find FASs of these codes in \cite{kyung2012finding} (in smaller ranges of $a$ and $b$ values than those considered in this work) 
ranges from 5 to 48 hours using a desktop computer with an Intel core2 2.4-GHz CPU with 2-GB memory. Compared to the LSS- or dot-based algorithm of \cite{mehdi2014}, \cite{yoones2015}, 
as we discussed in the tables, the runtime and memory requirements of dpl-based search are up to three orders of magnitude smaller. 

To the best of our knowledge, the most versatile algorithm for finding LETSs of LDPC codes is that of \cite{mehdi2012}. 
This algorithm, however, is not exhaustive in the sense that it does not provide any guarantee that all the instances of $(a,b)$ LETS
structures with $a$ and $b$ values within a given range of interest are found. For example, Codes ${\cal C}_{11}$, ${\cal C}_{12}$, 
and ${\cal C}_{14}$ have also been studied in \cite{mehdi2012}, and their LETSs (EASs) are presented in Tables II, I and III 
of \cite{mehdi2012}, respectively. The comparison of the multiplicity of LETSs in different classes reported in this paper, and 
that reported in \cite{mehdi2012}, for each of these codes, reveals that for each code, there are some classes for which the multiplicity
reported in  \cite{mehdi2012} is smaller than that of this work. This means that for such classes, the algorithm of \cite{mehdi2012}
has not been able to find all the LETSs in that class. These classes for Codes ${\cal C}_{11}$, ${\cal C}_{12}$, 
and ${\cal C}_{14}$ are $\{(9,3),(11,3)\}, \{(8,4),(10,4),(11,3),(12,2),(12,4)\}$, and $\{(8,6),(10,6)\}$, respectively.

\section{Conclusion}
\label{conclude}
In this paper, we proposed a novel hierarchical graph-based expansion approach to characterize leafless elementary trapping sets (LETS) of variable-regular low-density parity-check (LDPC) codes.
The proposed characterization is based on three basic expansion techniques, dubbed, $dot$, $path$ and $lollipop$, used in the space of normal graphs. Each LETS structure ${\cal S}$ is characterized 
as a sequence of embedded LETS structures that starts from a simple cycle, and grows in each step by using one of the three expansions, until it reaches 
${\cal S}$. It was demonstrated that the new characterization, called $dpl$, is minimal, in that, none of the expansions in the sequence can be divided into smaller expansions
and still maintain the property that all the embedded sub-structures are LETSs.  Moreover, it was proved that any minimal characterization based on embedded LETS structures must only use one of the three 
expansions, {\em dot}, {\em path} and {\em lollipop}, at each step. The minimality of the characterization, allowed us to devise search algorithms that are provably efficient in
finding all the instances of $(a,b)$ LETS structures with $a \leq a_{max}$ and $b \leq b_{max}$, for any choice of $a_{max}$ and $b_{max}$, in a guaranteed fashion.
The efficiency of the search is a consequence of the fact that to enumerate the instances of larger LETS structures, the proposed
algorithm searches for instances of smaller LETS structures that are of interest themselves, i.e., their $a$ and $b$ values are in the range of interest.
This is unlike the so-called LSS-based search algorithm of \cite{mehdi2014}, where a large number of instances of LETS structures with large $b$ values need to be enumerated,
not because they are of direct interest, but because they happen to be parents of other LETS structures of interest.  As shown through extensive numerical results,
this reduces the search complexity of the dpl-based approach compared to LSS-based approach significantly.

To the best of our knowledge, the proposed dpl-based algorithm is the most efficient exhaustive search algorithm available for finding LETSs of variable-regular LDPC codes.
It is also the most general, in that, it is applicable to codes with any variable degree, girth, rate and block length.

\end{document}